\documentclass[journal]{IEEEtran}
\usepackage{graphicx}
\usepackage{amssymb}
\usepackage{amsthm}
\usepackage{stmaryrd}
\usepackage{mathtools}
\usepackage{tikz}
\def\checkmark{\tikz\fill[scale=0.4](0,.35) -- (.25,0) -- (1,.7) -- (.25,.15) -- cycle;} 

\usepackage{array}
\newtheorem{theorem}{Theorem}
\newtheorem{lemma}{Lemma}
\newtheorem{corollary}{Corollary}[theorem]
\newtheorem{definition}{Definition}
\usepackage{balance}
\usepackage{amsmath, amssymb}
\usepackage{comment}
\usepackage{hyperref}
\usepackage{amsmath}
\usepackage{mathtools}
\usepackage{mathrsfs} 
\usepackage{mathrsfs} 
\usepackage{bbm}
\usepackage{amsmath}               
  {
      \theoremstyle{plain}
      \newtheorem{assumption}{Assumption}
  }
\usepackage{mathtools}
\usepackage{bbm}
\theoremstyle{definition}
\newtheorem{example}{Example}
\DeclareMathOperator*{\argmin}{argmin}

\usepackage{color}
\usepackage[outdir=bibtex/figures/]{epstopdf}

\definecolor{bgrd}{rgb}{1,1,1}
\definecolor{grey}{rgb}{0.9,0.9,0.6}
\definecolor{gray}{rgb}{0.5,0.5,0.5}
\definecolor{dkr}{rgb}{0.6,0.2,0.2}
\definecolor{dkg}{rgb}{0,0.5,0}
\definecolor{dkb}{rgb}{0.0,0.1,0.7}
\definecolor{light-gray}{gray}{0.85}

\ifCLASSINFOpdf
\else
\fi
\begin{document}
%
\title{Non-Stochastic Information Theory}
%
%
%
\author{%
  \IEEEauthorblockN{Anshuka Rangi }
  \and
 \IEEEauthorblockN{Massimo Franceschetti}\thanks{The authors are with the Department of Electrical and Computer Engineering, University of California San Diego, email: arangi@ucsd.edu, massimo@ece.ucsd.edu. A subset of the results has been presented at ISIT 2019. This work was partially supported by NSF grant CCF 1917177}\\
}

\maketitle

\begin{abstract}
In an effort to develop the foundations for a non-stochastic theory of information,
 the notion of 
$\delta$-mutual information  between uncertain variables is introduced as a generalization of Nair's non-stochastic   information functional. Several properties of this new quantity are illustrated, and used to prove a channel coding theorem in a non-stochastic setting. Namely, it is shown that the largest $\delta$-mutual information between 
received and transmitted codewords over $\epsilon$-noise channels equals
the $(\epsilon, \delta)$-capacity. This notion of capacity generalizes
the Kolmogorov $\epsilon$-capacity to packing sets of overlap at most $\delta$, and is a variation of a previous definition proposed by one of the authors. Results are then extended to more general  noise models, and to non-stochastic, memoryless, stationary channels. Finally, 
   sufficient conditions are established for the factorization of the $\delta$-mutual information and to obtain a  single letter capacity expression. 
 Compared to previous non-stochastic approaches, 
the presented theory admits the possibility of decoding errors as in   Shannon's  probabilistic setting, while retaining a  worst-case,   non-stochastic character.
\end{abstract}
\section{Introduction}
This paper introduces elements of a non-stochastic information theory that parallels Shannon's  probabilistic theory of information,   but that provides strict deterministic guarantees for every codeword transmission.   
When Shannon laid the mathematical foundations of communication theory   he embraced a probabilistic approach \cite{shannon1948mathematical}. 
A tangible consequence of this choice is that in today's  communication systems  performance is   guaranteed in an average sense, or with high probability. Occasional violations from a specification are permitted,  and cannot be avoided. This approach is well suited for   consumer-oriented digital communication devices, where the occasional loss of  data packets is not critical, and made Shannon's theory the golden standard to describe communication limits, and to construct codes that achieve these limits.  
The probabilistic approach, however, has also prevented Shannon's theory to be relevant in systems where occasional decoding errors can result in catastrophic failures;  or  in adversarial settings,  where the   behavior of the channel   may be unknown and cannot be described by a probability distribution. 
The basic consideration that is the leitmotiv of this paper  is that the probabilistic framework is not a fundamental component of Shannon's theory,   and that the path laid  by Shannon's work can be  extended to embrace a non-stochastic setting. 

The idea of adopting a non-stochastic approach in information theory is not new. A few years after introducing the notion of capacity of a communication system~\cite{shannon1948mathematical}, Shannon   introduced the  zero-error capacity~\cite{shannon1956zero}. While the first notion corresponds to the largest   rate  of communication such that the probability of decoding error \emph{tends to zero}, the second corresponds to the largest   rate  of communication such that the probability of decoding error \emph{equals zero}. Both definitions of capacity satisfy  coding theorems: Shannon's channel coding theorem states that the capacity is the supremum of the mutual information between the input and the output of the channel~\cite{shannon1948mathematical}.   Nair introduced a non-stochastic   mutual information functional  and established an analogous coding theorem for the zero-error capacity in a non-stochastic setting~\cite{nair2013nonstochastic}.   While Shannon's theorem leads to a single letter expression, Nair's result  is multi-letter, involving the non-stochastic information between codeword blocks of $n$ symbols.  The zero-error capacity can also be formulated as a graph-theoretic property  and the absence of a single-letter expression for general graphs is well known~\cite{shannon1956zero,rosenfeld1967problem}.  Extensions of Nair's nonstochastic approach to characterize   the zero-error capacity in the presence of feedback from the receiver to the transmitter using  nonstochastic directed mutual information  have also been considered~\cite{nair2012nonstochastic}. 

A parallel non-stochastic approach is due to Kolmogorov who,
motivated by Shannon's results,  introduced the   notions of $\epsilon$-entropy and $\epsilon$-capacity in the context of functional spaces~\cite{kolmogorov}. He defined the $\epsilon$-entropy  as the logarithm base two of the \emph{covering number} of the space, namely the logarithm of the minimum number of balls of radius $\epsilon$ that can cover the space. Determining this number is analogous to designing a \emph{source codebook} such that the distance between any signal in the space and a codeword is at most $\epsilon$. In this way, any transmitted signal can be represented by a codeword point with at most $\epsilon$-distortion.  Notions related to the $\epsilon$-entropy  are the Hartley entropy~\cite{Hartley}   and the R\'{e}nyi differential (0th-order) entropy~\cite{Renyi}. They  arise  for random variables with  known range but unknown distribution,  and are defined by taking  the logarithm of the cardinality (for discrete variables), or Lebesgue measure (for continuous variables)  of their range. Thus,  their definition does not require any statistical structure.
Using  these entropies, non-stochastic measures of mutual information have been constructed~\cite{ota,klir}.
Unfortunately,  the absence of coding theorems    makes the  operational significance   of these definitions lacking.

Rather than using mutual information and entropy, Kolmogorov  gave   an operational definition of the $\epsilon$-capacity     as the logarithm base two of the \emph{packing number} of the space, namely the logarithm of the maximum number of balls of radius $\epsilon$ that can be placed in the space without overlap. Determining this number is analogous to designing  a \emph{channel codebook}  such that the distance between any two codewords is at least $2\epsilon$. In this way, any transmitted codeword  that is subject to a perturbation   of at most $\epsilon$   can  be recovered at the receiver without error. It follows that the $\epsilon$-capacity   corresponds to the zero-error capacity of an additive  channel having arbitrary,   bounded noise of support at most $[0,\epsilon]$. 
Lim and  Franceschetti extended this concept introducing the   $(\epsilon,\delta)$ capacity~\cite{lim2017information}, defined as  the logarithm base two of the largest number of balls of radius $\epsilon$ that can be placed  in the space with  average codeword  overlap of at most $\delta$. In this setting, $\delta$   measures   the amount of error that can be tolerated when designing a codebook in a non-stochastic setting. Neither the Kolmogorov capacity, nor its $(\epsilon,\delta)$ generalization 
have a corresponding information-theoretic characterization in terms of mutual information and an associated coding theorem. This is offered in the present paper.
Some possible applications of  non-stochastic approaches arising in the context of  estimation,   control, security, communication over non-linear optical channels, and robustness of neural networks are described in~\cite{saberi2018estimation,saberi2019state,wiese2016uncertain,borujeny2020signal,weng2018evaluating,verma2019error,ferng2011multi}; and some   are also discussed in the context of the presented theory at the end of the paper.  

The rest of the paper is organized as follows. Section II provides a summary of our contributions; Section III introduces the mathematical framework of non-stochastic uncertain variables that is used throughout the paper.   Section IV introduces the concept of non-stochastic mutual information. Section V gives an operational definition of capacity of a communication channel and relates it to the mutual information.  Section VI extends results to more general channel models; and section VII concentrates on the special case of stationary, memoryless, uncertain channels. Sufficient conditions are obtained to obtain single-letter expressions for this case. Section VIII considers some examples of channels   and computes the corresponding capacity. Finally, Section IX discusses some possible application of the developed theory, and Section X draws conclusions and discusses future directions. A  subset of the results has been presented in \cite{rangi2019towards}. 
\section{Contributions}
 
 We   introduce a   notion of $\delta$-mutual information between non-stochastic, uncertain variables. In contrast to Nair's definition ~\cite{nair2013nonstochastic}, which only allows to measure information with \emph{full confidence},
 our definition considers the information revealed by one   variable regarding the other \emph{with  a given level of confidence}. We  then introduce a  notion of $(\epsilon,\delta)$-capacity,   defined as the logarithm base two of the largest number of balls of radius $\epsilon$ that can be placed in the   space such that the overlap between  any two balls  is    at most a ratio  of $\delta$ and the total number of balls. In contrast to the definition of Lim and Franceschetti~\cite{lim2017information}, which requires  the average overlap among all the balls  to be at most $\delta$, our definition requires to bound the overlap between any  pair of balls.   
 For $\delta=0$, our capacity definition reduces to the Kolmogorov $\epsilon$-capacity, or equivalently to the zero-error capacity of an additive, bounded noise channel, and our mutual information  definition reduces to Nair's one~\cite{nair2013nonstochastic}.
We establish  a  channel coding theorem    in this non-stochastic setting, showing that the largest mutual information, with confidence at least $(1-\delta)$,  
 between a transmitted codeword and its received version  corrupted  with noise   at most $\epsilon$,    is the $(\epsilon,\delta)$-capacity.    We then extend this result to more general non-stochastic channels, where the  noise   is expressed in terms of a set-valued map $N(\cdot)$ associating each transmitted codeword to a noise region in the received codeword space, that is not necessarily a ball of radius $\epsilon$.

 Next, we  consider the class of non-stochastic, memoryless,  stationary uncertain channels. In this case,  the noise $N(\cdot)$ experienced by a codeword  of $n$ symbols factorizes into $n$ identical terms describing the noise experienced by each codeword symbol. This  is the non-stochastic analogous of a  discrete memoryless channel (DMC), where  the  current  output symbol depends only on the current input symbol and not on any of the previous input symbols, and the noise distribution is constant across symbol transmissions. It differs from Kolmogorov's $\epsilon$-noise channel, where the noise experienced by one symbol affects the noise experienced by  other symbols.  In  Kolmogorov's setting,   the noise occurs within a ball of radius $\epsilon$. It follows that for any realization where the noise along one dimension (\emph{viz.} symbol)    is close to $\epsilon$, the noise  experienced by all other symbols lying in the remaining dimensions   must be close to zero. In contrast, for non-stochastic, memoryless,  stationary channels, the noise experienced by any transmitted symbol is   described by a single, non-stochastic set-value map from the transmitted alphabet to the received symbol space. We provide   coding theorems in this setting in terms of the $\delta$-mutual information rate  between received and transmitted codewords. Finally, we provide sufficient conditions for the factorization of the mutual information and to obtain a  single-letter expression for the  non-stochastic capacity of stationary, memoryless, uncertain channels.  We provide   examples in which these conditions are satisfied and compute the corresponding capacity, and we conclude with a discussion of some possible applications of the presented theory.

\section{Uncertain variables}\label{sec:Background}
We start by  reviewing  the mathematical framework used in~\cite{nair2013nonstochastic} to describe non-stochastic uncertain variables (UVs). An UV $X$ is a mapping from a  sample space $\Omega$ to a set $\mathscr{X}$, i.e. for all $\omega\in \Omega$,  we have $x=X(\omega) \in \mathscr{X}$. Given an UV $X$, the marginal range of $X$ is 
\begin{equation}
    \llbracket X\rrbracket=\{X(\omega):\omega\in \Omega\}.
\end{equation}
The joint range of two UVs $X$ and $Y$ is
\begin{equation}
    \llbracket X, Y\rrbracket=\{(X(\omega),Y(\omega)):\omega\in \Omega\}.
\end{equation}
Given an UV $Y$, the conditional range of $X$ given $Y=y$ is 
\begin{equation}
    \llbracket X|y\rrbracket=\{X(\omega):Y(\omega)=y,\omega\in \Omega\},
\end{equation}
and the conditional range of $X$ given $Y$ is
\begin{equation}
    \llbracket X|Y\rrbracket=\{\llbracket X|y\rrbracket: y\in \llbracket Y\rrbracket\}.
\end{equation}
Thus, $\llbracket X|Y\rrbracket$ denotes the uncertainty in $X$ given the realization  of $Y$ and  $\llbracket X, Y\rrbracket$ represents the total joint uncertainty of  $X$ and $Y$, namely 
\begin{equation}
    \llbracket X, Y\rrbracket= \cup_{y\in \llbracket Y\rrbracket}\llbracket X|y\rrbracket\times\{y\}.
\end{equation}
Finally, two UVs $X$ and $Y$ are independent if for all $x\in\llbracket X\rrbracket$
\begin{equation}
    \llbracket Y|x\rrbracket =\llbracket Y\rrbracket,
\end{equation}
which also implies that for all $y \in \llbracket Y\rrbracket$
\begin{equation}
    \llbracket X|y \rrbracket =\llbracket X\rrbracket.
\end{equation}

\section{$\delta$-Mutual information}\label{section:MaxMinInfo}
\subsection{Uncertainty function }\label{sec:ErrorRegion}
We now introduce a   class of functions that are   used  to express the amount of uncertainty in determining one UV given another.
In our setting, an uncertainty function associates a positive number to a given set, which expresses the ``massiveness'' or ``size'' of that set. 
\begin{definition}
Given any set $\mathscr{X}$,  
$m_{\mathscr{X}}: \mathscr{S}\subseteq \mathscr{X}\to \mathbb{R}_{0}^{+}$ is an uncertainty  function if it is  finite and strongly transitive, namely:

 For all $\mathscr{S}\subseteq \mathscr{X}, \mathscr{S} \not = \emptyset,$ we have 
\begin{align} \label{ass:finiteness}
 0 < m_{\mathscr{X}}(\mathscr{S})<\infty,  \; 
 m_{\mathscr{X}}(\emptyset)=0.
\end{align}

For all $ \mathscr{S}_1,\mathscr{S}_2 \subseteq \mathscr{X},$   we have  
\begin{equation} \label{eq:strongtransitivity}
\max\{m_{\mathscr{X}}(\mathscr{S}_1), m_{\mathscr{X}}(\mathscr{S}_2) \} \leq m_{\mathscr{X}}(\mathscr{S}_1\cup \mathscr{S}_2).
\end{equation}
\end{definition}
In the case  $\mathscr{X}$ is  measurable, an   uncertainty  function can easily be constructed using a measure. In the case  $\mathscr{X}$ is a bounded (not necessarily measurable) metric space and the input set $\mathscr{S}$ contains at least two points, an example of uncertainty function is the diameter.

\subsection{Association and dissociation between UVs}
We now introduce   notions of association and dissociation between UVs. In the following definitions, we let $m_{\mathscr{X}}(.)$ and $m_{\mathscr{Y}}(.)$ be  uncertainity functions  defined over  sets $\mathscr{X}$ and $\mathscr{Y}$ corresponding to UVs $X$ and $Y$. We use the notation $\mathscr{A} \succ\delta$ to indicate that for all $a \in \mathscr{A}$ we have   $a>\delta$.  Similarly, we use $\mathscr{A} \preceq \delta$ to indicate that for all $a \in \mathscr{A}$ we have $a \leq \delta$.  For $\mathscr{A}=\emptyset$, we assume $\mathscr{A} \preceq \delta$ is always satisfied, while   $\mathscr{A} \succ \delta$  is not. 
Whenever we consider  $i \not = j$, we also assume that  $y_i \not= y_j$ and $x_i \not = x_j$.  
\begin{definition}\label{def:AssSets}
The  sets of association for UVs $X$ and $Y$ are
\begin{align}\label{eqdef0}
 \mathscr{A}(X;Y) &= \bigg\{  \frac{m_{\mathscr{X}}(\llbracket X|y_{1}\rrbracket\cap \llbracket X|y_{2}\rrbracket|)}{m_{\mathscr{X}}(\llbracket X\rrbracket)}:    y_1,y_2\in\llbracket Y\rrbracket \bigg\} \setminus \big\{0\big\} ,
\end{align}
\begin{align} \label{eqdef1}
  \mathscr{A}(Y;X) &= \bigg\{ \frac{m_{\mathscr{Y}}(\llbracket Y|x_{1}\rrbracket\cap \llbracket Y|x_{2}\rrbracket|)}{m_{\mathscr{Y}}(\llbracket Y\rrbracket)}:     \;  x_{1},x_{2}\in \llbracket X \rrbracket  \bigg\} \setminus\big\{0\big\}.
\end{align}
\end{definition}
\begin{definition}\label{defn:association}
For any $\delta_{1},\delta_{2} \in [0,1)$,
 UVs $X$ and $Y$ are disassociated at levels $(\delta_{1},\delta_{2})$ if the following inequalities hold:
\begin{equation}\label{eq:AssocX}
 \mathscr{A}(X;Y)\succ\delta_{1},
\end{equation}
\begin{equation}\label{eq:AssocY}
 \mathscr{A}(Y;X)\succ \delta_{2},
\end{equation}
and  this case  we write $(X, Y) \stackrel{d}{ \leftrightarrow } (\delta_{1},\delta_{2})$.
\end{definition}
Having UVs $X$ and $Y$ be disassociated at levels  $(\delta_1,\delta_2)$    indicates that at least two conditional ranges $\llbracket X|y_1\rrbracket$ and $\llbracket X|y_2\rrbracket$  have nonzero overlap, and that given any two conditional ranges, either they do not overlap  or the uncertainty associated to their overlap is greater than a $\delta_1$ fraction of the total uncertainty associated to $\llbracket X \rrbracket$; and that  the same   holds for conditional ranges  $\llbracket Y|x_1\rrbracket$ and $\llbracket Y|x_2\rrbracket$ and level $\delta_2$.
The levels of disassociation    can be viewed as     lower bounds on the amount of  residual uncertainty in each  variable when the other is known. If $X$ and $Y$ are independent, then  all the conditional ranges completely overlap, $ \mathscr{A}(X;Y)$ and $ \mathscr{A}(Y;X)$ contain only the element one, and the variables  are  maximally disassociated (see Figure \ref{fig:combined}a). 
\begin{figure*}
\begin{center}
  \includegraphics[width=.7 \textwidth]{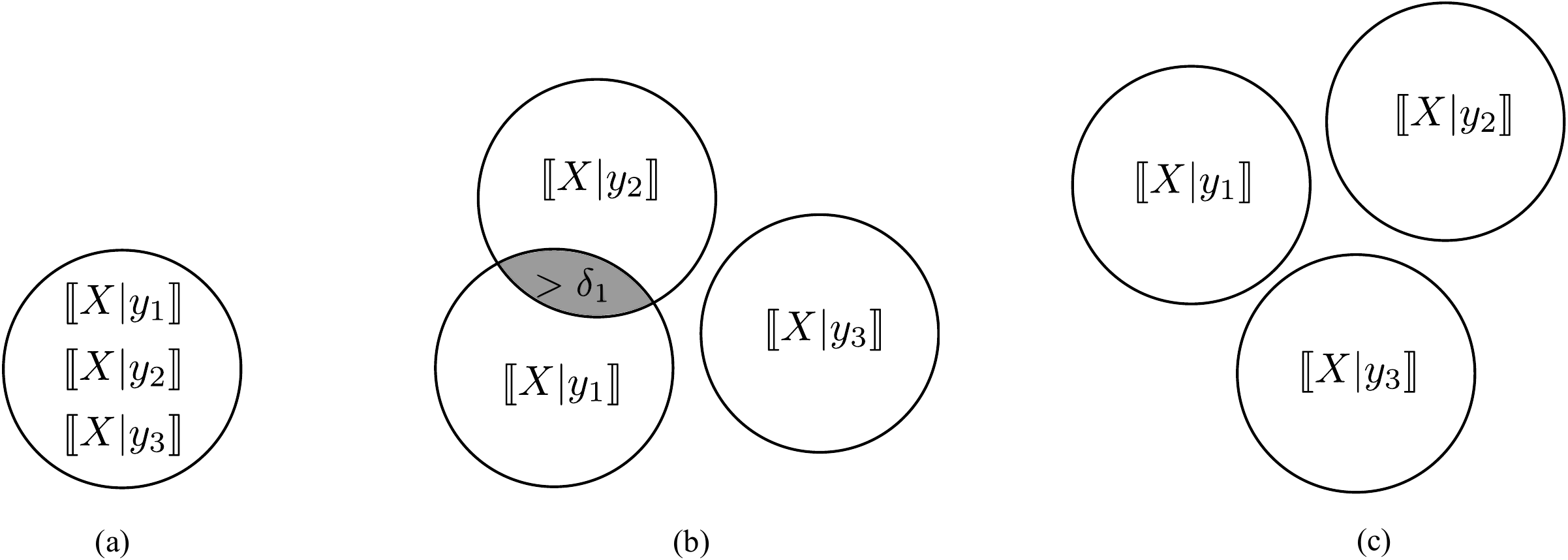}
  \end{center}
  \caption{ Illustration of disassociation between UVs. Case (a):  variables are maximally disassociated and all conditional ranges completely overlap. Case (b): variables are disassociated at some levels  $(\delta_1,\delta_2)$, and there is some overlap between at least two conditional ranges. Case (c): variables are not disassociated at any  levels, and there is no overlap between the conditional ranges.
 }
\label{fig:combined}
\end{figure*}
In this case,  
knowledge of $Y$ does not reduce the uncertainty of $X$, and vice versa. On the other hand, when the uncertainty associated to any of the non-zero intersections of the  conditional ranges  decreases, but remains positive, then   $X$ and $Y$  become less disassociated,  in the sense that knowledge of $Y$ can reduce the   residual uncertainty of $X$, and vice versa  (see Figure \ref{fig:combined}b).  When the intersection between every pair of conditional ranges becomes  empty, the variables cease being disassociated (see Figure \ref{fig:combined}c).

 An analogous definition of association is given to provide upper bounds on the residual uncertainty of one uncertain variable when the other is known.
 \begin{definition}\label{defn:dissassociation}
For any $\delta_{1},\delta_{2} \in [0,1]$,
 we say that UVs $X$ and $Y$ are  associated at levels $(\delta_{1},\delta_{2})$ if the following inequalities hold:
\begin{equation}\label{eq:DissassX}
 \mathscr{A}(X;Y)\preceq\delta_{1},
\end{equation}
\begin{equation}\label{eq:DissassY}
 \mathscr{A}(Y;X)\preceq \delta_{2},
\end{equation}
and in this case we write $(X, Y) \stackrel{a}{ \leftrightarrow } (\delta_{1},\delta_{2})$.
\end{definition}

The following lemma provides   necessary and sufficient conditions for  association at  given levels to hold. 
 These conditions   are stated for all points in the  marginal ranges $\llbracket Y\rrbracket$ and $\llbracket X\rrbracket$. They show that in the case of association one can also include in the definition the conditional ranges that have zero intersection. 
This is not the case for disassociation.
\begin{lemma}\label{lemma:Dissassociation}
  For any $\delta_{1},\delta_{2} \in [0,1]$, $(X, Y) \stackrel{a}{ \leftrightarrow } (\delta_{1},\delta_{2})$  if and only if for all $y_{1},y_{2}\in \llbracket Y\rrbracket$, we have
\begin{equation}\label{eq:DissX1}
\frac{m_{\mathscr{X}}(\llbracket X|y_{1}\rrbracket\cap \llbracket X|y_{2}\rrbracket|)}{m_{\mathscr{X}}(\llbracket X\rrbracket)}\leq \delta_{1},
\end{equation}
and for all  $x_{1},x_{2}\in \llbracket X\rrbracket$, we have
\begin{equation}\label{eq:DissX}
\frac{m_{\mathscr{Y}}(\llbracket Y|x_{1}\rrbracket\cap \llbracket Y|x_{2}\rrbracket|)}{m_{\mathscr{Y}}(\llbracket Y\rrbracket)}\leq \delta_{2}.
\end{equation}
\end{lemma}
\begin{proof} The proof is given in   Appendix \ref{sec:Lemma1}.
\end{proof}



An immediate, yet important  consequence of our definitions  is that   both association and disassociation at given levels $(\delta_1, \delta_2)$ cannot hold simultaneously.  We also have that, given any two UVs, one can always choose $\delta_1$ and $\delta_2$ to be large enough  such that they are associated at levels $(\delta_1,\delta_2)$. In contrast,  as the smallest value in the sets $\mathcal{A}(X;Y)$ and $\mathcal{A}(Y;X)$ tends to zero, 
the variables eventually cease being disassociated. Finally, it is possible that two  uncertain variables are neither associated nor disassociated at given levels   $(\delta_1, \delta_2)$.

\begin{example}
 Consider three individuals $a$, $b$ and $c$ going for a walk along a path. Assume they  take at most $15$, $20$ and $10$ minutes to finish their walk, respectively. 
Assume $a$ starts walking at time 5:00, $b$ starts  walking at 5:10 and $c$ starts walking at 5:20. Figure \ref{fig:example1} 
\begin{figure}[t]
\begin{center}
  \includegraphics[width=.70 \columnwidth]{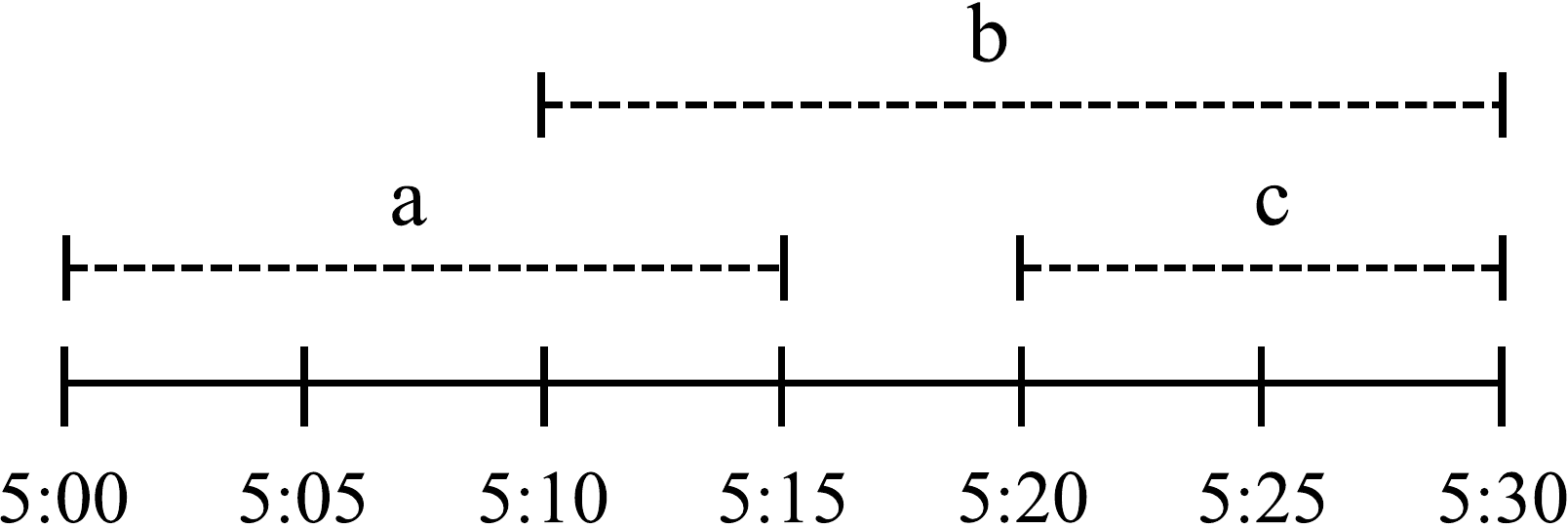}
  \end{center}
  \caption{ Illustration of the possible time intervals for the walkers on the path.
 }
\label{fig:example1}
\end{figure}
shows the possible time intervals for the walkers on the path. Let an uncertain variable $W$ represent  the set of walkers that are present on the path at any time, and  an uncertain variable $T$ represent  the time at which any walker on the path finishes its walk.
Then, we have the marginal ranges
\begin{equation}
    \llbracket W\rrbracket=\{\{a\},\{b\},\{c\},\{a,b\},\{b,c\}\}, 
\end{equation}
\begin{equation}
    \llbracket T\rrbracket=[\mbox{5:00},\mbox{5:30}].
\end{equation}
We also have the conditional ranges
\begin{equation}
    \llbracket T|\{a\}\rrbracket=[\mbox{5:00},\mbox{5:15}],
\end{equation}
\begin{equation}
    \llbracket T|\{b\}\rrbracket=[\mbox{5:10},\mbox{5:30}],
\end{equation}
\begin{equation}
    \llbracket T|\{c\}\rrbracket=[\mbox{5:20},\mbox{5:30}],
\end{equation}
\begin{equation}
    \llbracket T|\{a,b\}\rrbracket=[\mbox{5:10},\mbox{5:15}],
\end{equation}
\begin{equation}
    \llbracket T|\{b,c\}\rrbracket=[\mbox{5:20},\mbox{5:30}].
\end{equation}
For all $t\in [\mbox{5:00},\mbox{5:10})$, we   have
 \begin{equation}
     \llbracket W|t\rrbracket=\{\{a\}\},
 \end{equation}
 for all $t\in [\mbox{5:10},\mbox{5:15}]$, we have
 \begin{equation}
     \llbracket W|t\rrbracket=\{\{a,b\},\{a\},\{b\}\},
 \end{equation}
  for all $t\in (\mbox{5:15},\mbox{5:20})$, we have
 \begin{equation}
     \llbracket W|t\rrbracket=\{\{b\}\},
 \end{equation}
 and for all $t\in [\mbox{5:20},\mbox{5:30}]$, we have
 \begin{equation}
     \llbracket W|t\rrbracket=\{\{b,c\},\{b\},\{c\}\}.
 \end{equation}
Now, let the uncertainty function   of a time set $\mathscr{S}$ be
\begin{equation}
    m_{\mathscr{T}}(\mathscr{S})=\begin{cases} \mathcal{L}(\mathscr{S})+10 \mbox{ if }\mathscr{S}\neq \emptyset,\\
    0 \mbox{ otherwise },
    \end{cases}
\end{equation}
where  $\mathcal{L}(\cdot)$ is the Lebesgue measure. 
Let the uncertainty function $m_{\mathscr{W}}(.)$ associated to a set of individuals be the cardinality of the set. 
Then,  the sets of association are
\begin{equation}
    \mathscr{A}(W;T)=\{1/5,3/5\},
\end{equation}
\begin{equation}
    \mathscr{A}(T;W)=\{3/8,1/2\}.
\end{equation}
It follows that that for all $\delta_1<1/5$ and $\delta_2<3/8$, we have
\begin{equation}\label{eq:ref122}
    (W,T)\stackrel{d}{\leftrightarrow}(\delta_1,\delta_2),
\end{equation}
and the residual uncertainty in $W$ given $T$  is at least $\delta_1$ fraction of the total uncertainty in $W$, while the residual uncertainty in $T$ given $W$ is at least  $\delta_2$ fraction of the total uncertainty in $T$.
On the other hand, for all $\delta_1\geq 3/5$ and $\delta_2\geq 1/2$  we have 
\begin{equation}\label{eq:ref123}
    (W,T)\stackrel{a}{\leftrightarrow}(\delta_1,\delta_2),
\end{equation}
and the residual uncertainty in $W$ given $T$  is at most $\delta_1$ fraction of the total uncertainty in $W$, while the residual uncertainty   in $T$ given $W$ is at most $\delta_2$ fraction of the total uncertainty in $T$.

Finally, if  $1/5\leq \delta_1<3/5$ or $3/8\leq \delta_2<1/2$, then $W$ and $T$ are neither associated nor disassociated. 

\end{example} 

\subsection{$\delta$-mutual information}
We now introduce the mutual information between uncertain variables in terms of some structural properties of  covering  sets. Intuitively, for any $\delta \in [0,1]$ the $\delta$-mutual information, expressed in bits, represents the most refined knowledge  that one uncertain variable provides about the other, at a given level of confidence $(1-\delta)$.  We express this idea   by considering the   quantization of the range of uncertainty of one variable, induced by the knowledge of  the other. Such quantization ensures that the variable can be identified  with uncertainty at most $\delta$.
The notions of association and disassociation introduced above are used to ensure that the mutual information is well defined, namely it can be positive, and enjoys a certain symmetric property.   
\begin{definition}\label{defn:overlapCon}{$\delta$-Connectedness and $\delta$-isolation.} 
\begin{itemize}
\item For any $\delta \in [0,1]$, points $x_{1}, x_{2} \in \llbracket X\rrbracket$ are $\delta$-connected via $\llbracket X|Y\rrbracket$, and are denoted by $x_{1}\stackrel{\delta}{\leftrightsquigarrow} x_{2}$, if there exists a finite  sequence $\{\llbracket X|y_{i}\rrbracket\}_{i=1}^{N}$ of conditional sets such that $x_{1}\in \llbracket X|y_{1}\rrbracket$, $x_{2}\in \llbracket X|y_{N}\rrbracket$ and 
for all $1<i\leq N$, we have 
\begin{equation}\label{eq:condtionConnectedness}
    \frac{m_{\mathscr{X}}(\llbracket X|y_{i}\rrbracket\cap \llbracket X|y_{i-1}\rrbracket)}{m_{\mathscr{X}}(\llbracket X\rrbracket)}> \delta.
\end{equation}
If $x_{1}\stackrel{\delta}{\leftrightsquigarrow} x_{2}$ and  $N=1$, then we say that $x_{1}$ and $x_{2}$ are  singly $\delta$-connected via $\llbracket X|Y\rrbracket$, i.e. there exists a $y$ such that $x_{1}, x_{2}\in \llbracket X|y\rrbracket$.
\item A set $\mathscr{S}\subseteq\llbracket X\rrbracket$ is (singly) $\delta$-connected via $\llbracket X|Y\rrbracket$ if every pair of points in the set is (singly) $\delta$-connected via $\llbracket X|Y\rrbracket$. 
\item Two sets $\mathscr{S}_1,\mathscr{S}_2\subseteq\llbracket X\rrbracket$ are $\delta$-isolated via $\llbracket X|Y\rrbracket$ if no point in $\mathscr{S}_1$ is  $\delta$-connected to any point in $\mathscr{S}_2$.
\end{itemize}
\end{definition}

\begin{definition}{$\delta$-overlap family.}  \label{defoverlap}\\
 For any $\delta \in [0,1]$, a  $\llbracket X|Y \rrbracket$ 
$\delta$-overlap family of $\llbracket X\rrbracket$, denoted by $\llbracket X|Y\rrbracket^*_{\delta}$, is a largest family of distinct sets covering $\llbracket X\rrbracket$ such that:
\begin{enumerate}
\item  Each set in the family is $\delta$-connected and contains at least one singly $\delta$-connected set  of  the form $\llbracket X|y\rrbracket$.
\item The measure of overlap between any two distinct sets in the family is at most $\delta m_{\mathscr{X}}(\llbracket X\rrbracket)$.
\item For every singly $\delta$-connected set,  
there exist a set in the family containing it.
 \end{enumerate}
 \end{definition}

The first   property of the $\delta$-overlap family ensures  that points in the same set of the family  \emph{cannot} be distinguished with confidence at least $(1-\delta)$, while also ensuring that each set cannot be arbitrarily small. The second and third properties ensure that points  that are not covered by the same set of the family \emph{can} be distinguished with confidence at least $(1-\delta)$.  
It follows that the cardinality of the covering family represents the most refined knowledge at a given level of confidence $(1-\delta)$ that we can have about $X$, given the knowledge of $Y$. This also corresponds to  the most refined quantization of the set $\llbracket X\rrbracket$  induced by   $Y$. This interpretation is analogous to the one in~\cite{nair2013nonstochastic}, extending the concept of overlap partition introduced there to a $\delta$-overlap family in this work. The stage is now set to introduce the $\delta$-mutual information in terms of the $\delta$-overlap family. 
\begin{definition}\label{defn:information}
The $\delta$-mutual information  provided by  $Y$  about   $X$   is  
\begin{equation}
    I_{\delta}(X ; Y)=\log_2 |\llbracket X|Y\rrbracket^*_{\delta}| \mbox{  bits},
\end{equation}
if a $\llbracket X|Y \rrbracket$ $\delta$-overlap family of $\llbracket X \rrbracket$ exists,
otherwise it is zero.
\end{definition}

We now show that when variables are associated at level $(\delta,\delta_2)$, then there exists a $\delta$-overlap family, so that the mutual information is well defined.
\begin{theorem}\label{thm:associativityPArtition} 
If $( X, Y)\stackrel{a}{\leftrightarrow}(\delta,\delta_{2})$, then 
there exists a  $\delta$-overlap family  $\llbracket X|Y\rrbracket^*_{\delta}$. 
\end{theorem}
\begin{proof}We  show that 
\begin{equation}
    \llbracket X|Y\rrbracket=\{\llbracket X|y\rrbracket:y\in \llbracket Y\rrbracket\}
\end{equation}
is a  $\delta$-overlap family.  First, note that $\llbracket X|Y\rrbracket$ is a cover of $\llbracket X\rrbracket$, since 
$\llbracket X\rrbracket=\cup_{y\in\llbracket Y\rrbracket}\llbracket X|y\rrbracket$. Second, each set in the family $\llbracket X|Y\rrbracket$  is singly $
\delta$-connected via $\llbracket X|Y\rrbracket$, since trivially any two points $x_1,x_2 \in \llbracket X|y \rrbracket$ are singly $\delta$-connected via the same set.
It follows that Property 1 of Definition~\ref{defoverlap} holds. 

Now, since  $(X, Y) \stackrel{a}{ \leftrightarrow } (\delta,\delta_{2})$,  then by Lemma \ref{lemma:Dissassociation} for all $y_1,y_2 \in \llbracket Y\rrbracket$  we have
\begin{equation}
    \frac{m_{\mathscr{X}}(\llbracket X|y_1\rrbracket\cap\llbracket X|y_2\rrbracket)}{m_{\mathscr{X}}(\llbracket X\rrbracket)}\leq \delta,
\end{equation}
which shows that Property 2 of Definition~\ref{defoverlap} holds.
Finally, it is also easy to see that Property 3 of Definition~\ref{defoverlap} holds, since $\llbracket X|Y\rrbracket$  contains all   sets $\llbracket X|y\rrbracket$.
\end{proof}

Next, we show that   a $\delta$-overlap family also exists when variables are disassociated at level $(\delta, \delta_2)$.  In this case, we also characterize the mutual information in terms of a partition of $\llbracket X\rrbracket$.
\begin{definition}{$\delta$-isolated partition.} \\ 
A   $\llbracket X|Y\rrbracket$  $\delta$-isolated partition of $\llbracket X\rrbracket$, denoted by $\llbracket X|Y\rrbracket_{\delta}$,  is a partition of  $\llbracket X\rrbracket$ such that any two sets in the partition are $\delta$-isolated via  $\llbracket X|Y\rrbracket$. 
\end{definition}

\begin{theorem}\label{lemma:overlap} If $( X, Y)\stackrel{d}{\leftrightarrow}(\delta,\delta_{2})$, then we have:
\begin{enumerate}
\item There exists a unique $\delta$-overlap family  $\llbracket X|Y\rrbracket^*_{\delta}$.
\item The $\delta$-overlap family is the $\delta$-isolated partition of largest cardinality, namely for any 
$\llbracket X|Y\rrbracket_{\delta}$, we have
\begin{equation}\label{eq:1.2}
    |\llbracket X|Y\rrbracket_{\delta}|\leq |\llbracket X|Y\rrbracket^*_{\delta}|,
\end{equation}
where the  equality holds if and only if $\llbracket X|Y\rrbracket_{\delta}= \llbracket X|Y\rrbracket^*_{\delta}$.
\end{enumerate}
\end{theorem}
\begin{proof}
First, we show the existence of a ${\delta}$-overlap family. For all $x\in \llbracket X \rrbracket$, let $\mathscr{C}(x)$ be the set of points that are $\delta$-connected to $x$ via $\llbracket X|Y \rrbracket$, namely
\begin{equation} \label{eq:C1}
\mathscr{C}(x)=\{x_{1}\in\llbracket X\rrbracket: x\stackrel{\delta}{\leftrightsquigarrow}x_{1}\}.
\end{equation}
Then, we let
\begin{equation} \label{eq:C}
\mathcal{C}=\{\mathscr{C}(x): x\in \llbracket X \rrbracket \}, 
\end{equation}
and show that this  is a $\delta$-overlap family. 
First, 
note that  since $\llbracket X\rrbracket=\cup_{\mathscr{S}\in\mathcal{C}}\mathscr{S}$, we have that $\mathcal{C}$
is a cover of $ \llbracket X \rrbracket$. Second,
for all   $\mathscr{C}(x) \in \mathcal{C}$ there exists a $y \in \llbracket Y \rrbracket$ such that $x \in \llbracket X | y \rrbracket$, and since
 any two points $x_1,x_2\in \llbracket X|y\rrbracket$ are  singly $\delta$-connected via $\llbracket X|Y \rrbracket$, we have that    $\llbracket X | y \rrbracket \subseteq \mathscr{C}(x)$. It follows that every set in the family $\mathcal{C}$ contains at least one singly $\delta$-connected set.
For all $x_1,x_2\in \mathscr{C}(x)$, we also have  $x_1\stackrel{\delta}{\leftrightsquigarrow}x$ and $x\stackrel{\delta}{\leftrightsquigarrow}x_2$. Since   $( X, Y)\stackrel{d}{\leftrightarrow}(\delta,\delta_{2})$,  by  Lemma \ref{lemma:commutativeProperty} in   Appendix \ref{sec:AuxResult} this implies $x_1\stackrel{\delta}{\leftrightsquigarrow}x_2$. It follows that  every set in the family $\mathcal{C}$ is $\delta$-connected and contains at least one singly $\delta$-connected set, and we conclude that Property~1 of Definition \ref{defoverlap} is satisfied. 

We now claim that for all  $x_{1},x_{2}\in \llbracket X \rrbracket$, if 
\begin{equation}
    \mathscr{C}(x_{1})\neq \mathscr{C}(x_{2}),
\end{equation}
then  
\begin{equation}\label{eq:deltaIsolated}
    m_{\mathscr{X}}(\mathscr{C}(x_{1})\cap \mathscr{C}(x_{2}))=0.
\end{equation}
 This can be proven by contradiction. Let $\mathscr{C}(x_{1})\neq \mathscr{C}(x_{2})$ and assume that $ m_{\mathscr{X}}(\mathscr{C}(x_{1})\cap \mathscr{C}(x_{2}))\neq 0$. By \eqref{ass:finiteness} this implies that  $\mathscr{C}(x_{1})\cap \mathscr{C}(x_{2}) \neq \emptyset$. We can then pick $z\in \mathscr{C}(x_{1})\cap \mathscr{C}(x_{2})$, such that we have $z\stackrel{\delta}\leftrightsquigarrow x_{1}$
and $z\stackrel{\delta}\leftrightsquigarrow x_{2}$. Since    $( X, Y)\stackrel{d}{\leftrightarrow}(\delta,\delta_{2})$, by Lemma  \ref{lemma:commutativeProperty}   in  Appendix \ref{sec:AuxResult} this also implies $x_{1}\stackrel{\delta}{\leftrightsquigarrow} x_{2}$ , and therefore $\mathscr{C}(x_{1})= \mathscr{C}(x_{2})$, which is a contradiction. It follows that if $\mathscr{C}(x_{1})\neq \mathscr{C}(x_{2})$, then we must have $m_{\mathscr{X}}(\mathscr{C}(x_{1})\cap \mathscr{C}(x_{2}))=0$, 
and therefore
\begin{equation} \label{ventotto}
    \frac{m_{\mathscr{X}}(\mathscr{C}(x_{1})\cap \mathscr{C}(x_{2}))}{m_{\mathscr{X}}(\llbracket X\rrbracket)}=0\leq \delta. 
\end{equation}
We conclude that   Property~2 of Definition \ref{defoverlap} is satisfied. 

 Finally, 
  we have that  for any singly $\delta$-connected set  $\llbracket X|y \rrbracket $,  there exist an $x \in \llbracket X	\rrbracket$ such that $x \in \llbracket X|y \rrbracket$, which by   \eqref{eq:C1} implies $ \llbracket X|y \rrbracket  \subseteq \mathscr{C}(x)$. Namely, for every singly $\delta$-connected set,  
there exist a set in the family containing it. We can then conclude that $\mathcal{C}$ satisfies all the properties of a $\delta$-overlap family.

Next, we show that  $\mathcal{C}$ is a unique $\delta$-overlap family.  
 By contradiction, consider another $\delta$-overlap family  $\mathcal{D}$. For all $x\in \llbracket X\rrbracket$,
let $\mathscr{D}(x)$ denote a set in $\mathcal{D}$ containing $x$. Then, using the definition of $\mathscr{C}(x)$ and the fact that $\mathscr{D}(x)$ is $\delta$-connected,  it follows that
\begin{equation}\label{eq:1}
    \mathscr{D}(x)\subseteq \mathscr{C}(x).
\end{equation}
Next, we show that for all $x\in\llbracket X\rrbracket $, we also have 
\begin{equation} \label{eq:2}
    \mathscr{C}(x)\subseteq \mathscr{D}(x),
\end{equation}
from which we conclude that $\mathcal{D}=\mathcal{C}$.

The proof of \eqref{eq:2}  is also  obtained by contradiction. Assume there exists a point $\tilde{x}\in \mathscr{C}(x) \setminus \mathscr{D}(x)$.   
Since both $x$ and $\tilde{x}$ are contained in $\mathscr{C}(x)$, we have $\tilde{x}\stackrel{\delta}{\leftrightsquigarrow} x$.
Let $x^*$ be a point in a singly-connected set that is contained in $\mathscr{D}(x)$, namely $x^* \in \llbracket X|y^* \rrbracket \subseteq \mathscr{D}(x)$. Since both $x$ and $x^*$ are in $\mathscr{D}(x)$, we have that $x \stackrel{\delta}{\leftrightsquigarrow} x^*$.   Since     $( X, Y)\stackrel{d}{\leftrightarrow}(\delta,\delta_{2})$, we can apply Lemma  \ref{lemma:commutativeProperty}   in   Appendix \ref{sec:AuxResult} to conclude that $\tilde{x} \stackrel{\delta}{\leftrightsquigarrow} x^*$.
It follows that there exists a sequence of conditional ranges $\{\llbracket X|y_i\rrbracket\}_{i=1}^{N}$  such that $\tilde{x} \in \llbracket X|y_1 \rrbracket$ and $x^* \in \llbracket X|y_N \rrbracket$, which satisfies  \eqref{eq:condtionConnectedness}.  Since $x^*$ is in both $\llbracket X|y_N \rrbracket$ and  $\llbracket X|y^* \rrbracket $, we have  $\llbracket X|y_N \rrbracket \cap \llbracket X|y^* \rrbracket  \not = \emptyset$ and since $(X,Y) \stackrel{d}{\leftrightarrow}(\delta,\delta_2)$, we have
\begin{equation}
\frac{m_{\mathscr{X}}( \llbracket X|y_N \rrbracket \cap \llbracket X|y^*\rrbracket)}{m_{\mathscr{X}}(\llbracket X\rrbracket)} > \delta.
\end{equation}
Without loss of generality, we can then assume that the last element of our sequence is   $\llbracket X|y^* \rrbracket $. By Property 3 of Definition~\ref{defoverlap}, every conditional range in the sequence must be contained in some set of the $\delta$-overlap family $\mathcal{D}$. Since
 $\llbracket X|y^* \rrbracket \subseteq \mathscr{D}(x)$ and  $\llbracket X|y_1 \rrbracket \not \subseteq \mathscr{D}(x)$, it follows that there exist
two consecutive conditional ranges along the sequence and two sets of the $\delta$-overlap family covering them, such that $\llbracket X|y_{i-1}\rrbracket \subseteq \mathscr{D}(x_{i-1})$, $\llbracket X|y_{i}\rrbracket  \subseteq \mathscr{D}(x_{i})$, and  $\mathscr{D}(x_{i-1}) \not = \mathscr{D}(x_{i})$.
Then, we have
\begin{equation}\label{eq:lowerboundOverlap}
 \begin{split}
& \; m_{\mathscr{X}}( \mathscr{D}(x_{i-1}) \cap \mathscr{D}(x_{i}) )    \\
= & \; m_\mathscr{X}( (\llbracket X|y_{i-1}  \rrbracket \cap \llbracket X|y_{i} \rrbracket)    \cup  (\mathscr{D}(x^*_{i-1})\cap \mathscr{D}(x^*_{i}))) \\
\stackrel{(a)}{\geq} & \; m_\mathscr{X} ( \llbracket X|y_{i-1} \rrbracket \cap \llbracket X|y_{i}\rrbracket) \\
\stackrel{(b)}{>} & \; \delta m_{\mathscr{X}}(\llbracket X\rrbracket),
 \end{split}   
 \end{equation}
 where $(a)$ follows from \eqref{eq:strongtransitivity} and  $(b)$ follows from \eqref{eq:condtionConnectedness}. 
 It follows that 
\begin{equation}
    \frac{m_{\mathscr{X}}( \mathscr{D}(x_{i-1})\cap \mathscr{D}({x_{i}}))}{m_{\mathscr{X}}(\llbracket X\rrbracket)} > \delta,
    \end{equation}
 and Property~2 of Definition \ref{defoverlap}  is violated. Thus, $\tilde{x}$ does not exists, which implies $\mathscr{C}(x)\subseteq \mathscr{D}(x)$. 
Combining (\ref{eq:1}) and (\ref{eq:2}), 
we conclude that the $\delta$-overlap family $\mathcal{C}$ is  unique. 



We now turn to the proof of the second part of the Theorem. Since by \eqref{ventotto} the uncertainty associated to the  overlap between any two   sets of the $\delta$-overlap family $\mathcal{C}$ is zero,  it follows that   $\mathcal{C}$ is also a  partition.   

Now, we show that $\mathcal{C}$ is also a $\delta$-isolated partition. This can be proven by contradiction. Assume $\mathcal{C}$ is not a $\delta$-isolated partition. Then, there exists two distinct sets $\mathscr{C}(x_1),\mathscr{C}(x_2)\in \mathcal{C}$ such that $\mathscr{C}(x_1)$ and $\mathscr{C}(x_2)$ are not $\delta$-isolated. This implies that there exists a point $\bar{x}_1\in \mathscr{C}(x_1)$ and 
$\bar{x}_2\in \mathscr{C}(x_2)$ such that $\bar{x}_1\stackrel{\delta}{\leftrightsquigarrow}\bar{x}_2$. Using the fact that $\mathscr{C}(x_1)$ and $\mathscr{C}(x_2)$ are $\delta$-connected and   Lemma \ref{lemma:commutativeProperty} in   Appendix \ref{sec:AuxResult}, this implies that all points in the set $\mathscr{C}(x_1)$ are $\delta$-connected to all points in the set $\mathscr{C}(x_2)$. Now, let $x_1^*$ and $x_2^*$ be   points in a singly $\delta$-connected set contained in $\mathscr{C}(x_1)$ and $\mathscr{C}(x_2)$ respectively, namely  $x^*_1\in \llbracket X|y^*_1\rrbracket\subseteq \mathscr{C}(x_1)$ and $x^*_2\in \llbracket X|y^*_2\rrbracket\subseteq \mathscr{C}(x_2)$. Since $x^*_1\stackrel{\delta}{\leftrightsquigarrow}x^*_2$, there exists a sequence of conditional ranges $\{\llbracket X|y_i\rrbracket\}_{i=1}^N $ satisfying \eqref{eq:condtionConnectedness}, such that $x_1\in \llbracket X|y_1\rrbracket$ and $x_2\in \llbracket X|y_N\rrbracket$. Without loss of generality, we can assume $\llbracket X|y_1\rrbracket=\llbracket X|y^*_1\rrbracket$ and $\llbracket X|y_2\rrbracket=\llbracket X|y^*_2\rrbracket$. Since $\mathcal{C}$ is a partition, we have that
 $\llbracket X|y^*_1 \rrbracket \subseteq \mathscr{C}(x_1)$ and  $\llbracket X|y^*_2 \rrbracket \not \subseteq \mathscr{C}(x_1)$. It follows that there exist
two consecutive conditional ranges along the sequence $\{\llbracket X|y_i\rrbracket\}_{i=1}^N $ and two sets of the $\delta$-overlap family $\mathcal{C}$ covering them, such that $\llbracket X|y_{i-1}\rrbracket \subseteq \mathscr{C}(x_{i-1})$ and $\llbracket X|y_{i}\rrbracket  \subseteq \mathscr{C}(x_{i})$, and  $\mathscr{C}(x_{i-1}) \not = \mathscr{C}(x_{i})$. Similar to \eqref{eq:lowerboundOverlap}, we have 
\begin{equation}
    \frac{m_{\mathscr{X}}( \mathscr{C}(x_{i-1})\cap \mathscr{C}({x_{i}}))}{m_{\mathscr{X}}(\llbracket X\rrbracket)} > \delta,
    \end{equation}
and Property~2 of Definition \ref{defoverlap}  is violated. Thus, $\mathscr{C}(x_1)$ and $\mathscr{C}(x_2)$ do  not exist, which implies $\mathcal{C}$ is $\delta$-isolated partition. 

Let $\mathcal{P}$ be  any other  $\delta$-isolated partition.  We wish to show that $|\mathcal{C}| \geq |\mathcal{P}|$, and that  the equality holds if and only if $\mathcal{P}=\mathcal{C}$.
First, note that every set $\mathscr{C}(x)\in \mathcal{C}$ can intersect   at most one set in $ \mathcal{P}$, otherwise the sets in $\mathcal{P}$ would not be $\delta$-isolated. 
Second, since $\mathcal{C}$ is a cover of $\llbracket X\rrbracket$,   every set in $\mathcal{P}$ must be intersected by at least one set in $\mathcal{C}$. 
 It follows that
\begin{equation}
| \mathcal{C} |\geq |\mathcal{P}|. 
\end{equation}
Now, assume the equality holds. In this case, there is a one-to-one  correspondence $\mathscr{P}: \mathcal{C} \rightarrow \mathcal{P}$, 
such that for all $x \in \llbracket X \rrbracket$, we have $\mathscr{C}(x) \subseteq\mathscr{P}( \mathscr{C}(x))$, and since both $\mathcal{C}$ and $\mathcal{P}$ are partitions of  $\llbracket X \rrbracket$, it follows that $\mathcal{C}=\mathcal{P}$.
%
%
%
Conversely, assuming $\mathcal{C}=\mathcal{P}$, then $|\mathcal{C}|= |\mathcal{P}|$   follows trivially.  
\end{proof}

We have introduced the notion of  mutual information  from $Y$ to $X$ in terms of the conditional range $\llbracket X|Y\rrbracket$.
 Since in general we have   $\llbracket X|Y\rrbracket\neq  \llbracket Y|X\rrbracket$, one may   expect  the definition of mutual information to be asymmetric in its arguments. Namely, the amount of information provided about $X$ by the knowledge of $Y$ may not be the same as the amount of information provided about $Y$ by the knowledge of $X$. Although this is true in general, we show that for   disassociated UVs   symmetry is retained, provided that when  swapping $X$ with $Y$ one also rescales $\delta$ appropriately. 
The following   theorem establishes the symmetry in the mutual information under the appropriate scaling of the parameters $\delta_1$ and $\delta_2$. 
The proof  requires the introduction of the notions of taxicab connectedness, taxicab family, and taxicab partition, which are given in   Appendix \ref{sec:taxicabSection}.

\begin{theorem}\label{corr:symmetricPropOfInfo}
If  $( X, Y)\stackrel{d}{\leftrightarrow}(\delta_{1},\delta_{2})$, and a $(\delta_{1},\delta_{2})$-taxicab family of $\llbracket X,Y\rrbracket$ exists, then we have
\begin{equation}
      I_{\delta_1}(X;Y) = I_{\delta_2} (Y;X).
\end{equation}
\end{theorem}

\section{($\epsilon$,$\delta$)-Capacity} \label{sec:cap}
We now give an operational definition of capacity of a communication channel and relate it to the notion of mutual information between UVs introduced above.
Let $\mathscr{X}$ be a totally bounded, normed  metric space such that for all $x\in \mathscr{X}$ we have $\|x\|\leq 1$,  where $\|.\|$ represents    norm.
This normalization is for convenience of notation and all results can   easily be extended to  metric spaces of any bounded norm. 
Let $\mathcal{X} \subseteq \mathscr{X}$ be a discrete set of points in the space, which represents a codebook. Any point $x \in \mathcal{X}$ represents a codeword that can be selected at the transmitter, sent over the channel, and  received with noise perturbation at most $\epsilon$. Namely, for any     transmitted codeword $x \in \mathcal{X}$,  we receive a point in the set 

\begin{figure}[t]
\begin{center}
\includegraphics[width=.35\textwidth ]{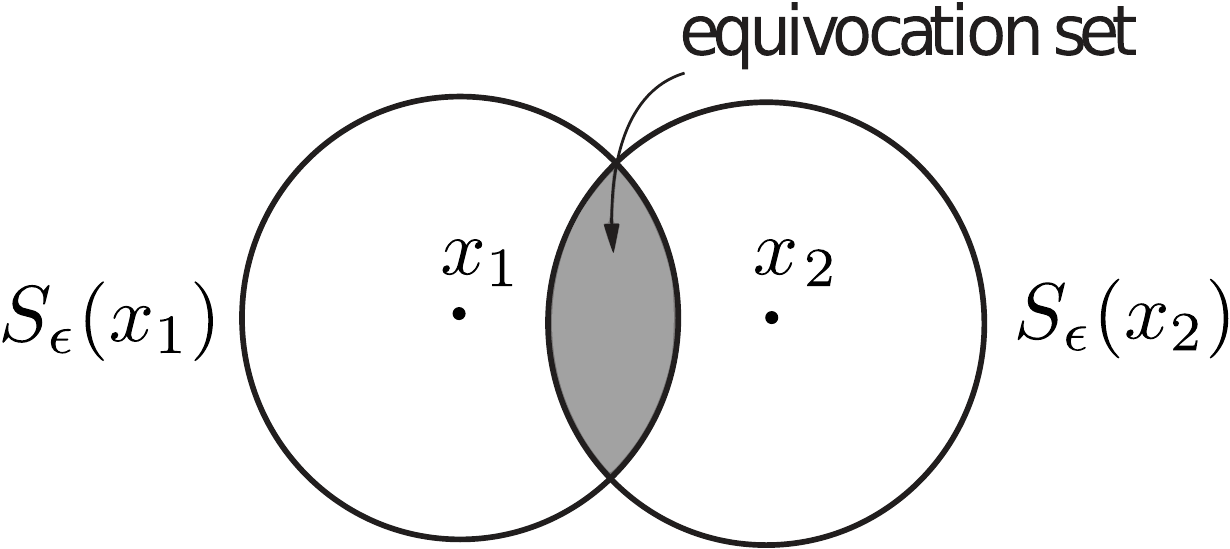}
\end{center}
\caption{
The size of the equivocation  set is inversely proportional to the amount of adversarial effort required to induce an error. }
\label{fig:equivocation}
\end{figure}
\begin{equation}\label{eq:setOfReceivedCodewords}
    S_\epsilon(x)=\{y \in   \mathscr{X} : \|x-y \|\leq \epsilon\}.
\end{equation}

It follows that all received codewords lie in the set $\mathcal{Y}=\bigcup_{x\in \mathcal{X}}S_\epsilon(x)$, where  $\mathcal{Y}\subseteq \mathscr{Y}= \mathscr{X}$. 
Transmitted codewords can be decoded correctly as long as the corresponding uncertainty sets at the receiver do not overlap. This can be done by simply associating the received codeword to the point in the codebook that is closest to it.  

For any $x_{1},x_{2}$ $\in$ $\mathcal{X}$, we now let 
\begin{equation}\label{eqerr}
   { e_\epsilon(x_{1},x_{2})=\frac{m_{\mathscr{Y}}(S_\epsilon(x_{1})\cap S_\epsilon(x_{2}))}{m_{\mathscr{Y}}(\mathscr{Y})}},
\end{equation}
where $m_{\mathscr{Y}}(.)$ is an uncertainity function defined over the  space $\mathscr{Y}$. We also assume without loss of generality that the uncertainty associated to the whole space $\mathscr{Y}$ of  received codewords is $m_\mathscr{Y}(\mathscr{Y}) = 1$.  
Finally, 
{ we let  $V_{\epsilon} \subseteq \mathscr{Y}$  be the smallest uncertainty set corresponding to a transmitted codeword, namely $V_{\epsilon}=S_{\epsilon}(x^*)$, where $x^*=\argmin_{x\in\mathscr{X}}m_{\mathscr{Y}}(S_{\epsilon}(x))$}.
 The quantity  $1-e_\epsilon(x_1,x_2)$ can be viewed   as the confidence we have of not confusing $x_1$ and $x_2$ in any transmission, or equivalently as the amount of adversarial effort required to induce a confusion between the two codewords. For example, if the uncertainty function  is constructed using a measure, then all the erroneous codewords generated by an adversary to decode $x_2$ instead than $x_1$ must  lie inside the equivocation set depicted in Figure~\ref{fig:equivocation}, whose relative size is given by \eqref{eqerr}. The smaller the equivocation set is, the larger must  be the effort  required by the adversary to induce an error. If the uncertainty function represents the diameter of the set, then all the erroneous codewords generated by an adversary to decode $x_2$ instead than $x_1$ will  be   close to each other, in the sense of  \eqref{eqerr}.  Once again, the closer the possible erroneous codewords are, the harder must  be for the adversary to generate an error, since any small deviation  allows the decoder  to correctly identify the transmitted codeword.   

We now introduce the notion of \emph{distinguishable codebook}, ensuring that every codeword cannot be confused with any other codeword, rather than with a specific one, at a given level of confidence.



\begin{definition}\label{def:distin}{$(\epsilon,\delta)$-distinguishable codebook.}\\
For any $0< \epsilon\leq 1$, $0\leq \delta < m_{\mathscr{Y}}(V_{\epsilon})$, a codebook $\mathcal{X} \subseteq \mathscr{X}$ is $(\epsilon,\delta)$-distinguishable if for all $x_{1},x_{2}\in \mathcal{X}$, we have $e_\epsilon(x_{1},x_{2}) \leq {\delta/|\mathcal{X}|}$.
\end{definition} 
For any   $(\epsilon,\delta)$-distinguishable codebook $\mathcal{X}$ and $x \in \mathcal{X}$,  we   let
\begin{equation}
    e_{\epsilon}(x)=\sum_{x'\in\mathcal{X}:x'\neq x}e_{\epsilon}(x,x').
\end{equation}
It now follows from   Definition~\ref{def:distin} that
\begin{equation}
    e_{\epsilon}(x) \leq \delta,
\end{equation}
and each codeword in   an $(\epsilon,\delta)$-distinguishable codebook can be decoded correctly with confidence at least $1-\delta$. Definition~\ref{def:distin}   guarantees even more, namely that the confidence of not confusing any pair of codewords is uniformly bounded by $1-\delta/|\mathcal{X}|$. This stronger constraint implies that we cannot ``balance'' the error associated to a codeword transmission by  allowing some decoding pair to have a lower confidence and enforcing other pairs to have higher confidence. This  is the main difference between our definition and the one used in~\cite{lim2017information} which bounds the average confidence, and allows us to relate the notion of  capacity to the mutual information between   pairs of codewords. 

\begin{definition}\label{defn:capacity}{$(\epsilon, \delta)$-capacity.}\\ For any totally bounded, normed metric space $\mathscr{X}$, $0< \epsilon\leq 1$, and $0\leq \delta <  m_{\mathscr{Y}}(V_{\epsilon})$,   the $(\epsilon,\delta)$-capacity of  $\mathscr{X}$ is 
\begin{equation}
    C_{\epsilon}^{\delta}=\sup_{\mathcal{X} \in \mathscr{X}_{\epsilon}^{\delta}} \log_2 |\mathcal{X}| \mbox{ bits},
\end{equation}
where $\mathscr{X}_{\epsilon}^{\delta}=\{\mathcal{X}: \mathcal{X} \mbox{ is } (\epsilon,\delta)\mbox{-distinguishable}\}$ is the set of  $(\epsilon,\delta)$-distinguishable codebooks.
\end{definition}

The $(\epsilon,\delta)$-capacity represents the largest number of bits that can be communicated by using any $(\epsilon, \delta)$-distinguishable codebook. 
The corresponding geometric picture is illustrated in Figure~\ref{fig:capacity}. For $\delta=0$, our notion of capacity reduces to Kolmogorov's $\epsilon$-capacity, that is the logarithm of the packing number of the space with balls of radius $\epsilon$. 
\begin{figure}
\begin{center}
\includegraphics[width=.45 \columnwidth ]{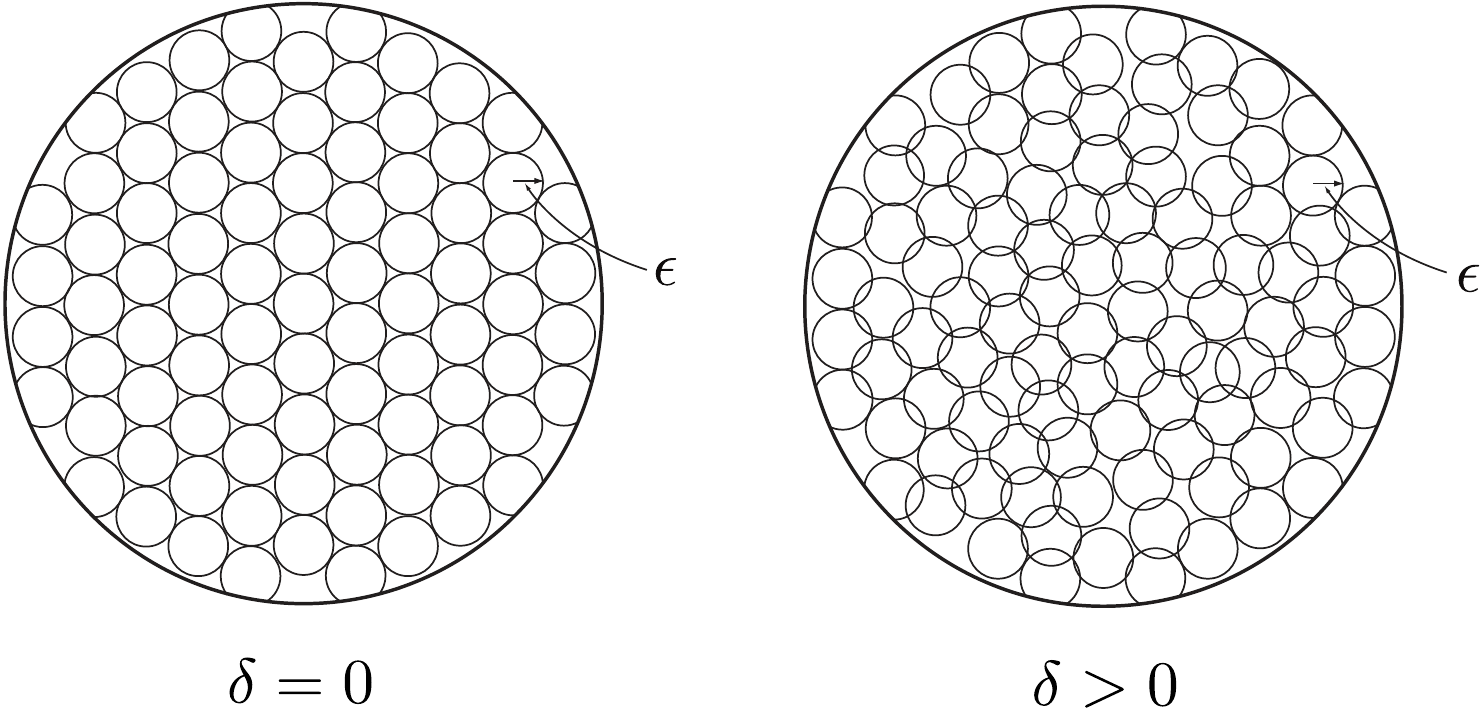}
\end{center}
\caption{Illustration of the $(\epsilon,\delta)$-capacity in terms of packing $\epsilon$-balls with maximum overlap $\delta$.}
\label{fig:capacity}
\end{figure}
In the   definition of capacity, we have restricted $\delta < m_{\mathscr{Y}}(V_{\epsilon}) $ to rule out the  case  when the decoding error
can be at least as large as the error introduced by the channel, and  the $(\epsilon,\delta)$-capacity is infinite. Also, note that $m_{\mathscr{Y}}(V_\epsilon)\leq 1$ since $V_\epsilon\subseteq\mathscr{Y}$ and \eqref{eq:strongtransitivity} holds.

We now relate our operational definition of capacity to the notion of UVs and mutual information introduced in Section~\ref{section:MaxMinInfo}.
Let $X$ be the UV corresponding to the transmitted { codeword}. 
This is a map $X: \mathscr{X}\to \mathcal{X}$ and $\llbracket X\rrbracket=\mathcal{X}\subseteq \mathscr{X}$.
Likewise, let $Y$ be the UV corresponding to the received codeword. This is a map $Y: \mathscr{Y}\to \mathcal{Y}$ and $\llbracket Y\rrbracket=\mathcal{Y}\subseteq \mathscr{Y}$. 
For our $\epsilon$-perturbation channel, these UVs are such that for all $y\in \llbracket Y\rrbracket$ and $x\in\llbracket X\rrbracket$, we have
\begin{align}
 \llbracket Y|x\rrbracket &= \{y\in   \llbracket Y \rrbracket: \|x-y\|\leq \epsilon\},  \label{uno}  \\
\llbracket X|y\rrbracket &= \{x\in \llbracket{X}\rrbracket: \|x-y\|\leq \epsilon\},  \label{due}
\end{align}
see Figure~\ref{points}. Clearly, the set in \eqref{uno} is continuous, while the set in \eqref{due} is discrete.
\begin{figure}[t]
\begin{center}
\includegraphics[width=.31 \textwidth ]{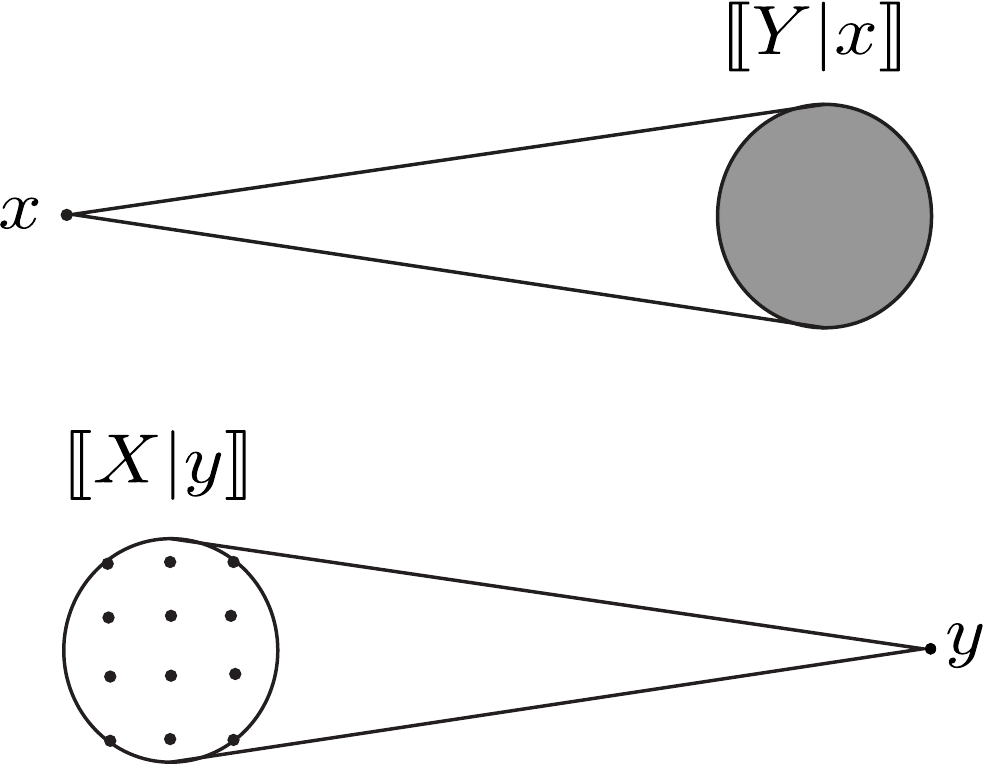}
\end{center}
\caption{Conditional ranges $\llbracket Y|x\rrbracket$ and $ \llbracket X|y \rrbracket$ due to the $\epsilon$-perturbation channel. }
\label{points}
\end{figure}

To measure the levels of association and disassociation between   $X$ and $Y$, we use an uncertainty function $m_{\mathscr{X}}(.)$  defined over $\mathscr{X}$, and  $m_{\mathscr{Y}}(.)$  defined over $\mathscr{Y}$. We 
introduce the feasible set   
{
\begin{equation}\label{eq:feasibleSet}
\begin{split}
    \mathscr{F}_{\delta}&= \{ X :
    \llbracket X\rrbracket\subseteq \mathscr{X}, 
    \mbox{ and either } \\
    &\qquad( X, Y)\stackrel{d}{\leftrightarrow}(0,{{\delta}}/|\llbracket X\rrbracket|) \\
&   \qquad   \mbox{ or } ( X, Y)\stackrel{a}{\leftrightarrow}(1,{{\delta}}/|\llbracket X\rrbracket|) \},
\end{split}
\end{equation}
}
representing the set of UVs $X$ such that the marginal range $\llbracket X\rrbracket$ is a discrete set representing a codebook, and the UV can either achieve   $(0,\delta/|\llbracket X\rrbracket|)$ levels of disassociation or $(1,\delta/|\llbracket X\rrbracket|)$ levels of  association with $Y$. In our channel model, this feasible set also depends on the $\epsilon$-perturbation   through \eqref{uno} and \eqref{due}.

We can now state the non-stochastic channel coding theorem for our $\epsilon$-perturbation channel.
\begin{theorem}\label{thm:channelCodingTheorem}   For any totally bounded,  normed metric space $\mathscr{X}$,  $\epsilon$-perturbation channel satisfying \eqref{uno} and \eqref{due}, $0< \epsilon\leq 1$ and $0\leq \delta < m_{\mathscr{Y}}(V_{\epsilon})$,
we have
\begin{equation}\label{eq:channelCodingTheorem}
       C_{\epsilon}^{\delta}
        = \sup_{ X \in \mathscr{F}_{\tilde{\delta}}, \tilde{\delta}\leq\delta/m_{\mathscr{Y}}(\llbracket Y\rrbracket)} I_{\tilde{\delta}/|\llbracket X\rrbracket|}(Y;X)\mbox{ bits}.
\end{equation}
\end{theorem}
\begin{proof}
{ First, we show that there exists a UV $X$ and $\tilde{\delta}\leq \delta/m_{\mathscr{Y}}(\llbracket Y\rrbracket)$ such that ${X}\in\mathscr{F}_{\tilde{\delta}}$, which implies that the supremum is well defined. }
Second, for all $X$ and $\tilde{\delta}$ such that
\begin{equation}
    X \in \mathscr{F}_{\tilde{\delta}},
\end{equation}
and
\begin{equation}
    \tilde{\delta}\leq\delta/m_{\mathscr{Y}}(\llbracket Y\rrbracket),
\end{equation}
we show that 
\begin{equation}
I_{\tilde{\delta}/|\llbracket X\rrbracket|}(Y;X)   \leq C_{\epsilon}^{\delta}.
\end{equation}
 Finally, we show the existence of    $X\in  \mathscr{F}_{\tilde{\delta}}$ and $\tilde{\delta}\leq \delta/m_{\mathscr{Y}}(\llbracket Y\rrbracket)$ such that $I_{\tilde{\delta}/|\llbracket X\rrbracket|}(Y;X)=C_\epsilon^\delta$.

{Let us begin with the first step. Consider a point $x\in\mathscr{X}$. Let $X$ be a UV such that 
\begin{equation}\label{eq:Xval}
    \llbracket X\rrbracket=\{x\}. 
\end{equation}
 Then, we have that the marginal range of the UV $Y$ corresponding to the received variable  is 
\begin{equation}\label{eq:Yval}
    \llbracket Y\rrbracket =\llbracket Y|x\rrbracket,
\end{equation}
and therefore for all $y\in\llbracket Y\rrbracket$, we have
\begin{equation}\label{eq:Yval1.1}
    \llbracket X|y\rrbracket=\{x\}.
\end{equation}
Using Definition \ref{def:AssSets} and \eqref{eq:Xval}, we have that
\begin{equation}\label{eq:Xsetdef}
    \mathscr{A}(Y;X)=\emptyset,
\end{equation}
because  $\llbracket X\rrbracket$ consists of a single point, and therefore the set in \eqref{eqdef1} is empty. 

On the other hand, using Definition \ref{def:AssSets} and \eqref{eq:Yval1.1}, we have
\begin{equation}\label{eq:Ysetdef}
    \mathscr{A}(X;Y)=\begin{cases}
    \{1\}  \mbox{ if   $\exists  y_1,y_2 \in \llbracket Y\rrbracket$,}\\
    \emptyset \quad\mbox{ otherwise}.
    \end{cases}
\end{equation}
Using \eqref{eq:Xsetdef} and since $\mathscr{A} \preceq \delta$ holds for  $\mathscr{A}=\emptyset$,   we have 
\begin{equation}\label{eq:cond1_1}
    \mathscr{A}(Y;X)\preceq \delta/(|\llbracket X\rrbracket|m_{\mathscr{Y}}(\llbracket Y\rrbracket)).
\end{equation}
Similarly, using \eqref{eq:Ysetdef},   we have 
\begin{equation}\label{eq:cond1_2}
    \mathscr{A}(X;Y)\preceq 1.
\end{equation}
Now, combining \eqref{eq:cond1_1}  and \eqref{eq:cond1_2}, we have
\begin{equation} (X,Y)\stackrel{a}{\leftrightarrow}(1,\delta/(|\llbracket X\rrbracket|m_{\mathscr{Y}}(\llbracket Y\rrbracket))).
\end{equation}
Letting $\tilde{\delta}=\delta/m_{\mathscr{Y}}(\llbracket Y\rrbracket)$, this implies that  $X\in \mathscr{F}_{\tilde\delta}$ and the first step of the proof is complete. 
}

To prove the second step, we define the set of discrete UVs
\begin{equation}
\begin{split}
    \mathscr{G}&=\{X:\llbracket X\rrbracket\subseteq \mathscr{X}, \exists \tilde{\delta}\leq \delta/m_{\mathscr{Y}}(\llbracket Y\rrbracket)\\
    &\qquad\qquad\mbox{ such that } \forall \mathscr{S}_{1},\mathscr{S}_{2}\in \llbracket Y|X\rrbracket,\\
    &\qquad \qquad  m_{\mathscr{Y}}(\mathscr{S}_{1}\cap \mathscr{S}_{2})/m_{\mathscr{Y}}(\llbracket Y\rrbracket)\leq \tilde{\delta}/|\llbracket X\rrbracket|\},
\end{split}
\end{equation}
which is a larger set than the one containing all UVs $X$ that are $(1,\tilde{\delta}/|\llbracket X\rrbracket|)$ associated to $Y$. Now, we will show that if an UV $X\in \mathscr{G}$, then the corresponding  codebook $\mathcal{X}\in \mathscr{X}_{\epsilon}^{\delta}$. 
If $X\in\mathscr{G}$, then there exists a $\tilde\delta\leq \delta/m_{\mathscr{Y}}(\llbracket Y\rrbracket)$ such that 
for all $\mathscr{S}_1,\mathscr{S}_2\in\llbracket Y|X\rrbracket$ we have
\begin{equation}
   \frac{ m_{\mathscr{Y}}(\mathscr{S}_{1}\cap \mathscr{S}_{2})}{m_{\mathscr{Y}}(\llbracket Y\rrbracket)}\leq \frac{\tilde{\delta}}{|\llbracket X\rrbracket|}.
\end{equation}
It follows that
for all $x_1,x_2\in\llbracket X\rrbracket$, we have
\begin{equation}
    \frac{m_{\mathscr{Y}}(\llbracket Y|x_1\rrbracket\cap \llbracket Y|x_2\rrbracket)}{m_{\mathscr{Y}}(\llbracket Y\rrbracket)}\leq \frac{\tilde\delta}{|\llbracket{X}\rrbracket|}.
\end{equation}
Using $\mathcal{X}=\llbracket X\rrbracket$, \eqref{uno}, $\llbracket Y\rrbracket=\mathcal{Y}=\bigcup_{x\in \mathcal{X}}S_\epsilon(x)$ and $m_{\mathscr{Y}}(\mathscr{Y})=1$,  for all $x_1,x_2\in\mathcal{X}$  we have  
\begin{equation}
\begin{split}
    \frac{m_{\mathscr{Y}}(S_\epsilon(x_{1})\cap S_\epsilon(x_{2}))}{m_{\mathscr{Y}}(\mathscr{Y})}&\leq \frac{\tilde{\delta} m_{\mathscr{Y}}(\llbracket Y\rrbracket)}{|\mathcal{X}|},\\
    &\stackrel{(a)}{\leq} \frac{\delta}{|\mathcal{X}|},
\end{split}
\end{equation}
where $(a)$ follows from     $\tilde{\delta}\leq \delta/m_{\mathscr{Y}}(\llbracket Y\rrbracket)$. 
Putting things together, it follows that
\begin{equation}\label{eq:TwoSetsEqual_1}
X\in \mathscr{G} \implies \mathcal{X}\in \mathscr{X}_{\epsilon}^{\delta}  
\end{equation}
{ Consider now a pair  $X$ and $\tilde{\delta}$ such  that  $\tilde{\delta}\leq\delta/m_{\mathscr{Y}}(\llbracket Y\rrbracket)$  and 
\begin{equation}\label{eq:setin_1}
    X \in \mathscr{F}_{\tilde{\delta}}.
\end{equation}
If $(X,Y)\stackrel{d}{\leftrightarrow}(0,\tilde{\delta}/|\llbracket X\rrbracket|)$, then using Lemma \ref{lemma:equivalenceAssDiss} in   Appendix \ref{sec:AuxResult}   there exist   two UVs $\bar{X}$ and $\bar{Y}$, and   $\bar{\delta}\leq \delta/m_{\mathscr{Y}}(\llbracket\bar{Y}\rrbracket)$ such that
\begin{equation}\label{eq:che12}
    (\bar X, \Bar{Y})  
\stackrel{a} \leftrightarrow (1,\bar\delta/|\llbracket\bar{X}\rrbracket|),
\end{equation}  and
\begin{equation}\label{eq:che13}
    |\llbracket Y|X\rrbracket^*_{\tilde\delta/|\llbracket X\rrbracket|}|= |\llbracket \Bar Y|\bar{X} \rrbracket_{\bar\delta/|\llbracket\bar{X}\rrbracket|}^*|.
\end{equation}
On the other hand, if $ ( X, {Y})  
\stackrel{a} \leftrightarrow (1,\tilde\delta/|\llbracket{X}\rrbracket|)$, then \eqref{eq:che12} and \eqref{eq:che13} also trivially hold. It then follows that \eqref{eq:che12} and \eqref{eq:che13} hold for all $X\in \mathscr{F}_{\tilde{\delta}}$.
We now have 
\begin{equation}
\begin{split}
I_{\tilde{\delta}/|\llbracket X\rrbracket|}(Y;X)&=\log(|\llbracket Y|X\rrbracket_{\tilde{\delta}/|\llbracket X\rrbracket|}^*|)\\
&\stackrel{(a)}{=}\log(|\llbracket \bar Y|\bar X\rrbracket_{\bar{\delta}/|\llbracket \bar X\rrbracket|}^*|)\\
&\stackrel{(b)}{\leq} \log(|\llbracket \bar X\rrbracket|),\\
&\stackrel{(c)}{=} \log(|\bar{\mathcal{X}}|),\\ 
&\stackrel{(d)}{\leq} C^\delta_\epsilon,
\end{split}
\end{equation}
where $(a)$ follows from  \eqref{eq:che12} and \eqref{eq:che13}, $(b)$ follows from Lemma \ref{lemma:RelationToCardinality} in   Appendix \ref{sec:AuxResult} since $\bar{\delta}\leq\delta/m_{\mathscr{Y}}(\llbracket \bar Y\rrbracket) < m_{\mathscr{Y}}(V_{\epsilon})/m_{\mathscr{Y}}(\llbracket \bar Y\rrbracket)$, $(c)$ follows by defining the codebook $\bar{\mathcal{X}}$ corresponding  to the UV $\bar{X}$, and $(d)$ follows from the fact that using \eqref{eq:che12} and Lemma \ref{lemma:Dissassociation}, we have  $\bar{X}\in \mathcal{G}$, which implies by  \eqref{eq:TwoSetsEqual_1} that $\bar{\mathcal{X}}\in\mathscr{X}_{\epsilon}^{\delta}$. 
}

Finally, 
let 
\begin{equation}
\mathcal{{X}}^* =\mbox{argsup}_{\mathcal{X}\in\mathscr{X}_{\epsilon}^{\delta}}\log(|\mathcal{X}|), 
\end{equation}
which achieves the capacity $C_{\epsilon}^{\delta}$. Let ${{X}}^*$ be the UV whose marginal range corresponds to the codebook $\mathcal{{X}}^*$.
It follows that for all $\mathscr{S}_{1},\mathscr{S}_{1}\in \llbracket Y^*|{X}^*\rrbracket$, we have 
\begin{equation}
   \frac{ m_{\mathscr{Y}}(\mathscr{S}_{1}\cap \mathscr{S}_{1})}{m_{\mathscr{Y}}(\mathscr{Y})}\leq \frac{\delta}{|\llbracket X^*\rrbracket|},
\end{equation}
which implies using the fact that $m_{\mathscr{Y}}(\mathscr{Y})=1$,
\begin{equation}
    \frac{m_{\mathscr{Y}}(\mathscr{S}_{1}\cap \mathscr{S}_{1})}{m_{\mathscr{Y}}(\llbracket Y^*\rrbracket)}\leq \frac{\delta}{|\llbracket X^*\rrbracket| m_{\mathscr{Y}}(\llbracket Y^*\rrbracket)}.
\end{equation}
Letting   $\delta^*=\delta/m_{\mathscr{Y}}(\llbracket Y^*\rrbracket)$, and using Lemma \ref{lemma:Dissassociation}, we have that $( X^*,{Y}^*) \stackrel{a}\leftrightarrow (1,\delta^*/|\llbracket X^*\rrbracket|)$, which implies  ${{X}^*}\in \cup_{\tilde{\delta}\leq\delta/m_{\mathscr{Y}}(\llbracket Y^*\rrbracket} \mathscr{F}_{\tilde{\delta}}$ and the proof is complete. 

\end{proof}

Theorem~\ref{thm:channelCodingTheorem} characterizes the capacity as the supremum of the mutual information over all {UVs} in the feasible set.  The following theorem shows that the same characterization is obtained if we optimize the right hand side in \eqref{eq:channelCodingTheorem} over all {UVs} in the space. It follows that  by Theorem~\ref{thm:channelCodingTheorem}, rather than optimizing over all {UVs  representing all the codebooks in the space},  a capacity achieving codebook can  be found within the  smaller class $\cup_{\tilde{\delta}\leq \delta/m_{\mathscr{Y}}(V_\epsilon)} \mathscr{F}_{\tilde{\delta}}$ of feasible sets with error at most $\delta/m_{\mathscr{Y}}(V_\epsilon)$, since for all $\llbracket Y\rrbracket\subseteq\mathscr{Y}$, $m_{\mathscr{Y}}(V_\epsilon)\leq m_{\mathscr{Y}}(\llbracket Y\rrbracket)$. 

\begin{theorem} \label{cor}
The $(\epsilon,\delta)$-capacity in \eqref{eq:channelCodingTheorem} can also be written as
\begin{equation}
       C_{\epsilon}^{\delta}
        = \sup_{\substack{ X: \llbracket X\rrbracket \subseteq \mathscr{X},\\\tilde{\delta}\leq\delta/m_{\mathscr{Y}}(\llbracket Y\rrbracket)}} I_{\tilde{\delta}/|\llbracket X\rrbracket|}(Y;X)\mbox{ bits}.
\end{equation}
\end{theorem}
\begin{proof}
Consider an UV $X \not \in \cup_{\tilde{\delta}\leq \delta/m_{\mathscr{Y}}(\llbracket Y\rrbracket)}\mathscr{F}_{\tilde{\delta}}$, where $Y$ is the  corresponding UV at the receiver. The idea of the proof is to show the existence of an UV $\bar{X}\in \cup_{\tilde{\delta}\leq \delta/m_{\mathscr{Y}}(\llbracket \bar{Y}\rrbracket)}\mathscr{F}_{\tilde{\delta}}$ and the corresponding UV $\bar{Y}$ at the receiver, and  
\begin{equation}
\Bar{\delta}=\tilde{\delta}  m_{\mathscr{Y}}(\llbracket {Y}\rrbracket)/m_{\mathscr{Y}}(\llbracket \bar{Y}\rrbracket) \leq \delta/m( \llbracket \bar{Y} \rrbracket),
\end{equation}
such that the cardinality of the overlap partitions    
\begin{equation}\label{eq:CardinalityCheck2}
|\llbracket \bar Y|\bar{X}\rrbracket^{*}_{\bar{\delta}/|\llbracket \bar{X}\rrbracket|}|=|\llbracket  Y|{X}\rrbracket^{*}_{\tilde{\delta}/|\llbracket X\rrbracket|}|.
\end{equation}

Let the cardinality
\begin{equation}\label{eq:CardCodebok}
    |\llbracket Y|X\rrbracket_{{\tilde{\delta}}/|\llbracket X\rrbracket|}^{*}|=K.
\end{equation}
  By Property 1 of Definition~\ref{defoverlap}, we have that for all $\mathscr{S}_{i}\in\llbracket Y|X\rrbracket_{\tilde{\delta}/|\llbracket X\rrbracket|}^{*}$, there exists an $x_{i}\in \llbracket X\rrbracket$ such that $\llbracket Y|x_i\rrbracket \subseteq \mathscr{S}_{i}$. Now, consider another UV $\Bar{X}$ whose marginal range is composed of $K$ elements of $\llbracket {X}\rrbracket$, namely 
\begin{equation}\label{eq:newCodebook}
    \llbracket \bar{X}\rrbracket=\{x_1,\ldots x_K\}.
\end{equation}
 Let $\bar{Y}$ be the UV corresponding to the   received variable. Using the fact that for all $x\in\mathscr{X}$, we have $\llbracket \Bar{Y}|x\rrbracket=\llbracket{Y}|x\rrbracket$ since  \eqref{uno} holds, and using Property 2 of Definition~\ref{defoverlap},  for all $x,x^{\prime}\in \llbracket\bar{X}\rrbracket$,   we have that
\begin{equation}\label{eq:existence1}
\begin{split}
      \frac{m_{\mathscr{Y}}(\llbracket \bar{Y}|x\rrbracket\cap \llbracket \bar{Y}|x^{\prime}\rrbracket)}{m_{\mathscr{Y}}(\llbracket {Y}\rrbracket)}&\leq \frac{\tilde{\delta}}{|\llbracket X\rrbracket|},\\
      &\stackrel{(a)}{\leq}  \frac{\tilde{\delta}}{|\llbracket \bar{X}\rrbracket|},
\end{split}
\end{equation}
where $(a)$ follows from the fact that $\llbracket\bar{X}\rrbracket\subseteq\llbracket X\rrbracket$ using \eqref{eq:newCodebook}. Then,  for all $x,x^{\prime}\in \llbracket\bar{X}\rrbracket$, we have that
\begin{equation}\label{eq:existence2}
\begin{split}
    \frac{m_{\mathscr{Y}}(\llbracket \bar{Y}|x\rrbracket\cap \llbracket \bar{Y}|x^{\prime}\rrbracket)}{m_{\mathscr{Y}}(\llbracket \bar{Y}\rrbracket)}&\leq
    \frac{\tilde{\delta} m_{\mathscr{Y}}(\llbracket {Y}\rrbracket)}{|\llbracket \bar X\rrbracket| m_{\mathscr{Y}}(\llbracket \bar{Y}\rrbracket)}\\
    &=\frac{\bar{\delta}}{|\llbracket \bar X\rrbracket|},
\end{split}
\end{equation}
since $\bar\delta=\tilde{\delta} m_{\mathscr{Y}}(\llbracket {Y}\rrbracket)/m_{\mathscr{Y}}(\llbracket \bar{Y}\rrbracket)$. Then, by Lemma~\ref{lemma:Dissassociation} it follows  that
 \begin{equation}
     (\bar X,\bar{Y})\stackrel{a}{\leftrightarrow}(1,\bar{\delta}/|\llbracket \bar X\rrbracket|). \label{eq:asc}
 \end{equation} 
Since $\tilde{\delta}\leq \delta/m_{\mathscr{Y}}(\llbracket {Y}\rrbracket)$, we have
\begin{equation}\label{eq:rangeDelta1}
    \bar\delta\leq \delta/m_{\mathscr{Y}}(\llbracket \bar{Y}\rrbracket)<m_{\mathscr{Y}}(V_{\epsilon})/m_{\mathscr{Y}}(\llbracket \bar{Y}\rrbracket). 
\end{equation}
 Therefore, $\bar{X}\in \mathscr{F}_{\bar{\delta}}$ and $\bar{\delta}\leq \delta/m_{\mathscr{Y}}(\llbracket \bar{Y}\rrbracket)$.
 We now have that 
 \begin{equation}\label{eq:toShowStepTwo}
 \begin{split}
     |\llbracket \bar Y|\bar{X}\rrbracket^{*}_{\bar{\delta}/|\llbracket \bar X\rrbracket|}| &\stackrel{(a)}{=}|\llbracket\bar{X}\rrbracket| \\
&     \stackrel{(b)}{=}  | \llbracket Y|X\rrbracket_{{\tilde{\delta}}/|\llbracket X\rrbracket|}^{*}|,
 \end{split}
 \end{equation}
 where $(a)$ follows by applying Lemma 
 \ref{lemma:cardinalityAssociation} in   Appendix \ref{sec:AuxResult} using    \eqref{eq:asc} and \eqref{eq:rangeDelta1}, and $(b)$ follows from \eqref{eq:CardCodebok} and \eqref{eq:newCodebook}.
  Combining \eqref{eq:toShowStepTwo}  with Theorem \ref{thm:channelCodingTheorem}, the  proof is complete.
\end{proof}

 
 Finally, we make some considerations with respect to previous results in the literature.
First, we note that for $\delta=0$, all of our definitions reduce to Nair's ones and  
Theorem \ref{thm:channelCodingTheorem} recovers Nair's coding theorem \cite[Theorem 4.1]{nair2013nonstochastic}  for the zero-error capacity    of an additive $\epsilon$-perturbation channel. 

Second, we point out  that the $(\epsilon,\delta)$-capacity  considered in \cite{lim2017information}  defines the set of $(\epsilon,\delta)$-distinguishable codewords  
such that the \emph{average} overlap among all codewords is at most $\delta$.
In contrast,   our definition  requires the  overlap for \emph{each} pair of codewords to be at most $\delta/|\mathcal{X}|$. 
The following theorem provides the relationship between our  $C_{\epsilon}^{\delta}$ and the capacity    $\tilde{C}_{\epsilon}^\delta$   considered in~\cite{lim2017information}, which is defined using the Euclidean norm.
\begin{theorem}\label{lemma:relationbetweenTwoCapacity}Let 
$\tilde{C}_{\epsilon}^\delta$   be the $(\epsilon, \delta)$-capacity defined  in~\cite{lim2017information}. 
 We have
\begin{equation} \label{eqi}
C_{\epsilon}^{\delta}\leq \tilde{C}_{\epsilon}^{\delta /(2 m_{\mathscr{Y}}(V_{\epsilon}))},
\end{equation}
 and 
\begin{equation} \label{eqii}
    \tilde{C}^{\delta}_{\epsilon}\leq  {C}^{\delta m_{\mathscr{Y}}(V_{\epsilon}) 2^{2\tilde{C}^{\delta}_{\epsilon}+1}}_{\epsilon}. 
\end{equation}
\end{theorem}
\begin{proof}
For every codebook $\mathcal{X}\in\mathscr{X}^{\delta}_{\epsilon}$ and $x_1,x_2\in \mathcal{X}$, we have
\begin{equation}
    e_{\epsilon}(x_1,x_2)\leq \delta/|\mathcal{X}|. 
\end{equation}
Since $m_{\mathscr{Y}}(\mathscr{Y})=1$, this implies that for all $x_1,x_2\in \mathcal{X}$, we have
\begin{equation}\label{eq:uppBoundCap}
\begin{split}
    m_{\mathscr{Y}}(S_\epsilon(x_{1})\cap S_\epsilon(x_{2}))&\leq \delta /|\mathcal{X}|.
\end{split}
\end{equation}

 For all $\mathcal{X}\in\mathscr{X}$, the average overlap   defined in \cite[(53)]{lim2017information} is
{
 \begin{equation}
     \Delta=\frac{1}{|\mathcal{X}|}\sum_{x\in\mathcal{X}}\frac{e_{\epsilon}(x)}{2 m_{\mathscr{Y}}(V_\epsilon)}.
 \end{equation}
 }
 Then, we have
\begin{equation}
\begin{split}
    \Delta&=\frac{1}{2|\mathcal{X}|m_{\mathscr{Y}}(V_{\epsilon})}\sum_{x_1,x_2\in\mathcal{X}} m_{\mathscr{Y}}(S_\epsilon(x_{1})\cap S_\epsilon(x_{2})),\\
    &\stackrel{(a)}{\leq}\frac{\delta  |\mathcal{X}|^2 }{2|\mathcal{X}|^2 m_{\mathscr{Y}}(V_{\epsilon}) },\\
    &\leq \frac{\delta }{2m_{\mathscr{Y}}(V_{\epsilon}) },
\end{split}
\end{equation}
where $(a)$ follows from \eqref{eq:uppBoundCap}.
Thus, we have 
\begin{equation}
    C_{\epsilon}^{\delta}\leq \tilde{C}_{\epsilon}^{\delta /(2 m_{\mathscr{Y}}(V_{\epsilon}))},
\end{equation}
 and  \eqref{eqi} follows.

Now, let $\mathcal{X}$ be a codebook  with average overlap at most $\delta$, namely 
\begin{equation}
    \frac{1}{2 |\mathcal{X}| m_{\mathscr{Y}}(V_{\epsilon})} \sum_{x_1,x_2\in\mathcal{X}}m_{\mathscr{Y}}(S_\epsilon(x_{1})\cap S_\epsilon(x_{2}))\leq \delta.
\end{equation}
This implies that for all $x_1,x_2\in \mathcal{X}$, we have
\begin{equation}
\begin{split}
\frac{|\mathcal{X}|  m_{\mathscr{Y}}(S_\epsilon(x_{1})\cap S_\epsilon(x_{2}))}{  m_{\mathscr{Y}} (\mathscr{Y})}&\leq \frac{2 \delta |\mathcal{X}|^2 m_{\mathscr{Y}}(V_{\epsilon})}{m_{\mathscr{Y}} (\mathscr{Y})},\\
&\stackrel{(a)}{=} {2\delta |\mathcal{X}|^2 m_{\mathscr{Y}}(V_{\epsilon})},\\
&\leq{\delta  2^{2 \tilde{C}^{\delta}_{\epsilon}+1}} m_{\mathscr{Y}}(V_{\epsilon}),
\end{split}
\end{equation}
where $(a)$ follows from the fact that   $m_{\mathscr{Y}} (\mathscr{Y})=1$. Thus, we have
\begin{equation}
    \tilde{C}^{\delta}_{\epsilon}\leq  {C}^{\delta 2^{2\tilde{C}^{\delta}_{\epsilon}+1}m_{\mathscr{Y}}(V_{\epsilon})}_{\epsilon}, 
\end{equation}
and \eqref{eqii} follows.
\end{proof}

\section{$(N,\delta)$-Capacity of General Channels} \label{sec:general}
We now extend our results to more general channels where the noise  can be different across   codewords, and not necessarily contained within a ball of radius $\epsilon$.

Let $\mathcal{X} \subseteq \mathscr{X}$ be a discrete set of points in the space, which represents a codebook. Any point $x \in \mathcal{X}$ represents a codeword that can be selected at the transmitter, sent over the channel, and  received with perturbation. 
A channel with transition mapping $N:\mathscr{X}\to \mathscr{Y}$ associates to any  point in $\mathscr{X}$   a set in $\mathscr{Y}$, such that the received codeword lies in the set 
\begin{equation}
    S_{N}(x)= \{y\in \mathscr{Y}: y\in N(x) \}. \label{eqSN}
\end{equation}
Figure~\ref{fig:points1} illustrates  possible uncertainty sets associated to three different codewords.
\begin{figure}
\begin{center}
\includegraphics[width=.34 \textwidth ]{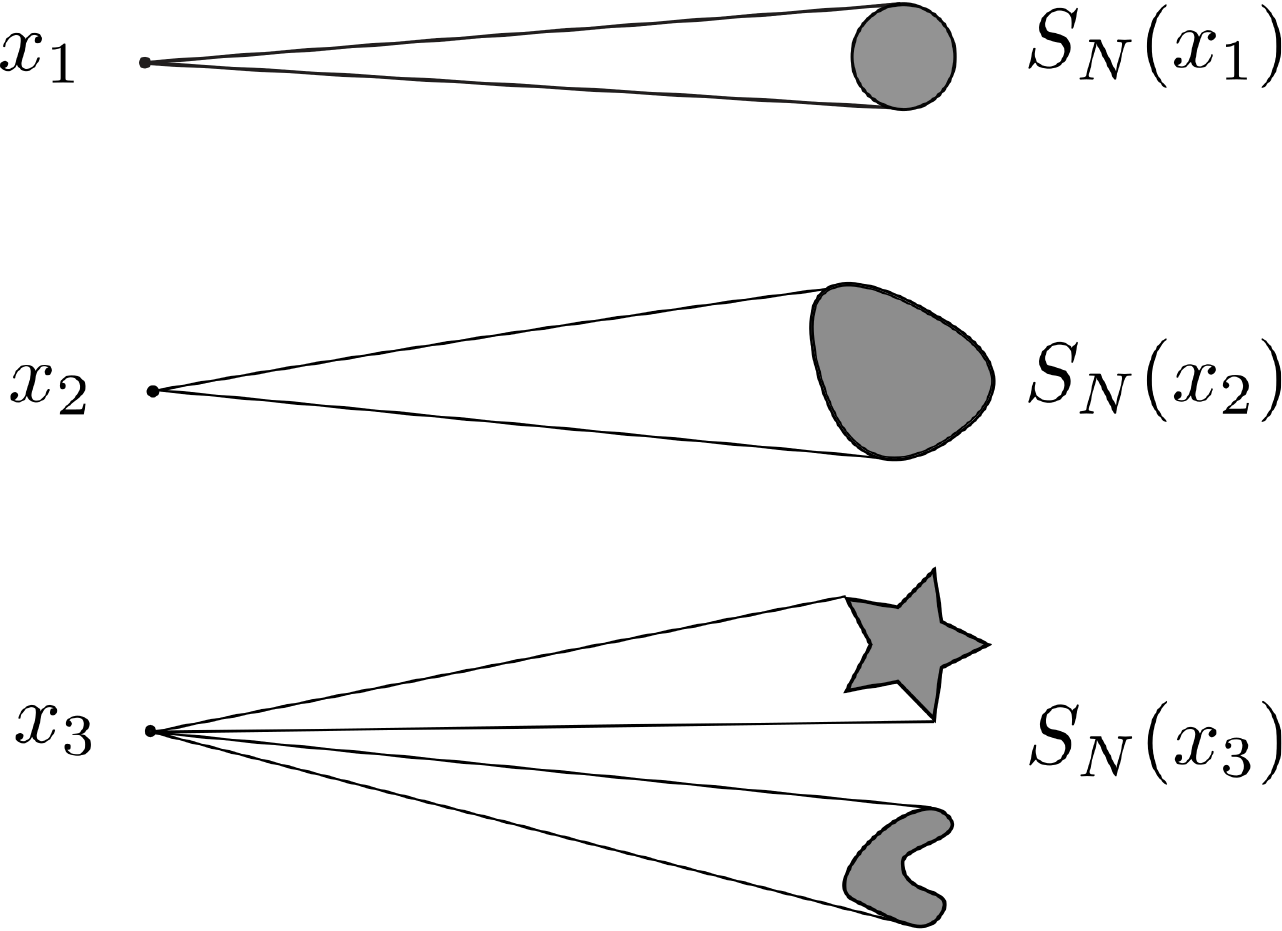}
\end{center}
\caption{Uncertainty sets associated to three different codewords. Sets are not necessarily balls, they can be different across codewords, and also be composed of disconnected subsets.}
\label{fig:points1}
\end{figure}

All received codewords lie in the set $\mathcal{Y}=\bigcup_{x\in \mathcal{X}}S_N(x)$, where  $\mathcal{Y}\subseteq \mathscr{Y}$. 
For any $x_{1},x_{2}$ $\in$ $\mathcal{X}$, we now let
\begin{equation}
    e_N(x_{1},x_{2})=\frac{m_{\mathscr{Y}}(S_N(x_{1})\cap S_N(x_{2}))}{m_{\mathscr{Y}}(\mathscr{Y})},
\end{equation}
where $m_{\mathscr{Y}}(.)$ is an uncertainty function defined over $\mathscr{Y}$.  We also assume without loss of generality that the uncertainty associated with the space $\mathscr{Y}$ of received codewords is $m_{\mathscr{Y}}(\mathscr{Y})=1$. We  also let  $V_{N}=N(x^*)$, where $x^*=\mbox{argmin}_{x\in \mathscr{X}} m_{\mathscr{Y}}(N(x))$. Thus, $V_{N}$ is the set corresponding to the minimum uncertainty introduced by the noise mapping $N$.  

\begin{definition}{$(N,\delta)$-distinguishable codebook.}\\
For any  $0\leq \delta < m_{\mathscr{Y}}(V_{N})$, a codebook $\mathcal{X} \subseteq \mathscr{X}$ is $(N,\delta)$-distinguishable if for all $x_{1},x_{2}\in \mathcal{X}$, we have $e_N(x_{1},x_{2}) \leq \delta/|\mathcal{X}|$. 
\end{definition}

\begin{definition}{$(N,\delta)$-capacity.} \label{defn:capacityN} \\ For any  totally bounded, normed metric  space $\mathscr{X}$, channel with transition mapping $N$, and $0\leq \delta <  m_{\mathscr{Y}}(V_{N})$,  the $(N,\delta)$-capacity of  $\mathscr{X}$ is 
\begin{equation}
    C_{N}^{\delta}=\sup_{\mathcal{X} \in \mathscr{X}_{N}^{\delta}} \log_2 |\mathcal{X}| \mbox{ bits},
\end{equation}
where $\mathscr{X}_{N}^{\delta}=\{\mathcal{X}: \mathcal{X} \mbox{ is } (N,\delta)\mbox{-distinguishable}\}$.
\end{definition}

We now relate our operational definition of capacity to the notion of UVs and mutual information introduced in Section~\ref{section:MaxMinInfo}.
As usual, let $X$ be the UV corresponding to the transmitted {codeword}  
and $Y$ be the UV corresponding to the received {codeword}. 
For a channel with transition mapping $N$, these UVs are such that for all $y\in \llbracket Y\rrbracket$ and $x\in \llbracket X\rrbracket$, we have
\begin{align}
\llbracket Y|x\rrbracket &= \{y\in \llbracket Y\rrbracket: y\in N(x)\},  \label{uno1}  \\
\llbracket X|y\rrbracket &= \{x\in \llbracket{X}\rrbracket: y\in N(x)\}.  \label{due1}
\end{align}
To measure the levels of association and disassociation between UVs $X$ and $Y$, we use an uncertainty function $m_{\mathscr{X}}(.)$  defined over  $\mathscr{X}$, and  $m_{\mathscr{Y}}(.)$ is defined over  $\mathscr{Y}$. 
The definition of feasible set is the same  as the one given in \eqref{eq:feasibleSet}.  
In our channel model, this feasible set   depends on the transition mapping $N$   through \eqref{uno1} and \eqref{due1}.

We can now state the non-stochastic channel coding theorem for channels with transition mapping $N$.
\begin{theorem}\label{thm:channelCodingTheorem_general}   For any  totally bounded, normed metric space $\mathscr{X}$, channel with transition mapping $N$ satisfying \eqref{uno1} and \eqref{due1}, and $0\leq \delta <  m_{\mathscr{Y}}(V_{N})$, 
we have
\begin{equation}\label{eq:channelCodingTheorem_general}
       C_{N}^{\delta}
        = \sup_{ X \in \mathscr{F}_{\tilde{\delta}},\tilde{\delta}\leq\delta/m_{\mathscr{Y}}(\llbracket Y\rrbracket)} I_{\tilde{\delta}/|\llbracket X\rrbracket|}(Y;X) \; \mbox{ bits}.
\end{equation}
\end{theorem}
The proof   is along the same lines as the one of Theorem \ref{thm:channelCodingTheorem} and is omitted.

Theorem~\ref{thm:channelCodingTheorem_general} characterizes the capacity as the supremum of the mutual information over all codebooks in the feasible set.  The following theorem shows that the same characterization is obtained if we optimize the right hand side in \eqref{eq:channelCodingTheorem_general} over all codebooks in the space. It follows that  by Theorem~\ref{thm:channelCodingTheorem_general}, rather than optimizing over all codebooks,  a capacity achieving codebook can  be found within the  smaller class $\cup_{\tilde{\delta}\leq \delta/m_{\mathscr{Y}}(V_N)} \mathscr{F}_{\tilde{\delta}}$ of feasible sets with error at most $\delta/m_{\mathscr{Y}}(V_N)$.

\begin{theorem} \label{cor3}
The $(N,\delta)$-capacity in \eqref{eq:channelCodingTheorem_general} can also be written as
\begin{equation}
       C_{N}^{\delta}
        =\sup_{\substack{ X: \llbracket X\rrbracket \subseteq \mathscr{X},\\\tilde{\delta}\leq\delta/m_{\mathscr{Y}}(\llbracket Y\rrbracket)}} I_{\tilde{\delta}/|\llbracket X\rrbracket|}(Y;X)\mbox{ bits}.
\end{equation}
\end{theorem}
The proof is along the same lines as the one of Theorem \ref{cor} and is omitted.

\section{Capacity of Stationary Memoryless Uncertain Channels}
We now consider the special case of stationary, memoryless, uncertain channels.

Let $\mathscr{X}^{\infty}$ be the space of $\mathscr{X}$-valued discrete-time functions $x:\mathbb{Z}_{> 0}\to \mathscr{X}$, where $\mathbb{Z}_{> 0}$ is the set of positive integers denoting the time step. Let $x(a:b)$ denote the function $x\in \mathscr{X}^{\infty}$ restricted over the time interval $[a,b]$. 
Let $\mathcal{X}\subseteq \mathscr{X}^{\infty}$ be a discrete set   which represents a codebook. Also, let $\mathcal{X}(1:n)=\cup_{x\in \mathcal{X}}x(1:n)$ denote the set of all codewords  up to time $n$, and $\mathcal{X}(n)=\cup_{x\in \mathcal{X}}x(n)$    the set of all codeword symbols in the codebook at time $n$. The codeword symbols can be viewed as the coefficients representing a  continuous signal in an infinite-dimensional space.  For example,  transmitting  one symbol per time step
can be viewed as  transmitting a signal of unit spectral support over time. The perturbation of the signal   at the receiver due to the noise can be described as a  displacement   experienced by  the corresponding codeword symbols  $x(1), x(2), \ldots$. 
To describe this perturbation 
we consider the set-valued map
$N^{\infty}:\mathscr{X}^{\infty}\to \mathscr{Y}^{\infty}$, associating any point in $\mathscr{X}^{\infty}$ to a set in $\mathscr{Y}^{\infty}$.
For any transmitted  codeword $x\in \mathcal{X}\subseteq \mathscr{X}^{\infty}$, the corresponding received   codeword lies in the set 
\begin{equation}
    S_{N^{\infty}}(x)=\{y\in\mathscr{Y}^{\infty}:y\in N^{\infty}(x)\}.
\end{equation}
Additionally, the noise set associated to $x(1:n)\in \mathcal{X}(1:n)$ is 
\begin{equation}
    S_{N^{\infty}}(x(1:n))=\{y(1:n)\in \mathscr{Y}^{n}: y\in N^{\infty}(x)\},
\end{equation}
where 
\begin{equation}
    \mathscr{Y}^n = \underbrace{\mathscr{Y} \times \mathscr{Y} \times  \cdots \times \mathscr{Y}}_n.
\end{equation}
We are now ready to define  stationary, memoryless, uncertain channels.
\begin{definition}
A stationary memoryless uncertain channel is   a transition mapping $N^{\infty}:\mathscr{X}^{\infty}\to \mathscr{Y}^{\infty}$
that can be factorized into   identical terms describing the noise experienced by the codeword symbols. Namely, there
 exists a set-valued map $N:\mathscr{X}\to \mathscr{Y}$ such that
for all $n \in \mathbb{Z}_{>0}$ and $x(1:n)\in \mathscr{X}^{\infty}$, we have
\begin{equation} \label{eq:facnoise}
    S_{N^{\infty}}(x(1:n))=N(x(1))\times\ldots\times N(x(n)).
\end{equation}
\end{definition}
According to the definition, a stationary memoryless uncertain channel  maps the $n$th input symbol into the $n$th output symbol in a way  that does not depend on the symbols at other time steps, and the mapping is the same at all   time steps. 
Since the channel can be characterized by the mapping $N$ instead of $N^{\infty}$, to  simplify the notation in the following we use $S_{N}(.)$ instead of $S_{N^{\infty}}(.)$.

Another important observation is that   the   $\epsilon$-noise channel that we considered in Section~\ref{sec:cap}  may not   admit a factorization like the one in \eqref{eq:facnoise}.  For example, consider  the space to be equipped with the $L^2$ norm, the codeword symbols to represent the coefficients of an orthogonal   representation of a transmitted signal,   and   the noise experienced by any codeword to be within a ball of radius $\epsilon$. In this case,    if a  codeword symbol   is perturbed by a value close to  $\epsilon$    the perturbation of all other symbols must be close to zero. On the other hand,  the general channels considered in Section~\ref{sec:general} can be stationary and memoryless, if the noise   acts on the coefficients in a way that satisfies \eqref{eq:facnoise}.

For stationary memoryless uncertain channels, all received codewords lie in the set $\mathcal{Y}=\cup_{x\in \mathcal{X}} S_{N}(x)$, and the received codewords up to time $n$ lie in the set $\mathcal{Y}(1:n)=\cup_{x\in \mathcal{X}} S_{N}(x(1:n))$. Then, for any $x_{1}(1:n),x_{2}(1:n)\in \mathcal{X}(1:n)$, we let
\begin{equation}
\begin{split}
&e_{N}(x_1(1:n),x_2(1:n))\\
&=\frac{m_{\mathscr{Y}}(S_N(x_1(1:n))\cap S_N(x_2(1:n)))}{m_{\mathscr{Y}}(\mathscr{Y}^n)}.
\end{split}
\end{equation}
We also assume 
without loss of generality that at any time step $n$, the uncertainty associated to the space $\mathscr{Y}^n$ of received codewords is $m_{\mathscr{Y}}(\mathscr{Y}^n)=1$. 
We also let $V_{N}=N(x^*)$, where $x^*=\mbox{argmin}_{x\in \mathscr{X}}m_{\mathscr{Y}}(N(x))$. Thus, $V_N$ is the set corresponding to the minimum uncertainty introduced by the noise mapping at a single time step.  Finally, we let 
\begin{align}
V_N^n  &= \underbrace{V_N \times V_N \times  \cdots \times V_N}_n.
 \end{align}

\begin{definition}{$(N,\delta_n)$-distinguishable codebook.} \label{defdistinguishable} \\
For all $n \in \mathbb{Z}_{>0}$ and $0\leq\delta_n< m_{\mathscr{Y}}(V_N^n)$, a codebook $\mathcal{X}_n=\mathcal{X}(1:n)$
is   $(N,\delta_n)$-distinguishable if for all  $x_{1}(1:n),x_2(1:n)\in \mathcal{X}_n$, we have
\begin{equation}
e_{N}(x_1(1:n),x_2(1:n))\leq \delta_n/|\mathcal{X}_n|.
\end{equation} 
\end{definition}
It immediately follows that for any $(N,\delta_n)$-distinguishable codebook $\mathcal{X}_n$, we have
\begin{equation}
e_N(x(1:n))=\sum_{\substack{x'(1:n) \in \mathcal{X}_n:\\ x'(1:n)\neq x(1:n)}} e_{N}(x(1:n),x'(1:n)) \leq \delta_n,
\end{equation}
so that each codeword in $\mathcal{X}_n$ can be decoded correctly with confidence at least $1-\delta_n$. Definion \ref{defdistinguishable} guarantees even more namely, that the confidence of not confusing any pair of codewords is at least $1-\delta_n/|\mathcal{X}_n|$.

We now associate to any   sequence $\{\delta_n\}$   the largest distinguishable   rate sequence $\{R_{\delta_n}\}$, whose elements represent the largest rates that satisfy that confidence sequence.
\begin{definition}{Largest $\{\delta_n\}$-distinguishable rate sequence.} \label{deflargestachievable}\\ 
For any   sequence $\{\delta_n\}$, the  largest  $\{\delta_n\}$-distinguishable rate sequence $\{R_{\delta_n}\}$ is such that for all $n \in \mathbb{Z}_{>0}$ we have 
\begin{equation} \label{eqachievablerate}
R_{\delta_n} = \sup_{\mathcal{X}_n\in \mathscr{X}^{\delta_n}_{N}(n)}\frac{\log |\mathcal{X}_n|}{n} \mbox{ bits per symbol},
\end{equation} 
where
\begin{equation}
    \mathscr{X}_{N}^{\delta_n}(n)=\{\mathcal{X}_n: \mathcal{X}_n \mbox{ is } (N,\delta_n)\mbox{-distinguishable}\}.
\end{equation}
\end{definition}
We say that any constant rate $R$ that lays below the largest $\{\delta_n\}$-distinguishable rate sequence is $\{\delta_n\}$-distinguishable. Such a $\{\delta_n\}$-distinguishable rate   ensures the existence of a sequence of distinguishable codes that, for all $n \in \mathbb{Z}_{>0}$, have  rate at least $R$ and confidence at least $1-\delta_n$.
\begin{definition}{$\{\delta_n\}$-distinguishable rate.}\\
For any sequence $\{\delta_n\}$, a constant rate $R$ is said to be  $\{\delta_n\}$-distinguishable if for all $n \in \mathbb{Z}_{>0}$, we have
\begin{equation}
R \leq R_{\delta_n}.
    \end{equation}
    \end{definition}
We call    any $\{\delta_n\}$-distinguishable rate $R$ achievable, if $\delta_n \rightarrow 0$ as $n \rightarrow \infty$. An achievable rate $R$ ensures the existence of a sequence of distinguishable codes of rate at least $R$ whose confidence tends to one as $n \rightarrow \infty$. It follows that in this case  we can achieve   communication at rate $R$ with arbitrarily high confidence by choosing a sufficiently large codebook.
 \begin{definition}{Achievable rate.}\\
 A constant rate $R$ is  achievable if there exists a  sequence $\{\delta_n\}$ such that $\delta_n \rightarrow 0$ as $n \rightarrow \infty$, and $R$ is $\{\delta_n\}$-distinguishable.
 \end{definition}

 We now give a first  definition of capacity as the supremum of the $\{\delta_n\}$-distinguishable rates. Using this definition,  transmitting at constant rate below capacity ensures the existence of a sequence of codes that, for all $n \in \mathbb{Z}_{>0}$, have confidence at least $1-\delta_n$.

\begin{definition}\label{defour}{$(N,\{\delta_n\})_*$ capacity.} \\
For any stationary memoryless uncertain channel with transition mapping $N$,  and any given sequence  $\{\delta_n\}$, we let   
\begin{align}
C_{N}(\{ \delta_n \})_*& =
 \sup \{R: R \mbox{ is } \{\delta_n\}\mbox{-distinguishable}\} \\
&= \inf_{n \in \mathbb{Z}_{>0}}    R_{\delta_n}\;\; \mbox{ bits per symbol}.
 \end{align}
\end{definition}

Another definition of capacity arises if rather than the largest lower bound to the sequence of rates one considers the least upper bound  for which we can transmit satisfying a given confidence sequence. Using this definition, transmitting at constant rate below capacity ensures the existence of a code that satisfies at least one confidence value along the sequence $\{ \delta_n\}$.
 
\begin{definition}\label{defive}{$(N,\{\delta_n\})^*$ capacity.} \\
 For any stationary memoryless uncertain channel with transition mapping $N$,  and any given sequence  $\{\delta_n\}$,
we define   
\begin{align}
 C_{N}(\{ \delta_n \})^*& = \sup_{n \in \mathbb{Z}_{>0}}   R_{\delta_n}\;\; \mbox{ bits per symbol}.
 \end{align}
\end{definition}

Definitions~\ref{defour} and \ref{defive} lead to  non-stochastic analogues of Shannon's probabilistic  and zero-error capacities, respectively.

First,  consider Definition~\ref{defour} and take
 the supremum of the achievable rates, rather than the supremum of the $\{\delta_n\}$-distinguishable rates. This  means that we can pick any confidence sequence such that $\delta_n$ tends to zero as $n \rightarrow \infty$. In this way, we obtain the non-stochastic analogue of Shannon's probabilistic capacity, where $\delta_n$ plays the role of the probability of error and the capacity is the largest rate that can be achieved by a sequence of codebooks  with an arbitrarily high confidence level. 

\begin{definition}\label{defS}{$(N,\{\downarrow 0\})_*$ capacity.} \\
 For 
 any stationary memoryless uncertain channel with transition mapping $N$,   
we define the $(N,\{ \downarrow 0\})_*$  capacity   as 
\begin{align}
 C_{N}(\{\downarrow 0\})_* &= \sup \{R: R \mbox{ is  achievable}\} \\
 &=\sup_{\{\delta_n\}:\delta_n=o(1)} C_N(\{\delta_n\})_*.
 \end{align}
\end{definition}

Next, consider Definition~\ref{defive} in the case $\{\delta_n\}$ is  a constant sequence, namely for all $n \in \mathbb{Z}_{>0}$ we have   $\delta_n=\delta \geq 0$. In this case, transmitting below capacity ensures  the existence of   a finite-length code that has confidence at least $1-\delta$. 

\begin{definition}\label{defive1}{$(N,\delta)^*$ capacity.} \\
 For any stationary memoryless uncertain channel with transition mapping $N$,  and any   sequence  $\{\delta_n\}$ such that for all $n \in \mathbb{Z}_{>0}$ we have $\delta_n= \delta \geq 0$,
we define   
\begin{align}
 C_{N}^{\delta*}& = \sup_{n \in \mathbb{Z}_{>0}}    R_{\delta_n}\;\; \mbox{ bits per symbol}.
 \end{align}
\end{definition}

Letting $\delta=0$, we obtain the zero-error capacity. In this case,    below capacity there exists at a   code with which we can transmit with full confidence. 

  We point out the key difference between Definitions~\ref{defS} and \ref{defive1}. Transmitting below the $(N, \{\downarrow 0\})_*$ capacity, allows to achieve arbitrarily high confidence   by increasing the codeword size.  In contrast, transmitting below the $(N,\delta)^*$ capacity, ensures the existence of a fixed codebook  that has confidence  at least $1-\delta$.

We now relate our notions of capacity to the \emph{mutual information rate} between transmitted and received { codewords}. 
Let $X$ be the UV corresponding to the transmitted { codeword}. This is a  map $X:\mathscr{X}^\infty\to\mathcal{X}$ and $\llbracket X\rrbracket=\mathcal{X}\subseteq \mathscr{X}^\infty$.
Restricting this map  to a finite time $n\in \mathbb{Z}_{> 0}$ yields another UV $X(n)$ and $\llbracket X(n)\rrbracket=\mathcal{X}(n) \subseteq \mathscr{X}$.
Likewise, a codebook segment is an UV $X(a:b)=\{X(n)\}_{a \leq n \leq b}$, of marginal range
\begin{equation}
   \llbracket X(a:b)\rrbracket= \mathcal{X}(a:b)\subseteq \mathscr{X}^{b-a+1}.
\end{equation}
Likewise, let $Y$ be the UV corresponding to the received {codeword}. It is a map $Y:\mathscr{Y}^\infty\to \mathcal{Y}$ and $\llbracket Y\rrbracket=\mathcal{Y}\subseteq \mathscr{Y}^\infty$. $Y(n)$ and $Y(a:b)$ are UVs, and $\llbracket Y(n)\rrbracket = \mathcal{Y}\subseteq \mathscr{Y}^\infty$ and $\llbracket Y(a:b)\rrbracket=\mathcal{Y}(a:b) \subseteq \mathscr{Y}^{b-a+1}$. For a stationary memoryless channel with transition mapping $N$, these UVs are such that for all $n\in \mathbb{Z}_{>0}$, $y(1:n) \in \llbracket Y(1:n)\rrbracket$ and $x(1:n) \in \llbracket X(1:n)\rrbracket$, and we have
\begin{equation}
\begin{split}
     \llbracket Y(1:n)|x(1:n)\rrbracket =&\{y(1:n)\in \llbracket Y(1:n)\rrbracket:\\
     & y(1:n)\in S_{N}(x(1:n))\},
\end{split}\label{eq:statcond1}
\end{equation}
\begin{equation}\label{eq:statcond2}
\begin{split}
     \llbracket X(1:n)|y(1:n)\rrbracket =&\{x(1:n)\in \llbracket X(1:n)\rrbracket:\\
     & y(1:n)\in S_{N}(x(1:n))\}.
\end{split}
\end{equation}

Now, we define the largest  $\delta_n$-\emph{mutual information rate}  as the supremum mutual information per unit-symbol  transmission that  a codeword $X(1:n)$  can provide about $Y(1:n)$ with confidence at least $1-\delta_n/|\llbracket X(1:n)\rrbracket|$.

\begin{definition}{Largest $\delta_n$-information rate.}\\
For all  $n \in \mathbb{Z}_{>0}$, the largest    $\delta_n$-information rate from { $X(1:n)$ to $Y(1:n)$} is 
\begin{equation} \label{depends}
    R^I_{\delta_n}=\sup_{\substack{X(1:n):\llbracket X
    (1:n)\rrbracket\subseteq \mathscr{X}^n, \\\tilde{\delta}\leq \delta_n/m_{\mathscr{Y}}(\llbracket Y(1:n)\rrbracket)}} \frac{I_{\tilde{\delta}/|\llbracket X(1:n)\rrbracket|}(Y(1:n); X(1:n))}{n}.
\end{equation}
\end{definition}
We let the feasible set at time $n$ be
\begin{equation}\label{eq:feasibleSetforRate}
\begin{split}
    \mathscr{F}_{\delta}(n)&= \{X(1:n): \llbracket X(1:n)\rrbracket\subseteq \mathscr{X}^{n},\mbox{ and either }\\
    &( X(1:n), Y(1:n))\stackrel{d}{\leftrightarrow}(0,{{\delta}}/|\llbracket X(1:n)\rrbracket|)\mbox{ or}\\
     &  ( X(1:n), Y(1:n))\stackrel{a}{\leftrightarrow}(1,{{\delta}}/|\llbracket X(1:n)\rrbracket|) \}. 
\end{split}
\end{equation}
In the following theorem  we   establish the relationship between $R_{\delta_n}$ and $R^{I}_{\delta_n}$. 

\begin{theorem}\label{thm:relBwRatenInfoRate}
For any totally bounded, normed metric space $\mathscr{X}$, disrete-time space $\mathscr{X}^\infty$, stationary memoryless uncertain channel with transition mapping $N$ satisfying \eqref{eq:statcond1} and \eqref{eq:statcond2}, and sequence  $\{\delta_n\}$ such that for all $n \in \mathbb{Z}_{>0}$ we have $0\leq\delta_n<m_{\mathscr{Y}}(V_{N}^n)$,   we have
\begin{equation}\label{eq:optimizationAtEachTimeN}
    R_{\delta_n}= \sup_{ \substack{ X(1:n)\in \mathscr{F}_{\tilde \delta}(n),\\ \tilde{\delta}\leq \delta_n/m_{\mathscr{Y}}(\llbracket Y(1:n)\rrbracket)}}\frac{I_{\tilde{\delta}/|\llbracket X(1:n)\rrbracket|}(Y(1:n); X(1:n))}{n}.
\end{equation}
We also have
\begin{equation}\label{eq:RateNInfoRate}
    R_{\delta_n}=R^I_{\delta_n}.
\end{equation}
\end{theorem}
\begin{proof}
The proof of the theorem is similar to the one of Theorem \ref{thm:channelCodingTheorem} and is given in   Appendix \ref{sec:Theorem9}.
\end{proof}

The following coding theorem is now an immediate consequence of Theorem~\ref{thm:relBwRatenInfoRate}  and of our capacity definitions. 
\begin{theorem}\label{thm:capacityandinformationrate}
For any totally bounded, normed metric space $\mathscr{X}$, disrete-time space $\mathscr{X}^\infty$, stationary memoryless uncertain channel with transition mapping $N$ satisfying \eqref{eq:statcond1} and \eqref{eq:statcond2}, and sequence  $\{\delta_n\}$ such that for all $n  \in \mathbb{Z}_{>0}$, $0\leq\delta_n<m_{\mathscr{Y}}(V_{N}^n)$ and $0\leq\delta<m_{\mathscr{Y}}(V_{N}^n)$,   we have

\begin{align}
1) \;\;\;\; &
    {C}_{N}(\{\delta_n\})_*=\inf_{n \in \mathbb{Z}_{>0}} R^{I}_{\delta_n},  \label{eq:ChannelCodingTheorem2} \\
 2) \; \; \; \; &
     {C}_{N}(\{\delta_n\})^*=\sup_{n \in \mathbb{Z}_{>0}} R^{I}_{\delta_n}, \label{eq:ChannelCodingTheorem3} \\
     3) \;\;\;\; &{C}_{N}(\{\downarrow 0\})_*=\sup_{\{\delta_n\}:\delta_n=o(1)}\inf_{n \in \mathbb{Z}_{>0}} R^{I}_{\delta_n}, \label{eq:ChannelCodingTheorem4} \\
  4) \;\;\;\; &{C}_{N}^{\delta*}=\sup_{n \in \mathbb{Z}_{>0}} R^{I}_{\delta_n}: \forall  n \in \mathbb{Z}_{>0}, \delta_n=\delta. \label{eq:ChannelCodingTheorem5}
\end{align}

\end{theorem}

Theorem~\ref{thm:capacityandinformationrate}  provides   multi-letter expressions of capacity,  since $R^I_{\delta_n}$  depends on $I_{\tilde\delta/|\llbracket X(1:n)\rrbracket|}(Y(1:n);X(1:n))$ according to \eqref{depends}. Next, we establish conditions on the uncertainty function, confidence sequence, and class of stationary, memoryless channels leading to the factorization of the mutual information and to single-letter expressions.

\subsection{Factorization of the Mutual Information}
 To obtain sufficient conditions for the factorization of the mutual information  we need to assume to work with a \emph{product uncertainty function}. 
\begin{assumption} (Product uncertainty function). \label{assumption:productrule}
The uncertainty function  
of a cartesian product of $n$ sets can be factorized in the product of its terms,  namely for any $n \in \mathbb{Z}_{>0}$ and $\mathscr{S} \subseteq \llbracket Y\rrbracket$ such that
\begin{equation}
   \mathscr{S}= \mathscr{S}_1\times\mathscr{S}_2\times\ldots\times\mathscr{S}_n,
\end{equation}
we have 
\begin{equation}
    m_{\mathscr{Y}}(\mathscr{S})= m_{\mathscr{Y}}(\mathscr{S}_1)\times  m_{\mathscr{Y}}(\mathscr{S}_2)\times\ldots\times m_{\mathscr{Y}}(\mathscr{S}_n).
\end{equation}
\end{assumption}

We also need to assume that the product uncertainty function   satisfies a union bound.

\begin{assumption}(Union bound).\label{assumption:triangleInequality}
For all $\mathscr{S}_1,\mathscr{S}_2\subseteq \llbracket Y\rrbracket$, we have
\begin{equation}
     m_{\mathscr{Y}}(\mathscr{S}_1\cup\mathscr{S}_2)\leq m_{\mathscr{Y}}(\mathscr{S}_1)+m_{\mathscr{Y}}(\mathscr{S}_2).
\end{equation}
\end{assumption}

Before stating the main result of this section, we prove the following useful lemma.


\begin{lemma}\label{thm:tensorization}
Let $X(1:n)$ and $Y(1:n)$ be two   UVs  such that 
\begin{equation}\label{eq:reference1}
    \llbracket X(1:n)\rrbracket=\llbracket X(1)\rrbracket\times\llbracket X(2)\rrbracket\ldots\times \llbracket X(n)\rrbracket,
\end{equation}
and for all $x(1:n)\in \llbracket X(1:n)\rrbracket$, we have
\begin{equation}\label{eq:reference3}
    \llbracket Y(1:n)|x(1:n)\rrbracket=\llbracket Y(1)|x(1)\rrbracket\times\ldots \llbracket Y(n)|x(n)\rrbracket.
\end{equation}
Let 
\begin{equation}\label{eq:rangeDeltaCheck}
    0\leq\delta< \min_{\substack{1\leq i\leq n, \\ x(i)\in \llbracket X(i)\rrbracket}}m_{\mathscr{Y}}(\llbracket Y(i)|x(i)\rrbracket).
\end{equation}
Finally, let either 
\begin{align}
(X(1:n),Y(1;n))\stackrel{d}{\leftrightarrow}(0,\delta^n), & \mbox{ or}\\ (X(1:n),Y(1;n))\stackrel{a}{\leftrightarrow}(1,\delta^n).
\end{align}
Under Assumption \ref{assumption:productrule},
we have:
\begin{enumerate}
    \item The cartesian product $\prod_{i=1}^{n}\llbracket Y(i)|X(i)\rrbracket^*_{\delta}$ is a covering of $\llbracket Y(1:n)\rrbracket$.
    \item Every $\mathscr{S}\in\prod_{i=1}^{n}\llbracket Y(i)|X(i)\rrbracket^*_{\delta}$ is  $\delta^n$-connected and 
contains at least one singly $\delta^n$-connected set of the form $\llbracket Y(1:n)|x(1:n)\rrbracket$. 
\item For every singly $\delta^n$-connected set of the form $\llbracket Y(1:n)|x(1:n)\rrbracket$, there exist a set in  $\prod_{i=1}^{n}\llbracket Y(i)|X(i)\rrbracket^*_{\delta}$ containing it, namely for all $x(1:n)\in\llbracket X(1:n)\rrbracket$, there exists a set $\mathscr{S}\in \prod_{i=1}^{n}\llbracket Y(i)|X(i)\rrbracket^*_{\delta}$ such that $\llbracket Y(1:n)|x(1:n)\rrbracket\subseteq \mathscr{S}$.
\item For all $\mathscr{S}_1,\mathscr{S}_2\in \prod_{i=1}^{n}\llbracket Y(i)|X(i)\rrbracket^*_{\delta}$, we have
\begin{equation}
    \frac{m_{\mathscr{Y}}(\mathscr{S}_1\cap \mathscr{S}_2)}{m_{\mathscr{Y}}(\llbracket Y(1:n)\rrbracket)}\leq \delta(\hat{\delta}(n))^{n-1},
\end{equation}
where 
\begin{equation}\label{eq:upperBoundDeltahat}
    \hat{\delta}(n)=\max_{\substack{1\leq i\leq n,\\ \mathscr{S}\in \llbracket Y(i)|X(i)\rrbracket^*_{\delta}}} \frac{m_{\mathscr{Y}}(\mathscr{S})}{m_{\mathscr{Y}}(\llbracket Y(i)\rrbracket)}. 
\end{equation}
\end{enumerate}
\end{lemma}
\begin{proof}
The proof  is given in   Appendix \ref{appendix:lemma2Proof}.
\end{proof}

Under    Assumption \ref{assumption:productrule} and Assumption \ref{assumption:triangleInequality}, given two UVs that can be written in Cartesian product form and that are either associated at level $(0,\delta^n)$ or disassociated at level $(1, \delta^n)$, we now obtain an upper bound on the mutual information at level $\delta^n$ in terms of  the sum of the mutual information at level $\delta$ of their components.  An analogous result in the stochastic setting states that
the mutual information between two $n$-dimensional random variables $X^n=\{X_1,\ldots,X_n\}$ and $Y^n=\{Y_1,\ldots,Y_n\}$  is at most the sum of the component-wise mutual information, namely
\begin{equation}
    I(X^n;Y^n)\leq\sum_{i=1}^{n}I(X_i;Y_i),
\end{equation}
where $I(X;Y)$ represents the Shannon mutual information between two random variables $X$ and $Y$.
In contrast to   the stochastic setting, here the mutual information is associated to a confidence parameter $\delta^n$ that  is re-scaled to $\delta$ when this is decomposed into the sum of $n$ terms. 
{

\begin{theorem}\label{thm:tensorization2}
Let $X(1:n)$ and $Y(1:n)$ be two   UVs  such that 
\begin{equation}\label{eq:ref4}
    \llbracket X(1:n)\rrbracket=\llbracket X(1)\rrbracket\times\llbracket X(2)\rrbracket\ldots \llbracket X(n)\rrbracket,
\end{equation}
and for all $x(1:n)\in \llbracket X(1:n)\rrbracket$, we have
\begin{equation}\label{eq:ref5}
    \llbracket Y(1:n)|x(1:n)\rrbracket=\llbracket Y(1)|x(1)\rrbracket\times\ldots \llbracket Y(n)|x(n)\rrbracket.
\end{equation}
Also, let 
\begin{equation}\label{eq:rangeDeltaCheck1}
    0\leq\delta < \frac{\min_{{1\leq i\leq n, \\ x(i)\in \llbracket X(i)\rrbracket}}m_{\mathscr{Y}}(\llbracket Y(i)|x(i)\rrbracket)}{\max_{1\leq i\leq n} |\llbracket X(i)\rrbracket|}.
\end{equation} 
{\color{brown}
}
Finally, let either 
\begin{align}\label{eq:disAssoc1}
(X(1:n),Y(1:n))\stackrel{d}{\leftrightarrow}(0,\delta^n), & \mbox{ or}\\ (X(1:n),Y(1:n))\stackrel{a}{\leftrightarrow}(1,\delta^n). \label{eq:disAssoc2}
\end{align}
Under Assumptions \ref{assumption:productrule} and 2,
we have 
\begin{equation}\label{eq:tensorization}
    I_{\delta^n}(Y(1:n);X(1:n))\leq \sum_{i=1}^n I_{\delta}(Y(i);X(i)).
\end{equation}
\end{theorem}
\begin{proof}
First, we will show that for all $\mathscr{S}\in \prod_{i=1}^{n}\llbracket Y(i)|X(i)\rrbracket^*_{\delta}$, there exists a point $x_{\mathscr{S}}(1:n)\in\llbracket X(1:n)\rrbracket$ and a set $\mathscr{D}(\mathscr{S})\in \llbracket Y(1:n)|X(1:n)\rrbracket^*_{\delta^n}$ such that \begin{equation}\label{eq:check1.2}
        \llbracket Y(1:n)|x_{\mathscr{S}}(1:n)\rrbracket\subseteq\mathscr{S},
    \end{equation}
    \begin{equation}\label{eq:check1.3}
        \llbracket Y(1:n)|x_{\mathscr{S}}(1:n)\rrbracket\subseteq\mathscr{D}(\mathscr{S}). 
    \end{equation}
Using this result, we  will then show that 
\begin{equation}\label{eq:TOProve1}
    |\prod_{i=1}^{n}\llbracket Y(i)|X(i)\rrbracket^*_{\delta}|\geq |\llbracket Y(1:n)|X(1:n)\rrbracket^*_{\delta^n}|,
\end{equation}
which immediately implies \eqref{eq:tensorization}. 

Let us begin with the first step. 
{We have
\begin{equation}\label{eq:rangeDeltaCond}
\begin{split}
    \delta&\stackrel{(a)}{<} \frac{\min_{{1\leq i\leq n,  x(i)\in \llbracket X(i)\rrbracket}}m_{\mathscr{Y}}(\llbracket Y(i)|x(i)\rrbracket)}{\max_{1\leq i\leq n} |\llbracket X(i)\rrbracket|},\\
    &\stackrel{(b)}{\leq} \min_{{1\leq i\leq n,  x(i)\in \llbracket X(i)\rrbracket}}m_{\mathscr{Y}}(\llbracket Y(i)|x(i)\rrbracket),
\end{split}
\end{equation}
where $(a)$ follows from \eqref{eq:rangeDeltaCheck1}, and $(b)$ follows from the fact that for all $1\leq i\leq n$, we have $|\llbracket X(i)\rrbracket|\geq 1$. 
}


Now, consider a set $\mathscr{S}\in \prod_{i=1}^{n}\llbracket Y(i)|X(i)\rrbracket^*_{\delta}$. 
Using   \eqref{eq:rangeDeltaCond}, by Lemma   \ref{thm:tensorization} part $2)$ we have that there exist a point $x^\prime(1:n)\in\llbracket X(1:n)\rrbracket$ such that
\begin{equation}\label{eq:check1.4}
    \llbracket Y(1:n)|x^\prime(1:n)\rrbracket\subseteq\mathscr{S}.
\end{equation}
Now, using  \eqref{eq:rangeDeltaCond}, part $1)$ in Lemma \ref{thm:tensorization}, and Definition \ref{defoverlap}, we have 
\begin{equation}\label{eq:covering1}
\begin{split}
     \cup_{\mathscr{S}\in\prod_{i=1}^{n}\llbracket Y(i)|X(i)\rrbracket^*_{\delta}}\mathscr{S}&=\llbracket Y(1:n)\rrbracket,\\
     &=\cup_{\mathscr{D}\in\llbracket Y(1:n)|X(1:n)\rrbracket^*_{\delta^n}}\mathscr{D}.
\end{split}
\end{equation}
Using \eqref{eq:covering1} and  Property 3 of Definition \ref{defoverlap}, there exists a set $\mathscr{D}(x^\prime(1:n))\in \llbracket Y(1:n)|X(1:n)\rrbracket^*_{\delta^n}$ such that \begin{equation}\label{eq:check1.5}
    \llbracket Y(1:n)|x^\prime(1:n)\rrbracket\subseteq\mathscr{D}(x^\prime(1:n)).
\end{equation}
Letting $x_{\mathscr{S}}(1:n)=x^\prime(1:n)$ and $\mathscr{D}(\mathscr{S})=\mathscr{D}(x^\prime(1:n))$ in \eqref{eq:check1.4} and \eqref{eq:check1.5}, we have that \eqref{eq:check1.2} and \eqref{eq:check1.3} follow.

We now proceed with proving \eqref{eq:TOProve1}. 
We distinguish two cases. In the first case, there exists two sets $\mathscr S\in \prod_{i=1}^{n}\llbracket Y(i)|X(i)\rrbracket^*_{\delta}$ and $ \mathscr D_1\in \llbracket Y(1:n)|X(1:n)\rrbracket^*_{\delta^n}$ such that
\begin{equation} \label{eq:satisfying}
    \mathscr D_1\cap \mathscr S\setminus \mathscr{D}(\mathscr{S})\neq\emptyset.
\end{equation} 
In the second case, the sets $\mathscr S$ and $\mathscr{D}_1$ satisfying \eqref{eq:satisfying} do not exist. We will show that the first case is not possible, and in the second case, we have that \eqref{eq:TOProve1} holds.  


To rule out the first case, consider two points 
\begin{equation}
    y_1(1:n)\in \llbracket Y(1:n)|x_{\mathscr{S}}(1:n)\rrbracket \subseteq\mathscr{D}(\mathscr{S})
\end{equation}
and 
\begin{equation}
    y_2(1:n)\in \mathscr D_1\cap \mathscr S\setminus \mathscr{D}(\mathscr{S}).
\end{equation}
If $(X(1:n),Y(1:n)\stackrel{a}{\leftrightarrow}(1,\delta^n)$, then 
{
we have
\begin{equation}\label{eq:rangeOfDelta7}
\begin{split}
&\delta^n |\llbracket X(1:n)\rrbracket|\\
&\stackrel{(a)}{<} \bigg( \frac{\min_{{1\leq i\leq n,  x(i)\in \llbracket X(i)\rrbracket}}m_{\mathscr{Y}}(\llbracket Y(i)|x(i)\rrbracket)}{\max_{1\leq i\leq n} |\llbracket X(i)\rrbracket|}\bigg)^n |\llbracket X(1:n)\rrbracket|,\\
&\stackrel{(b)}{\leq} \bigg(\min_{\substack{1\leq i\leq n, \\ x(i)\in \llbracket X(i)\rrbracket}} m_{\mathscr{Y}}(\llbracket Y(i)|x(i)\rrbracket)\bigg)^n,\\
&\stackrel{(c)}{\leq} \min_{x(1:n)\in\llbracket X(1:n)\rrbracket} m_{\mathscr{Y}}(\llbracket Y(1:n)|x(1:n)\rrbracket),\\
&\stackrel{(d)}{\leq}\frac{\min_{x(1:n)\in\llbracket X(1:n)\rrbracket} m_{\mathscr{Y}}(\llbracket Y(1:n)|x(1:n)\rrbracket)}{m_{\mathscr{Y}}(\llbracket Y(1:n)\rrbracket)},
\end{split}
\end{equation}
where $(a)$ follows from \eqref{eq:rangeDeltaCheck1}, $(b)$
follows from \eqref{eq:ref4}, and the fact that 
\begin{equation}
    |\llbracket X(1)\rrbracket\times\llbracket X(2)\rrbracket\ldots \llbracket X(n)\rrbracket|\leq (\max_{1\leq i\leq n}|\llbracket X(i)\rrbracket|)^n,
\end{equation}
 $(c)$ follows from \eqref{eq:ref5} and Assumption \ref{assumption:productrule}, and $(d)$ follows from \eqref{eq:strongtransitivity}, and the facts that $\llbracket Y(1:n)\rrbracket\subseteq \mathscr{Y}^n$ and $m_\mathscr{Y}(\mathscr{Y}^n)=1$.  
 Combining \eqref{eq:rangeOfDelta7}, Assumption \ref{assumption:triangleInequality} and Lemma \ref{lemma:uniquenessProperty} in   Appendix \ref{sec:AuxResult}, we have that there exists a point $y(1:n)\in \llbracket Y(1:n)|x_{\mathscr{S}}(1:n)\rrbracket$ such that for all $\llbracket Y(1:n)|x(1:n)\rrbracket\in \llbracket Y(1:n)|X(1:n)\rrbracket\setminus \{\llbracket Y(1:n)|x_{\mathscr{S}}(1:n)\rrbracket\}$, we have.
\begin{equation}\label{eq:uni1}
    y(1:n)\notin\llbracket Y(1:n)|x(1:n)\rrbracket. 
\end{equation}
Without loss of generality, let 
\begin{equation}
    y_1(1:n)=y(1:n).
\end{equation}
It now follows that  
}
$y_1(1:n)$ and $y_2(1:n)$ cannot be $\delta^n$-connected. This follows because
\begin{equation}
    y_1(1:n)\in \llbracket Y(1:n)|x_{\mathscr{S}}(1:n)\rrbracket,
\end{equation}
\begin{equation}
    y_2(1:n)\not\in \llbracket Y(1:n)|x_{\mathscr{S}}(1:n)\rrbracket,
\end{equation}
{\eqref{eq:uni1}} and  $(X(1:n),Y(1:n)\stackrel{a}{\leftrightarrow}(1,\delta^n)$, so that there does not exist a sequence $\{\llbracket Y(1:n)|x_{i}(1:n)\rrbracket\}_{i=1}^N$ such that for all $1<i\leq N$
\begin{equation}
    \frac{m_{\mathscr{Y}}(\llbracket Y(1:n)|x_{i}(1:n)\rrbracket\cap \llbracket Y(1:n)|x_{i-1}(1:n)\rrbracket)}{m_{\mathscr{Y}}(\llbracket Y(1:n)\rrbracket)}>\delta^n.
\end{equation}
On the other hand, if $(X(1:n),Y(1:n)\stackrel{d}{\leftrightarrow}(0,\delta^n)$, then using Theorem \ref{lemma:overlap}, we have that $\llbracket Y(1:n)|X(1:n)\rrbracket^*_{\delta^n}$ is a $\delta^n$-isolated partition. Thus, $y_1(1:n)$ and $y_2(1:n)$ are not $\delta^n$-connected, since $y_1(1:n)\in\mathscr{D}(\mathscr{S})$ and $y_2(1:n)\in \mathscr D_1\cap \mathscr S\setminus \mathscr{D}(\mathscr{S})$. However, since $y_1(1:n), y_2(1:n)\in\mathscr{S}$ and $\mathscr{S}$ is $\delta^n$-connected using   \eqref{eq:rangeDeltaCond} and $2)$ in Lemma \ref{thm:tensorization}, we have that $y_1(1:n)\stackrel{\delta^n}{\leftrightsquigarrow}y_2(1:n)$.  This contradiction  implies that $\mathscr{D}_1$ and $\mathscr{S}$ do not exist. 

In the second case, if $\mathscr S$ and $\mathscr{D}_1$ do not exist, then for all $\mathscr{S}^\prime\in \prod_{i=1}^{n}\llbracket Y(i)|X(i)\rrbracket^*_{\delta}$ and $\mathscr{D}^\prime\in \llbracket Y(1:n)|X(1:n)\rrbracket^*_{\delta^n}$, we have 
\begin{equation}\label{eq:check1234}
    \mathscr{D}^\prime\cap \mathscr{S}^\prime\setminus\mathscr{D}(\mathscr{S}^\prime)=\emptyset,
\end{equation}
which implies that 
\begin{align}
    \mathscr{S}^\prime
    &\stackrel{(a)}{=}\cup_{\mathscr{D}^\prime\in \llbracket Y(1:n)|X(1:n)\rrbracket^*_{\delta^n}} (\mathscr{S}^\prime\cap \mathscr{D}^\prime),\nonumber \\
    &\stackrel{(b)}{\subseteq}  \cup_{\mathscr{D}^\prime\in \llbracket Y(1:n)|X(1:n)\rrbracket^*_{\delta^n}}\Big(
    \mathscr{D}(\mathscr{S}^\prime)\cup(\mathscr{S}^\prime\cap \mathscr{D}^\prime\setminus \mathscr{D}(\mathscr{S}^\prime)\Big),\nonumber \\
    &=\mathscr{D}(\mathscr{S}^\prime)\cup_{\mathscr{D}^\prime\in \llbracket Y(1:n)|X(1:n)\rrbracket^*_{\delta^n}}(\mathscr{S}^\prime\cap \mathscr{D}^\prime\setminus \mathscr{D}(\mathscr{S}^\prime)),\nonumber \\
    &\stackrel{(c)}{=} \mathscr{D}(\mathscr{S}^\prime),\label{eq:CheckForContra}
\end{align}
where $(a)$ follows from \eqref{eq:covering1}, $(b)$ follows from the trivial fact that for any three sets $\mathscr{A}$, $\mathscr{B}$ and $\mathscr{C}$, 
\begin{equation}
    \mathscr{A}\cap \mathscr{B}\subseteq \mathscr{C}\cup (\mathscr{A}\cap \mathscr{B}\setminus \mathscr{C}),
\end{equation}
and $(c)$ follows from \eqref{eq:check1234}. Combining \eqref{eq:CheckForContra} and \eqref{eq:covering1}, we have that
\begin{equation}
    |\prod_{i=1}^{n}\llbracket Y(i)|X(i)\rrbracket^*_{\delta}|\geq |\llbracket Y(1:n)|X(1:n)\rrbracket^*_{\delta^n}|.
\end{equation}
The statement of the theorem now follows.
\end{proof}
}
The following corollary shows that the bound in Theorem \ref{thm:tensorization2} is tight in the  zero-error case. 

\begin{corollary}\label{corr:tensorization}
Let $X(1:n)$ and $Y(1:n)$ satisfy  (\ref{eq:ref4}) and \eqref{eq:ref5}. Under Assumptions \ref{assumption:productrule} and \ref{assumption:triangleInequality}, we have that 
\begin{equation}\label{eq:tensorization_zero}
    I_{0}(Y(1:n);X(1:n))= \sum_{i=1}^n I_{0}(Y(i);X(i)). 
\end{equation}
\end{corollary}
\begin{proof}
The proof is along the same lines as the one  of Theorem \ref{thm:tensorization2}. For all $X(1:n)$ and $Y(1:n)$, if $\mathcal{A}(Y;X)=\emptyset$, then
\begin{equation}
    (X(1:n),Y(1:n))\stackrel{a}{\leftrightarrow}(1,0),
\end{equation}
otherwise 
\begin{equation}
    (X(1:n),Y(1:n))\stackrel{d}{\leftrightarrow}(0,0). 
\end{equation}
Hence, either \eqref{eq:disAssoc1} or \eqref{eq:disAssoc2} holds for $\delta=0$. Now, by replacing $\delta=0$ in $1)-4)$ of Lemma \ref{thm:tensorization}, we have that $\prod_{i=1}^{n}\llbracket Y(i)|X(i)\rrbracket_0^*$ satisfies all the properties of a $0$-overlap family. Combining this fact and Theorem \ref{thm:tensorization2}, the statement of the corollary follows.
\end{proof}


\subsection{Single letter expressions}
We are now ready to present sufficient conditions leading to single-letter expressions for   $C(\{\delta_n\})_*$, ${C}_{N}^{\delta*}$, $C(\{\delta_n\})^*$ and ${C}_{N}(\{\downarrow 0\})_*$. Under these conditions, the  multi dimensional optimization problem of searching for a codebook that achieves   capacity over  a time interval of size $n$ can be reduced to searching for a codebook over a single time step.


First, we start with the single-letter expression for  $C(\{\delta_n\})^*$. 

\begin{theorem}\label{thm:sufficientConditiondeltaerror}  
For any  stationary  memoryless uncertain channel $N$  and for any $0\leq \delta_1<m_{\mathscr{Y}}(
V_{N})$,  let  $\bar{X} \in \mathscr{F}_{\bar {\delta}}(1)$  be an UV  over one time step associated with a one-dimensional codebook  that achieves the   capacity $C_N(\{\delta_1\})^*=C_N(\{\delta_1\})_* =R_{\delta_1}$, 
and let $\bar{Y}$  be the UV  corresponding to the received { codeword}, namely   
\begin{align}\label{eq:assumptionOnChannel}
C_N(\{\delta_1\})^* &= I_{\bar\delta/|\llbracket \bar X\rrbracket|}(\bar Y;\bar X)  \nonumber \\
&= \sup_{  \substack{X(1)\in \mathscr{F}_{\tilde \delta}(1): \\\tilde{\delta}\leq \delta_1/m_{\mathscr{Y}}(\llbracket Y(1)\rrbracket)}}{I_{\tilde{\delta}/|\llbracket X(1)\rrbracket|}(Y(1); X(1))}.
\end{align}
If for all one-dimensional codewords  $x\in \mathscr{X}\setminus\llbracket \bar{X} \rrbracket$  there exists a set $\mathscr{S}\in \llbracket \bar Y|\bar{X}\rrbracket_{\bar\delta / | \llbracket \bar X\rrbracket|}^*$ such that the uncertainty region $\llbracket  Y|x\rrbracket \subseteq \mathscr{S}$, { $\Bar{\delta}(1+1/|\llbracket \bar{X}\rrbracket|)\leq \delta_1/m_{\mathscr{Y}}(\llbracket \bar{Y}\rrbracket)$} ,  and for all $n>1$  we have $0\leq\delta_n\leq (\bar\delta m_{\mathscr{Y}}(V_N)/|\llbracket \bar X\rrbracket|)^n$,
then under   Assumptions~\ref{assumption:productrule} and \ref{assumption:triangleInequality}   we have that the $n$-dimensional capacity\\
\begin{equation}\label{eq:singleLeterSupDefin}
    C_{N}(\{\delta_n\})^* { = I_{\bar\delta/|\llbracket \Bar{X}\rrbracket|}(\bar{Y};\bar{X})}.
\end{equation} 
\end{theorem}
\begin{proof}
Let 
\begin{equation}\label{eq:Def1}
    \llbracket \bar X(1:n)\rrbracket=\underbrace{ \llbracket \bar X\rrbracket \times  \cdots \times \llbracket \bar X\rrbracket}_n,
\end{equation}
and 
\begin{equation}\label{eq:Def2}
    \llbracket \bar Y(1:n)\rrbracket=\underbrace{ \llbracket \bar Y\rrbracket \times  \cdots \times \llbracket \bar Y\rrbracket}_n. 
\end{equation}
For all $n>0$, we have
\begin{equation}\label{eq:CondUPPDelta}
\begin{split}
    \delta_n&\leq  \bigg(\frac{\bar\delta m_{\mathscr{Y}}(V_N)}{|\llbracket \bar X\rrbracket|}\bigg)^n,\\
    &\stackrel{(a)}{\leq}  \bigg(\frac{\delta_1 m_{\mathscr{Y}}(V_N)}{|\llbracket \bar X\rrbracket| m_{\mathscr{Y}}(\llbracket\bar{Y}\rrbracket)}\bigg)^n,\\
    &\stackrel{(b)}{\leq} \bigg(\frac{\delta_1}{|\llbracket \bar X\rrbracket|}\bigg)^n,\\
    &\stackrel{(c)}{<} (m_{\mathscr{Y}}(V_N))^n,\\
    &\stackrel{(d)}{=}m_{\mathscr{Y}}(V^n_N),
\end{split}
\end{equation}
where $(a)$ follows from the fact that $\bar{\delta}\leq \delta_1/ m_{\mathscr{Y}}(\llbracket\bar{Y}\rrbracket)$, $(b)$ follows from $m_{\mathscr{Y}}(V_N)\leq m_{\mathscr{Y}}(\llbracket\bar{Y}\rrbracket)$, $(c)$ follows from $\delta_1<m_{\mathscr{Y}}(V_N)$ and $|\llbracket\bar{X}\rrbracket|\geq 1$, and $(d)$ follows from Assumption \ref{assumption:productrule}. 


We now proceed in three steps.
First, using \eqref{eq:CondUPPDelta} and Theorem \ref{thm:relBwRatenInfoRate}, we have
\begin{equation}\label{eq:LowBoundSingleLetter}
\begin{split}
      C_{N}(\{\delta_n\})^*&=\sup_{n \in \mathbb{Z}_{>0}} R_{\delta_n}\\
      &\geq R_{\delta_1}= R^I_{\delta_1}=I_{\bar{\delta}/|\llbracket\bar{X}\rrbracket|}(\Bar{Y};\bar{X}).
\end{split}
\end{equation}
Second, we will show that for all $n \in \mathbb{Z}_{>0}$, we have
\begin{equation}\label{eq:ToProve}
\begin{split}
     &\sup_{\substack{X(1:n):\llbracket X(1:n)\rrbracket\subseteq \mathscr{X}^n,\\ \tilde{\delta}\leq \delta_n/m_{\mathscr{Y}}(\llbracket Y(1:n)\rrbracket)}} I_{\tilde{\delta}/|\llbracket X(1:n)\rrbracket|}(Y(1:n);X(1:n))\\
     &\qquad\leq n I_{\bar\delta/|\llbracket \Bar{X}\rrbracket|}(\bar{Y}; \bar{X}).
\end{split}
\end{equation}
Finally, using \eqref{eq:CondUPPDelta}, Theorem \ref{thm:relBwRatenInfoRate} and  \eqref{eq:ToProve},
 for all $n \in \mathbb{Z}_{>0}$, we have that
\begin{equation}
R_{\delta_n}=R^{I}_{\delta_n}\leq 
 I_{\bar{\delta}/|\llbracket\bar{X}\rrbracket|}(\Bar{Y};\bar{X}),
\end{equation}
which implies 
\begin{equation}\label{eq:UppBoundSingleLetter}
    C_{N}(\{\delta_n\})^*=\sup_{n \in \mathbb{Z}_{>0}} R_{\delta_n}\leq  I_{\bar{\delta}/|\llbracket\bar{X}\rrbracket|}(\Bar{Y};\bar{X}).
\end{equation}
Using \eqref{eq:LowBoundSingleLetter} and \eqref{eq:UppBoundSingleLetter}, the result \eqref{eq:singleLeterSupDefin} follows.

Now, we only need to prove \eqref{eq:ToProve}. We will prove this by contradiction.
Consider an UV $X(1:n)$ and
\begin{equation}\label{eq:rangeDelta}
    {\delta}^\prime\leq\delta_n/m_{\mathscr{Y}}(\llbracket Y(1:n)\rrbracket),
\end{equation}
 such that
\begin{equation}\label{eq:ProveContradiction}
\begin{split}
        &|\llbracket {Y}(1:n)|{X}(1:n)\rrbracket_{{\delta}^\prime/|\llbracket {X}(1:n)\rrbracket|} ^{*}|\\
        &> |\prod_{i=1}^n\llbracket \bar Y|\bar X\rrbracket_{\bar\delta/|\llbracket\bar X\rrbracket|}^*|.
\end{split}
\end{equation} 
We will show that \eqref{eq:ProveContradiction} cannot hold using the following four claims, whose proofs appear in  Appendix \ref{sec:5claims}.
\begin{itemize}
\item \textbf{Claim 1:} If \eqref{eq:ProveContradiction} holds, then there exists two UVs $\tilde{X}(1:n)$ and $\tilde{Y}(1:n)$ such that letting 
    \begin{equation}
\begin{split}
     \tilde{\delta}=\frac{\delta^\prime m_{\mathscr{Y}}(\llbracket Y(1:n)\rrbracket)}{m_{\mathscr{Y}}(\llbracket \tilde{Y}(1:n)\rrbracket)},
\end{split}
\end{equation}
we have 
\begin{equation}
    \tilde{\delta}\leq \frac{\delta_n}{m_{\mathscr{Y}}(\llbracket \tilde{Y}(1:n)\rrbracket)},
\end{equation}
\begin{equation}
     (\tilde X(1:n),\tilde{Y}(1:n))\stackrel{a}{\leftrightarrow}(1,\tilde{\delta} /|\llbracket \tilde X(1:n)\rrbracket|),
 \end{equation} 
 and
 \begin{equation}
 \begin{split}
     &|\llbracket \tilde Y(1:n)|\tilde{X}(1:n)\rrbracket^{*}_{\tilde{\delta}/|\llbracket \tilde X(1:n)\rrbracket|}|{>} |\prod_{i=1}^n\llbracket \bar Y|\bar X\rrbracket_{\bar\delta/|\llbracket\bar X\rrbracket|}^*|.
 \end{split}
 \end{equation}
   \item \textbf{Claim 2:} For all $\tilde{x}(1:n)\in\llbracket\tilde{X}(1:n)\rrbracket$, there exists a set $\mathscr{S}\in \prod_{i=1}^n\llbracket \bar{Y}|\bar{X}\rrbracket_{\bar\delta/|\llbracket \bar X\rrbracket|}^*$ such that
   \begin{equation}
       \llbracket\tilde Y(1:n)|\tilde x(1:n)\rrbracket\subseteq \mathscr{S}.
   \end{equation}
   \item \textbf{Claim 3:} Using Claims 1 and 2,  there exists a set $\mathscr{S}\in \prod_{i=1}^n\llbracket \bar Y|\bar X\rrbracket_{\bar\delta/|\llbracket\bar X\rrbracket|}^*$  and two points $\tilde{x}_1(1:n),\tilde{x}_2(1:n)\in\llbracket\tilde{X}(1:n)\rrbracket$ such that
   \begin{equation}
       \llbracket \tilde{Y}(1:n)|\tilde{x}_1(1:n)\rrbracket, \llbracket \tilde{Y}(1:n)|\tilde{x}_2(1:n)\rrbracket\subset \mathscr{S}.
   \end{equation}
   Also,
   there exists a $1\leq i^*\leq n$ such that
   \begin{equation}
   \begin{split}
       &\frac{m_{\mathscr{Y}}(\llbracket \tilde{Y}(i^*)|\tilde{x}_1(i^*)\rrbracket\cap \llbracket \tilde{Y}(i^*)|\tilde{x}_2(i^*)\rrbracket)}{(m_{\mathscr{Y}}(\llbracket\tilde{Y}(1:n)\rrbracket))^{1/n}}\\
      & \leq \frac{\bar{\delta}}{|\llbracket\bar{X}\rrbracket|}. 
   \end{split}
   \end{equation}
   \item \textbf{Claim 4:} Using   Claim 3, we have that there exist  two UVs $X^\prime$ and $Y^\prime$, and $\delta^*\leq \delta_1/m_{\mathscr{Y}}(\llbracket Y^\prime\rrbracket)$ such that 
   \begin{equation}
       |\llbracket Y^\prime|X^\prime\rrbracket^*_{\delta^*/|\llbracket X^\prime\rrbracket|}|>|\llbracket \bar{Y}|\bar{X}\rrbracket^*_{\bar{\delta}/|\llbracket\bar{X}\rrbracket|}|.
   \end{equation}
\end{itemize}
The result in Claim 4 contradicts \eqref{eq:assumptionOnChannel}. It follows that \eqref{eq:ProveContradiction} cannot hold and the proof of Theorem~\ref{thm:sufficientConditiondeltaerror} is complete.

\end{proof}

Since ${C}_{N}^{\delta*}$ is a special case of ${C}_{N}(\{\delta_n\})^*$ for which the sequence $\delta_n$ is constant, it seem natural to use Theorem \ref{thm:sufficientConditiondeltaerror} to obtain a single-letter expression for ${C}_{N}^{\delta*}$ as well. However,
the range of $\delta_n$ in Theorem \ref{thm:sufficientConditiondeltaerror} restricts the obtained  single-letter expression for this case  to the zero-error capacity ${C}_{N}^{0*}$ only. To see this, note that   $\delta_n$ in Theorem \ref{thm:sufficientConditiondeltaerror} 
is constrained to
\begin{align}\label{eq:deltaCapac}
    \delta_n &\leq (\bar{\delta} m_{\mathscr{Y}}(V_N)/|\llbracket\bar{X}\rrbracket|)^n \nonumber \\
    &< (m_{\mathscr{Y}}^2(V_N)/m_{\mathscr{Y}}(\llbracket\bar{Y}\rrbracket))^n. 
\end{align}
It follows that if $m_{\mathscr{Y}}(V_N)<m_{\mathscr{Y}}(\mathscr{Y})=1$, then we have $\delta_n = o(1)$ as $n \rightarrow \infty$. Hence, in this case a single letter expression for ${C}_{N}^{\delta*}$ can only be obtained for $\delta=0$. On the other hand,
if $m_{\mathscr{Y}}(V_N)=m_{\mathscr{Y}}(\mathscr{Y})$, then for any $0\leq\delta<m_{\mathscr{Y}}(V_N)$ the codebook can only contain a single codeword and in this case we have ${C}_{N}^{\delta*}=0$.
We conclude that the only non-trivial single-letter expression is obtained  for the zero-error capacity, as stated next.  
\begin{corollary}\label{corr:zeroerrorsufficientcondition}
For any  stationary  memoryless uncertain channel $N$, let  $\bar{X} \in \mathscr{F}_{0}(1)$  be an UV  over one time step associated with a one-dimensional codebook  that achieves the   capacity $C_N(\{0\})^* = C_N(\{0\})_*=R_0$,
and   let $\bar{Y}$  be the UV  corresponding to the received {codeword}, namely
\begin{align}
C_N(\{0\})^* & = I_{0}(\bar Y;\bar X)\nonumber \\
&= \sup_{  \substack{X(1)\in \mathscr{F}_{0}(1)}}{I_{0}(Y(1); X(1))}.
\end{align}
If for all one-dimensional codewords  $x\in \mathscr{X}\setminus\llbracket \bar{X} \rrbracket$, there exists a  set $\mathscr{S} \in \llbracket \bar Y|\bar{X}\rrbracket_{\bar\delta/|\llbracket \bar X\rrbracket|}^*$ such that   the uncertainty region $\llbracket  Y|x\rrbracket \subseteq \mathscr{S}$,
then  under   Assumptions ~\ref{assumption:productrule} and \ref{assumption:triangleInequality} we have that the $n$-dimensional zero-error capacity\\
\begin{equation}
    C_{N}^{0*} { = I_{0}(\bar{Y};\bar{X})}.
\end{equation}
\end{corollary}
Next, we present the sufficient conditions leading to the single letter expression for $C(\{\delta_n\})_*$.  

\begin{theorem}\label{thm:singleLetterInf}
For any stationary memoryless uncertain channel $N$, and for any $0\leq \delta_1<m_{\mathscr{Y}}(V_N)$, let $\Bar{X} \in \mathscr{F}_{\bar \delta}(1)$ be an UV  over one time step associated with a one-dimensional codebook  that achieves the   capacity $C_N(\{\delta_1\})_*=C_N(\{\delta_1\})^*=R_{\delta_1}$, and let $\bar{Y}$ be the UV  corresponding to the received {codeword}, namely
\begin{align}\label{eq:requiredCondition}
C_N(\{\delta_1\})_* &= I_{\bar\delta/|\llbracket \bar X\rrbracket|}(\bar Y;\bar X)  \nonumber \\
&= \sup_{  \substack{X(1)\in \mathscr{F}_{\tilde \delta}(1), \\\tilde{\delta}\leq \delta_1/m_{\mathscr{Y}}(\llbracket Y(1)\rrbracket)}}{I_{\tilde{\delta}/|\llbracket X(1)\rrbracket|}(Y(1); X(1))}.
\end{align}
Let
\begin{equation}
    \hat{\delta}=\max_{\mathscr{S}\in\llbracket\bar{Y}|\bar{X}\rrbracket^*_{\bar{\delta}/|\llbracket\bar{X}\rrbracket|}}\frac{m_{\mathscr{Y}}(\mathscr{S})}{m_{\mathscr{Y}}(\llbracket\bar{Y}\rrbracket)}.
\end{equation}
If for all $n>1$, we have $\bar{\delta}(\hat{\delta} |\llbracket\bar{X}\rrbracket|)^{n-1}\leq\delta_n<1$, then under    Assumption \ref{assumption:productrule} we have that the $n$-dimensional capacity
\begin{equation}\label{eq:CapacityInf}
     C_{N}(\{ \delta_n \})_*=I_{\bar\delta/|\llbracket \bar X\rrbracket|}(\bar Y;\bar X).
\end{equation}
\end{theorem}
\begin{proof}
First, we   show that for all $n  \in \mathbb{Z}_{>0}$ and $\delta_n\geq\bar{\delta}(\hat{\delta}|\llbracket\bar{X}\rrbracket|)^{n-1}$, there exists a codebook $\mathcal{\tilde{X}}(1:n)\in\mathscr{X}^{\delta_n}_{N}(n)$ such that 
\begin{equation}
    |\mathcal{\tilde{X}}(1:n)|=|\prod_{i=1}^{n}\llbracket\bar Y|\Bar{X}\rrbracket^*_{\bar{\delta}/|\llbracket\bar{X}\rrbracket|}|. 
\end{equation}
This, along with   Definition \ref{deflargestachievable}, implies that for all $n  \in \mathbb{Z}_{>0}$ and $\delta_n\geq\bar{\delta}(\hat{\delta}|\llbracket\bar{X}\rrbracket|)^{n-1}$, we have
\begin{equation}\label{eq:Rbound}
    R_{\delta_n}\geq I_{\bar\delta/|\llbracket \bar X\rrbracket|}(\bar Y;\bar X),
\end{equation}
and therefore
\begin{equation}\label{eq:ref11}
C_{N}(\{\delta_n\})_*=\inf_{n \in \mathbb{Z}_{>0}}R_{\delta_n}\geq I_{\bar{\delta}/|\llbracket\bar{X}\rrbracket|}(\bar{Y};\bar{X}).
\end{equation}
Second, we show that 
\begin{equation}\label{eq:ref12}
    C_{N}(\{\delta_n\})_*\leq I_{\bar\delta/|\llbracket \bar X\rrbracket|}(\bar Y;\bar X). 
\end{equation}
Thus, combining \eqref{eq:ref11} and \eqref{eq:ref12}, we have that \eqref{eq:CapacityInf} follows and the proof is complete.

We now start with the first step of showing \eqref{eq:ref11}. 
Without loss of generality, we   assume that $\llbracket\Bar{Y}|\bar{X}\rrbracket^*_{\bar{\delta}/|\llbracket\bar{X}\rrbracket|}>1$, otherwise 
\begin{equation}
    I_{\bar\delta/|\llbracket \bar X\rrbracket|}(\bar Y;\bar X)=0,
\end{equation}
and 
\eqref{eq:CapacityInf} holds trivially by the definition of $C_N(\{\delta_n\})_*$. 
Let 
\begin{equation}\label{eq:chekc123}
    \llbracket \bar X(1:n)\rrbracket=\underbrace{ \llbracket \bar X\rrbracket \times  \cdots \times \llbracket \bar X\rrbracket}_n,
\end{equation}
and 
\begin{equation}
    \llbracket \bar Y(1:n)\rrbracket=\underbrace{ \llbracket \bar Y\rrbracket \times  \cdots \times \llbracket \bar Y\rrbracket}_n. 
\end{equation}
Then, using $4)$ in Lemma \ref{thm:tensorization} and the fact that $N$ is a stationary memoryless channel, for all $\mathscr{S}_1,\mathscr{S}_2 \in \prod_{i=1}^n\llbracket\Bar{Y}|\bar{X}\rrbracket^*_{\bar{\delta}/|\llbracket\bar{X}\rrbracket|}$, we have 
\begin{equation}\label{eq:dissassociationCheck}
\begin{split}
     \frac{m_{\mathscr{Y}}(\mathscr{S}_1\cap\mathscr{S}_2)}{m_{\mathscr{Y}}(\llbracket \bar{Y}(1:n)\rrbracket)}&\leq \frac{\Bar{\delta}\hat{\delta}^{n-1}}{|\llbracket \bar{X}\rrbracket|},\\
     &\stackrel{(a)}{\leq} \frac{\delta_n}{(|\llbracket \bar{X}\rrbracket|)^n},\\
     &\stackrel{(b)}{=}\frac{\delta_n}{|\llbracket \bar{X}(1:n)\rrbracket|},
\end{split}
\end{equation}
where $(a)$ follows from the assumption in the theorem that $\delta_n\geq \bar{\delta}(\hat{\delta} |\llbracket\bar{X}\rrbracket|)^{n-1}$, and $(b)$ follows from \eqref{eq:chekc123}. 

Using $2)$ in Lemma \ref{thm:tensorization}, we have that for all $\mathscr{S}_{i}\in\prod_{i=1}^{n}\llbracket \bar{Y}|\bar{X}\rrbracket_{\bar{\delta}/|\llbracket \bar{X}\rrbracket|}^{*}$, there exists  $x_{i}(1:n)\in \llbracket \bar{X}(1:n)\rrbracket$ such that 
\begin{equation}\label{eq:containSet}
    \llbracket \bar{Y}(1:n)|x_i(1:n)\rrbracket \subseteq \mathscr{S}_{i}.
\end{equation}
 Now, let 
\begin{equation}\label{eq:codebookCard122}
    K=|\prod_{i=1}^{n}\llbracket \bar{Y}|\bar{X}\rrbracket_{\bar{\delta}/|\llbracket \bar{X}\rrbracket|}^{*}|.
\end{equation}
 Consider a new UV $\tilde{X}(1:n)$ whose marginal range is composed of $K$ elements of $\llbracket \bar{X}(1:n)\rrbracket$, namely 
\begin{equation}\label{eq:NewCodebookDesign}
    \llbracket \tilde{X}(1:n)\rrbracket=\{x_1(1:n),\ldots x_K(1:n)\}.
\end{equation}
 Let $\tilde{Y}(1:n)$  be the UV  corresponding to the received variable. For all $x(1:n)\in\llbracket\tilde{X}(1:n)\rrbracket
 $, we have $\llbracket\tilde{Y}(1:n)|x(1:n)\rrbracket=\llbracket\bar{Y}(1:n)|x(1:n)\rrbracket$ since $N$ is a stationary memoryless channel. Using   \eqref{eq:dissassociationCheck}, \eqref{eq:containSet}, it now follows that for all $x(1:n),x^{\prime}(1:n)\in \llbracket\tilde{X}(1:n)\rrbracket$,   we have  
\begin{equation}\label{eq:checkUpp}
\begin{split}
      &\frac{m_{\mathscr{Y}}(\llbracket \tilde{Y}(1:n)|x(1:n)\rrbracket\cap \llbracket \tilde{Y}(1:n)|x^{\prime}(1:n)\rrbracket)}{m_{\mathscr{Y}}(\llbracket \bar{Y}(1:n)\rrbracket)}\\
      &{\leq \frac{\delta_n}{|\llbracket \bar{X}(1:n)\rrbracket|},}\\
      &\stackrel{(a)}{\leq}  \frac{\delta_n}{|\llbracket \tilde{X}(1:n)\rrbracket|},
\end{split}
\end{equation}
where $(a)$ follows from the fact that using \eqref{eq:NewCodebookDesign}, we have $\llbracket\tilde{X}(1:n)\rrbracket\subseteq\llbracket \bar{X}(1:n)\rrbracket$. This implies that for all $x(1:n),x^{\prime}(1:n)\in \llbracket\tilde{X}(1:n)\rrbracket$,
\begin{equation}\label{eq:dissUpperBound}
\begin{split}
&e_{N}(x(1:n),x^\prime(1:n))\\
&=\frac{m_{\mathscr{Y}}(S_N(x(1:n))\cap S_N(x^\prime(1:n)))}{m_{\mathscr{Y}}(\mathscr{Y}^n)},\\
&\stackrel{(a)}{=}\frac{m_{\mathscr{Y}}(\llbracket \tilde{Y}(1:n)|x(1:n)\rrbracket\cap \llbracket \tilde{Y}(1:n)|x^{\prime}(1:n)\rrbracket)}{m_{\mathscr{Y}}(\mathscr{Y}^n)},\\
&\stackrel{(b)}{\leq} \frac{\delta_n}{|\llbracket \tilde{X}(1:n)\rrbracket|} \frac{ m_{\mathscr{Y}}(\llbracket \bar{Y}(1:n)\rrbracket)}{m_{\mathscr{Y}}(\mathscr{Y}^n)},\\
&\stackrel{(c)}{\leq}\frac{\delta_n}{|\llbracket \tilde{X}(1:n)\rrbracket|},
\end{split}
\end{equation}
where $(a)$ follows from the fact that $N$ is stationary memoryless and for all $x(1:n)\in \mathscr{X}^n$, we have
\begin{equation}
    \llbracket Y(1:n)|x(1:n)\rrbracket=S_N(x(1:n)),
\end{equation}
$(b)$ follows from \eqref{eq:checkUpp}, and $(c)$ follows from \eqref{eq:strongtransitivity} and $\llbracket\bar{Y}(1:n)\rrbracket\subseteq\mathscr{Y}^n$. This implies that the codebook $\mathcal{\tilde{X}}(1:n)$ corresponding to the UV $\tilde{X}(1:n)$ is $(N,\delta_n)$-distinguishable. 
It follows that \eqref{eq:Rbound} and \eqref{eq:ref11} hold and the 
the first step of the proof follows. 

Now, we prove the second step. We have 
\begin{equation}\label{eq:UpperBound11}
\begin{split}
   C_{N}(\{\delta_n\})_*&=\inf_{n \in \mathbb{Z}_{>0}}R_{\delta_n},\\
   &\stackrel{(a)}{\leq} R_{\delta_1},\\
   &\stackrel{(b)}{=}  R^{I}_{\delta_1},\\
   &\stackrel{(c)}{=} I_{\bar{\delta}/|\llbracket\bar{X}\rrbracket|}(\bar{Y};\bar{X}),
\end{split}
\end{equation}
where $(a)$ follows from the fact that 
\begin{equation}
    \inf_{n \in \mathbb{Z}_{>0}}R_{\delta_n}\leq R_{\delta_1},
\end{equation}
 $(b)$ follows from the fact that since $\delta_1<m_{\mathscr{Y}}(V_N)$, we have that
 \begin{equation}
     R_{\delta_1}=R^{I}_{\delta_1},
 \end{equation}
  using Theorem \ref{thm:relBwRatenInfoRate}, $(c)$ follows from the fact that 
  \begin{equation}
      R^{I}_{\delta_1}=I_{\bar{\delta}/|\llbracket\bar{X}\rrbracket|}(\bar{Y};\bar{X}),
  \end{equation}
   using \eqref{eq:optimizationAtEachTimeN}, \eqref{eq:RateNInfoRate} and \eqref{eq:requiredCondition}. 
 Hence, the second step of the proof follows. 
\end{proof}
Finally, we present the sufficient conditions leading to the single letter expression for ${C}_{N}(\{\downarrow 0\})_*$.  

\begin{theorem}\label{thm:singleLetterShannon}
Let $0\leq \delta_1<m_{\mathscr{Y}}(V_N)$. For any stationary memoryless uncertain channel $N$,   let ${X}^*\in\mathscr{F}_{\delta^*}(1)$   be an UV  over one time step associated with a one-dimensional codebook  that achieves  the  largest one-dimensional $\delta_1$-capacity,  and let  ${Y}^*$  be the UV  corresponding to the received {codeword}, namely $X^*$ achieves $\sup_{\delta_1<m_{\mathscr{Y}}(V_N)}C_N(\{\delta_1\})_*= \sup_{\delta_1<m_{\mathscr{Y}}(V_N)}C_N(\{\delta_1\})^* = \sup_{\delta_1<m_{\mathscr{Y}}(V_N)}R_{\delta_1}$, and we have
\begin{equation}\label{eq:CA1}
\begin{split}
    &\sup_{\delta_1<m_{\mathscr{Y}}(V_N)}C_N(\{\delta_1\})_*\\
    &=I_{\delta^*/|\llbracket  X^*\rrbracket|}( Y^*; X^*)\\ &=\sup_{\delta_1< m_{\mathscr{Y}}(V_N)} \sup_{  \substack{X(1)\in \mathscr{F}_{\tilde \delta}(1), \\\tilde{\delta}\leq \delta_1/m_{\mathscr{Y}}(\llbracket Y(1)\rrbracket)}}{I_{\tilde{\delta}/|\llbracket X(1)\rrbracket|}(Y(1); X(1))}.
\end{split}
\end{equation}
Let
\begin{equation}
    \hat{\delta}_*=\max_{\mathscr{S}\in\llbracket{Y}^*|{X}^*\rrbracket^*_{{\delta}^*/|\llbracket{X}^*\rrbracket|}}\frac{m_{\mathscr{Y}}(\mathscr{S})}{m_{\mathscr{Y}}(\llbracket{Y}^*\rrbracket)}.
\end{equation}
If $ \hat{\delta}_* |\llbracket X^*\rrbracket|<1$, then under  Assumption \ref{assumption:productrule} we have that the $n$-dimensional capacity
\begin{equation}\label{eq:ShannonAnalogue}
    {C}_{N}(\{\downarrow 0\})_*=I_{\delta^*/|\llbracket  X^*\rrbracket|}( Y^*; X^*).
\end{equation}
\end{theorem}
\begin{proof}

Consider a sequence of $\{\delta_n\}$ such that $\delta_1=\delta^*$ and $\delta_n={\delta^*}(\hat{\delta}_* |\llbracket X^*\rrbracket|)^{n-1}$ for $n>1$. 
Then, using Theorem \ref{thm:singleLetterInf} for this sequence   $\{\delta_n\}$, we have  that
\begin{equation}\label{eq:equalityOfCapacity}
    C_{N}(\{\delta_n\})_*=I_{\delta^*/|\llbracket  X^*\rrbracket|}( Y^*; X^*).
\end{equation}
Now, since $\hat{\delta}_* |\llbracket X^*\rrbracket|<1$ using the assumption in the theorem, 
we have
\begin{equation}\label{eq:sequenceProp}
    \lim_{n\to\infty}\delta_n=0.
\end{equation}
 Using \eqref{eq:equalityOfCapacity} and \eqref{eq:sequenceProp}, we have that 
\begin{equation}\label{eq:LowerBound111}
\begin{split}
    {C}_{N}(\{\downarrow 0\})_*&=\sup_{\{\delta^\prime_n\}:\delta^\prime_n=o(1)}C_{N}(\{\delta^\prime_n\})_*,\\
    &\geq C_{N}(\{\delta_n\})_*,\\
    &{=}I_{\delta^*/|\llbracket  X^*\rrbracket|}( Y^*; X^*).
\end{split}
\end{equation}
 We also have
\begin{equation}\label{eq:UpperBound111}
\begin{split}
   {C}_{N}(\{\downarrow 0\})_* &=\sup_{\{\delta^\prime_n\}:\delta^\prime_n=o(1)}\inf_{n \in \mathbb{Z}_{>0}}R_{\delta^\prime_n},\\
   &\stackrel{(a)}{\leq} \sup_{\{\delta^\prime_n\}:\delta^\prime_n=o(1)}R_{\delta^\prime_1},\\
   &\stackrel{(b)}{=}  \sup_{\{\delta^\prime_n\}:\delta^\prime_n=o(1)}R^{I}_{\delta^\prime_1},\\
   &\stackrel{(c)}{=} \sup_{\delta^\prime_1<m_{\mathscr{Y}(V_N)}}R^{I}_{\delta^\prime_1},\\
   &\stackrel{(d)}{=} I_{{\delta}^*/|\llbracket{X}^*\rrbracket|}({Y}^*;{X}^*),
\end{split}
\end{equation}
where $(a)$ follows from the fact that 
\begin{equation}
    \inf_{n \in \mathbb{Z}_{>0}}R_{\delta^\prime_n}\leq R_{\delta^\prime_1},
\end{equation}
 $(b)$ follows from the fact that since $\delta^\prime_1<m_{\mathscr{Y}}(V_N)$, we have
 \begin{equation}
     R_{\delta^\prime_1}=R^{I}_{\delta^\prime_1},
 \end{equation}
  using Theorem \ref{thm:relBwRatenInfoRate}, $(c)$ follows from the fact that $R^{I}_{\delta^\prime_1}$ is only dependent on $\delta^\prime_1$ in the sequence $\{\delta^\prime_n\}$, and $(d)$ follows from \eqref{eq:optimizationAtEachTimeN}, \eqref{eq:RateNInfoRate}, and the definition of $I_{{\delta}^*/|\llbracket{X}^*\rrbracket|}({Y}^*;{X}^*)$ in \eqref{eq:CA1}.

Combining \eqref{eq:LowerBound111} and \eqref{eq:UpperBound111}, the statement of the theorem follows.

\end{proof}


Table \ref{table:comparison} shows  a comparison between the sufficient conditions required to obtain single letter expressions for $C_N(\{\delta_n\})^*$, $C_{N}^{0*}$,  $C_N(\{\delta_n\})_*$ and ${C}_{N}(\{\downarrow 0\})_*$. We   point out that while in Theorems~\ref{thm:sufficientConditiondeltaerror} and \ref{thm:singleLetterInf} any  one-dimensional  $\delta_1$-capacity achieving codebook   can be used to obtain the single-letter expression, in Theorem~\ref{thm:singleLetterShannon} the single-letter expression requires a codebook   that achieves the largest capacity among all $\delta_1$-capacity achieving codebooks.
The sufficient conditions include in all cases Assumption~\ref{assumption:productrule}, which is required to factorize the uncertainty function over $n$ dimensions, and leads to the key Lemma~\ref{thm:tensorization}, and also to the upper bound on the mutual information   between associated, or disassociated   UVs in terms of the sum of the component-wise mutual information   expressed by Theorem~\ref{thm:tensorization2}.  The remaining conditions differ due to the different definitions of capacity.

\begin{table*}[t]
\begin{center}
\begin{tabular}{   m{4.0cm}    m{2cm}  m{2cm}  m{2cm}   m{2cm}   } \hline \hline \\
{SUFFICIENT CONDITIONS}  & $C_N(\{\delta_n\})^*$&$C_N^{0*}$& $C_N(\{\delta_n\})_*$ & ${C}_{N}(\{\downarrow 0\})_*$  \\
& (Theorem \ref{thm:sufficientConditiondeltaerror}) & (Corollary \ref{corr:zeroerrorsufficientcondition}) & (Theorem \ref{thm:singleLetterInf}) & (Theorem \ref{thm:singleLetterShannon})
\\ \\ \hline \\
  
 $m_{\mathscr{Y}}$ Satisfies Assumption 1& \checkmark & \checkmark & \checkmark &\checkmark\\ 
 
 $m_{\mathscr{Y}}$ Satisfies Assumption 2 & \checkmark & \checkmark & & \\ 
 
 1D Uncertainty Region Constraint & \checkmark & \checkmark & &  \\
 
Upper bound on $\delta_1$  & \checkmark & &\checkmark & \checkmark \\ 
 
Upper bound on $\{\delta_n\}_{n>1}$  &\checkmark & &  & \\ 
 
Lower bound  on $\{\delta_n\}_{n>1}$ & & & \checkmark & \\ 
{ Upper bound on $\bar{\delta}$} &\checkmark & & &\\
Upper bound on $\hat{\delta}_*$ & & & &\checkmark\\ 
 \hline \hline
\end{tabular}
\end{center}
\caption{Comparison of the Sufficient Conditions for the 
Existence of a Single-Letter expression\label{table:comparison}}
\end{table*}

\section{Examples}

To cast our sufficient conditions for the existence of   single letter expressions of 
capacity in a concrete setting, we now provide some examples and compute the corresponding capacity. 

In the following, we represent stationary memoryless uncertain channels in graph form. Let $\mathcal{G}(V,E)$ be a directed graph, where $V$ is the set of vertices and $E$ is the set of edges. The vertices  represent input and output codeword symbols, namely  $V=\mathscr{X}\cup\mathscr{Y}$. A directed edge from node $x\in \mathscr{X}$ to node $y\in \mathscr{Y}$, denoted by $x\to y$, shows that given symbol $x$ is transmitted, $y$ may be received at the output of the channel. It follows that  for all $x\in\mathscr{X}$,   the channel transition map representing the noise experienced by each codeword is given by 
\begin{equation}
    N(x)= \{y: (x\to y) \in E\}. 
\end{equation}

\begin{example} \label{ex:1}
We   consider a channel  with 
\begin{equation}
\mathscr{X}=\mathscr{Y}=\{1,2,3,\ldots,19\}.
\end{equation}
To define the channel transition map, we let for all $x\in \{1,2,3,4,5,6\}$
\begin{equation}
    N(x)=\{1,2,3,4,5,6,11\},
\end{equation}
  for all $x\in \{7,8,9,10,11,12\}$
\begin{equation}
    N(x)=\{7,8,9,10,11,12,2\},
\end{equation}
and for all $x\in \{13,14,15,16,17,18,19\}$
\begin{equation}
    N(x)=\{13,14,15,16,17,18,19\}.
\end{equation}
The corresponding   graph   is   depicted in Figure \ref{fig:Chaneldelta}. 
\begin{figure}
\begin{center}
\includegraphics[width=0.45\textwidth ]{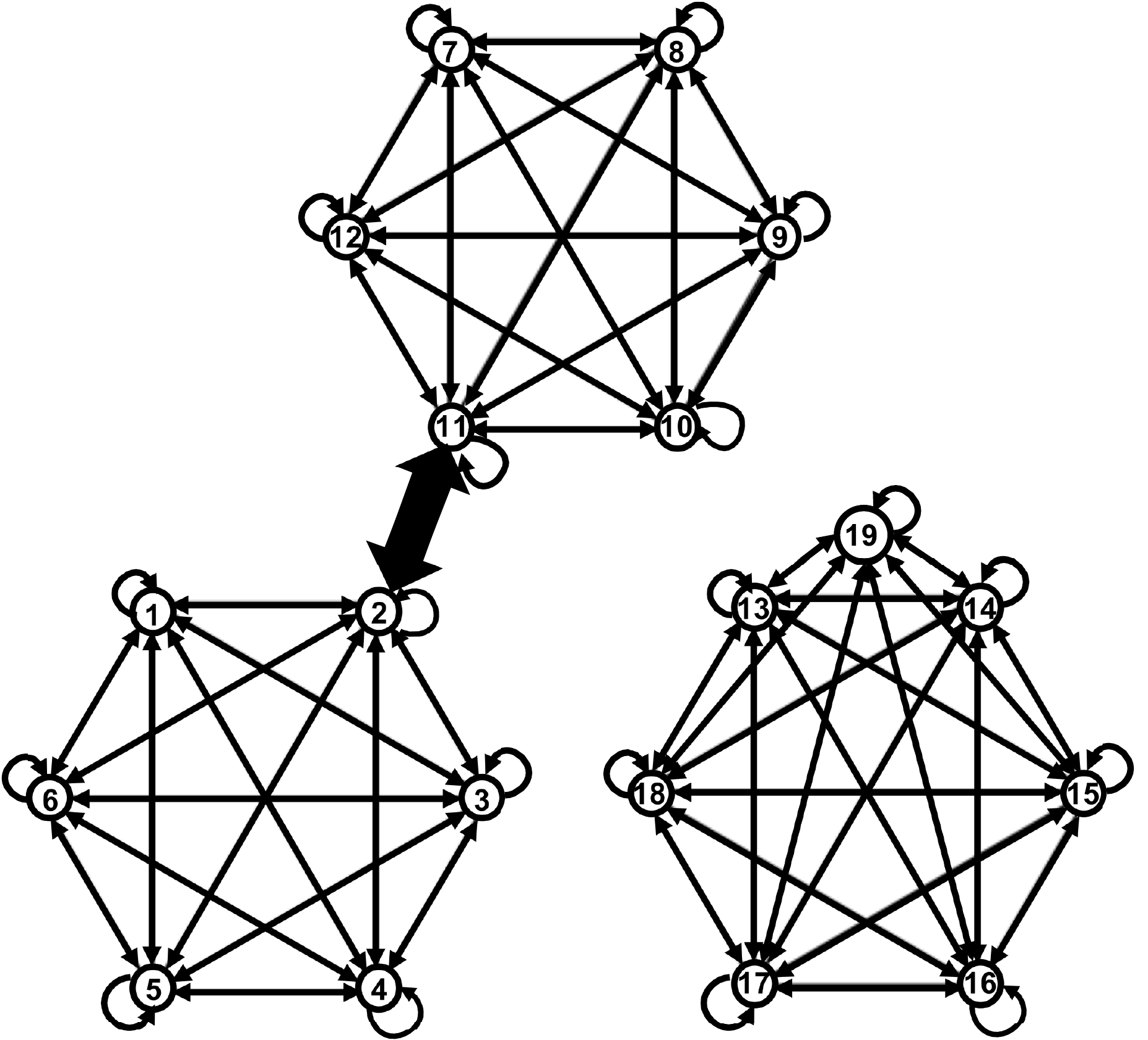}
\end{center}
\caption{Channel described in Example~\ref{ex:1}.
It consists of three complete graphs, and some additional edges. The thick solid arrow into node $y=11$ represents multiple edges connecting all the nodes in the set $\{1,2,3,4,5,6\}$ to node $11$. Similarly, all the nodes in the set $\{7,8,9,10,11,12\}$ are connected to node $2$.
}
\label{fig:Chaneldelta}
\end{figure}[t]
For any $\mathcal{Y}_n\subseteq \mathscr{Y}^n$, we define the uncertainty function  $m_{\mathscr{Y}}(\mathcal{Y}_n)$   in terms of cardinality
\begin{equation}\label{eq:uncertainFunct}
    m_{\mathscr{Y}}(\mathcal{Y}_n)=\frac{|\mathcal{Y}_n|}{|\mathscr{Y}^n|}.
\end{equation}
Note that for all $n  \in \mathbb{Z}_{>0}$, we have $m_{\mathscr{Y}}(\mathscr{Y}^n)=1$. 

It is easy to show that $ m_{\mathscr{Y}}(.)$ satisfies Assumption \ref{assumption:productrule}. Namely, for $n=1$, we have that for all $\mathcal{Y}\subseteq \mathscr{Y}$, 
\begin{equation}\label{eq:UConeDimen}
    m_{\mathscr{Y}}(\mathcal{Y})=\frac{|\mathcal{Y}|}{|\mathscr{Y}|}.
\end{equation}
Let $\mathcal{Y}_n=\mathcal{Y}(1)\times\mathcal{Y}(2)\ldots\times \mathcal{Y}(n)$, where $\mathcal{Y}(i)\subseteq\mathscr{Y}$. Then, we have
\begin{equation}
\begin{split}
&m_{\mathscr{Y}}(\mathcal{Y}_n)\\
&=m_{\mathscr{Y}}(\mathcal{Y}(1)\times\mathcal{Y}(2)\ldots\times \mathcal{Y}(n)),\\
&=\frac{|\mathcal{Y}(1)\times\mathcal{Y}(2)\ldots\times \mathcal{Y}(n)|}{|\mathscr{Y}^n|},\\
&\stackrel{(a)}{=}\frac{|\mathcal{Y}(1)|}{|\mathscr{Y}|} \frac{|\mathcal{Y}(2)|}{|\mathscr{Y}|}\ldots \frac{|\mathcal{Y}(n)|}{|\mathscr{Y}|},\\
&\stackrel{(b)}{=}m_{\mathscr{Y}}(\mathcal{Y}(1)) m_{\mathscr{Y}}(\mathcal{Y}(2))\ldots  m_{\mathscr{Y}}(\mathcal{Y}(n)),
\end{split}
\end{equation}
where $(a)$ follows from the fact that for any two sets $\mathscr{S}_1$ and $\mathscr{S}_2$, $|\mathscr{S}_1\times \mathscr{S}_2|=|\mathscr{S}_1||\mathscr{S}_2|$, and $(b)$ follows from \eqref{eq:UConeDimen}. It follows that $ m_{\mathscr{Y}}(.)$ satisfies Assumption \ref{assumption:productrule}. 

A similar argument shows that $ m_{\mathscr{Y}}(.)$ also satisfies Assumption \ref{assumption:triangleInequality}. Namely, let $\mathcal{Y}_n=\mathcal{Y}(1)\cup\mathcal{Y}(2)\ldots\cup \mathcal{Y}(n)$, where $\mathcal{Y}(i)\in\mathscr{Y}$. Then, we have
\begin{equation}
\begin{split}
&m_{\mathscr{Y}}(\mathcal{Y}_n)\\
&=m_{\mathscr{Y}}(\mathcal{Y}(1)\cup\mathcal{Y}(2)\ldots\cup \mathcal{Y}(n)),\\
&=\frac{|\mathcal{Y}(1)\cup\mathcal{Y}(2)\ldots\cup \mathcal{Y}(n)|}{|\mathscr{Y}|},\\
&\stackrel{(a)}{\leq}\frac{|\mathcal{Y}(1)|}{|\mathscr{Y}|}+ \frac{|\mathcal{Y}(2)|}{|\mathscr{Y}|}+\ldots+ \frac{|\mathcal{Y}(n)|}{|\mathscr{Y}|},\\
&\stackrel{(b)}{=}m_{\mathscr{Y}}(\mathcal{Y}(1))+ m_{\mathscr{Y}}(\mathcal{Y}(2))+\ldots+ m_{\mathscr{Y}}(\mathcal{Y}(n)),
\end{split}
\end{equation}
where $(a)$ follows from the fact that for any two sets $\mathscr{S}_1$ and $\mathscr{S}_2$, $|\mathscr{S}_1\cup \mathscr{S}_2|\leq|\mathscr{S}_1|+|\mathscr{S}_2|$, and $(b)$ follows from \eqref{eq:UConeDimen}. It follows that $ m_{\mathscr{Y}}(.)$ satisfies Assumption \ref{assumption:triangleInequality}.

We now compute the capacity $C_N(\{\delta_n\})^*$ for { $\delta_1=2/9$} 
 and for all $n>1$  $\delta_n=(7/342)^n$. 

Since $V_N$ contains seven elements, we have that $m_{\mathscr{Y}}(V_N)=7/19$, and $\delta_1<m_{\mathscr{Y}}(V_N)$. Consider an UV $\bar{X}$ representing a one-dimensional codebook such that
\begin{equation} \label{ex1a}
    \llbracket\bar{X}\rrbracket=\{1,7,13\}.
\end{equation}
It follows that the corresponding output UV  $\bar{Y}$ is such that
\begin{equation} \label{ex1b}
    \llbracket\bar{Y}\rrbracket=\{1,2,3,\ldots,18,19\}.
\end{equation}
Letting $\bar{\delta}=1/6$, we have that $\bar{\delta}/| \llbracket X \rrbracket | =1/18$ and the   overlap family
\begin{equation} \label{ex1c}
    \llbracket \bar Y|\bar X\rrbracket^*_{1/18}=\{\mathscr{S}_1,\mathscr{S}_2\},
\end{equation}
where
\begin{equation}
    \mathscr{S}_1=\cup_{x\in \{1,7\}}N(x),
\end{equation}
\begin{equation}
    \mathscr{S}_2=\cup_{x\in\{13\}}N(x).
\end{equation}
 
We now show that $\bar{X}$ satisfies the sufficient conditions   in Theorem \ref{thm:sufficientConditiondeltaerror}. 
First, we note that for all $x\in\{2,3,4,5,6,8,9,10,11,12\}$, we have that $\llbracket \bar{Y}|x\rrbracket\subseteq \mathscr{S}_1$, and for all $x\in\{14,15,16,17,18,19\}$, we have  $\llbracket \bar{Y}|x\rrbracket\subseteq \mathscr{S}_2$. 
It follows that for all $x\in \mathscr{X}\setminus\llbracket\bar{X}\rrbracket   =\{2,3,4,5,6,8,9,10,11,12,14,15,16,17,18,19\}$, the uncertainity region $\llbracket \bar{Y}|x\rrbracket\subseteq\mathscr{S}$, where $\mathscr{S}\in \llbracket\bar{Y}|\bar{X}\rrbracket^*_{\bar{\delta}/|\llbracket\bar{X}\rrbracket|}$. 
 Second, we have that 
\begin{equation}
    \bar{\delta} (1+1/|\llbracket\bar{X}\rrbracket|)=2/9\leq \delta_1/m_{\mathscr{Y}}(\llbracket\bar{Y}\rrbracket). 
\end{equation}

 Third, we note that $\bar{\delta}=1/6$,  
 \begin{equation}
     \bar\delta m_{\mathscr{Y}}(V_N)/|\llbracket \bar X\rrbracket|= 7/342.
 \end{equation}
It follows that for all $n>1$, we have that $\delta_n\leq (\bar\delta m_{\mathscr{Y}}(V_N)/|\llbracket \bar X\rrbracket|)^n$.

 Since all the sufficient conditions in Theorem \ref{thm:sufficientConditiondeltaerror} are satisfied, we have
 \begin{equation}
     C_{N}(\{\delta_n\})^*=\log_2|\llbracket\bar Y|\bar{X}\rrbracket^*_{1/18}|=1. 
 \end{equation}
\end{example}
\begin{example}
We now consider the same channel as in Example~\ref{ex:1}, shown in Figure \ref{fig:Chaneldelta}, 
and  we compute the capacity $C_N^{0*}$. We consider the one-dimensional codebook $\bar{X}$ in \eqref{ex1a}, and the corresponding   output UV  $\bar{Y}$ in \eqref{ex1b}. Then, we have
\begin{equation} 
    \llbracket \bar Y|\bar X\rrbracket^*_{0}=\{\mathscr{S}_1,\mathscr{S}_2\},
\end{equation}
where
\begin{equation}
    \mathscr{S}_1=\cup_{x\in \{1,7\}}N(x),
\end{equation}
\begin{equation}
    \mathscr{S}_2=\cup_{x\in\{13\}}N(x).
\end{equation}
We now show that $\bar{X}$ satisfies the sufficient conditions   in Corollary \ref{corr:zeroerrorsufficientcondition}. 
We note that for all $x\in\{2,3,4,5,6,8,9,10,11,12\}$, we have that $\llbracket \bar{Y}|x\rrbracket\subseteq \mathscr{S}_1$, and for all $x\in\{14,15,16,17,18,19\}$, we have  $\llbracket \bar{Y}|x\rrbracket\subseteq \mathscr{S}_2$. 
It follows that for all $x\in \mathscr{X}\setminus\llbracket\bar{X}\rrbracket   =\{2,3,4,5,6,8,9,10,11,12,14,15,16,17,18,19\}$, the uncertainity region $\llbracket \bar{Y}|x\rrbracket\subseteq\mathscr{S}$, where $\mathscr{S}\in \llbracket\bar{Y}|\bar{X}\rrbracket^*_{0}$. 

Since all the sufficient conditions in Corollary \ref{corr:zeroerrorsufficientcondition} are satisfied, we have
 \begin{equation}
     C_{N}^{0*}=\log_2|\llbracket\bar Y|\bar{X}\rrbracket^*_{0}|=1. 
 \end{equation}
\end{example}

\begin{example} \label{ex:2}  \label{exp:Channel1_deltaInfError}
 We now consider the same channel as in Example~\ref{ex:1}, shown in Figure \ref{fig:Chaneldelta}, 
and  we compute the capacity $C_N(\{\delta_n\})_*$ for $\delta_1=(2/6)^3$ and for all $n>1$ $\delta_n=(2/6)^3((7/19)^3 3)^{n-1}$.  

For any $\mathcal{Y}_n\subseteq \mathscr{Y}^n$, we define the uncertainty function  $m_{\mathscr{Y}}(\mathcal{Y}_n)$   in terms of cardinality
\begin{equation}\label{eq:uncertainFunct1}
    m_{\mathscr{Y}}(\mathcal{Y}_n)=\bigg(\frac{|\mathcal{Y}_n|}{|\mathscr{Y}^n|}\bigg  )^3.
\end{equation}
Note that for all $n  \in \mathbb{Z}_{>0}$, we have $m_{\mathscr{Y}}(\mathscr{Y}^n)=1$. It is easy to show that $ m_{\mathscr{Y}}(.)$ satisfies Assumption \ref{assumption:productrule}. For $n=1$, we have that for all $\mathcal{Y}\subseteq \mathscr{Y}$, 
\begin{equation}\label{eq:UConeDimen1}
    m_{\mathscr{Y}}(\mathcal{Y})=\bigg(\frac{|\mathcal{Y}|}{|\mathscr{Y}|}\bigg)^3.
\end{equation}
Let $\mathcal{Y}_n=\mathcal{Y}(1)\times\mathcal{Y}(2)\ldots\times \mathcal{Y}(n)$, where $\mathcal{Y}(i)\in\mathscr{Y}$. Then, we have
\begin{equation}
\begin{split}
&m_{\mathscr{Y}}(\mathcal{Y}_n)\\
&=m_{\mathscr{Y}}(\mathcal{Y}(1)\times\mathcal{Y}(2)\ldots\times \mathcal{Y}(n)),\\
&=\bigg(\frac{|\mathcal{Y}(1)\times\mathcal{Y}(2)\ldots\times \mathcal{Y}(n)|}{|\mathscr{Y}^n|}\bigg)^3,\\
&\stackrel{(a)}{=}\bigg(\frac{|\mathcal{Y}(1)|}{|\mathscr{Y}|}\bigg)^3 \bigg(\frac{|\mathcal{Y}(2)|}{|\mathscr{Y}|}\bigg)^3\ldots \bigg(\frac{|\mathcal{Y}(n)|}{|\mathscr{Y}|}\bigg)^3,\\
&\stackrel{(b)}{=}m_{\mathscr{Y}}(\mathcal{Y}(1)) m_{\mathscr{Y}}(\mathcal{Y}(2))\ldots  m_{\mathscr{Y}}(\mathcal{Y}(n)),
\end{split}
\end{equation}
where $(a)$ follows from the fact that for any two sets $\mathscr{S}_1$ and $\mathscr{S}_2$, $|\mathscr{S}_1\times \mathscr{S}_2|=|\mathscr{S}_1||\mathscr{S}_2|$, and $(b)$ follows from \eqref{eq:UConeDimen1}. It follows that $ m_{\mathscr{Y}}(.)$ satisfies Assumption \ref{assumption:productrule}. 

Since $m_{\mathscr{Y}}(V_N)=(7/19)^3$, we have $\delta_1<m_{\mathscr{Y}}(V_N)$. 
We consider an UV $\bar{X}$ representing a one-dimensional codebook,   a corresponding output UV  $\bar{Y}$, and $\bar{\delta}=(2/6)^3$, so that
\begin{equation}
    \llbracket\bar{X}\rrbracket=\{1,7,13\},
\end{equation}
\begin{equation}
    \llbracket\bar{Y}\rrbracket=\{1,2,\ldots,19\},
\end{equation}
\begin{equation}
     \llbracket\bar{Y}|\bar{X}\rrbracket^*_{1/81}=\{\mathscr{S}_1,\mathscr{S}_2,\mathscr{S}_3\},
\end{equation}
where 
\begin{equation}
    \mathscr{S}_1=\cup_{x\in \{1\}}N(x),
\end{equation}
\begin{equation}
    \mathscr{S}_2=\cup_{x\in \{7\}}N(x),
\end{equation}
\begin{equation}
    \mathscr{S}_3=\cup_{x\in\{13\}}N(x).
\end{equation}


  Since $\bar{\delta}=\delta_1=(2/6)^3$,  we have
 \begin{equation}
     \hat{\delta}=\max_{\mathscr{S}\in\llbracket\bar{Y}|\bar{X}\rrbracket^*_{\bar{\delta}/|\llbracket\bar{X}\rrbracket|}}\frac{m_{\mathscr{Y}}(\mathscr{S})}{m_{\mathscr{Y}}(\llbracket\bar{Y}\rrbracket)}=\bigg(\frac{7}{19}\bigg)^3.
 \end{equation}
It follows that for all $n>1$, we have $\delta_n\geq \bar{\delta}(\hat{\delta} |\llbracket\bar{X}\rrbracket|)^{n-1}$ and all the sufficient conditions in Theorem \ref{thm:singleLetterInf} are satisfied, so that
 \begin{equation}
     C_{N}(\{\delta_n\})_*=\log_2|\llbracket\bar Y|\bar{X}\rrbracket^*_{1/81}|=\log_2(3
     ). 
 \end{equation}
 \end{example}
 
 \begin{example} \label{ex:3}
We again consider the same channel and the same uncertainty function as in Example \ref{ex:2}, shown in Figure \ref{fig:Chaneldelta} and \eqref{eq:uncertainFunct1}, respectively. 
We compute the capacity $C_N(\{\downarrow 0\})_*$. 
Consider an UV $X^*$ representing a one-dimensional codebook such that
\begin{equation}
    \llbracket X^*\rrbracket=\{1,7,13\}.
\end{equation}
It follows that the corresponding output UV $Y^*$ is such that
\begin{equation}
    \llbracket Y^*\rrbracket=\{1,2,3\ldots 19\}. 
\end{equation}
Letting $\delta^*=(2/6)^3$, we have that $\delta^*/|\llbracket X^*\rrbracket|=1/81$ and the overlap family is 
\begin{equation}
    \llbracket  Y^*| X^*\rrbracket^*_{\delta^*/|\llbracket X^*\rrbracket|=1/81}=\{\mathscr{S}_1,\mathscr{S}_2,\mathscr{S}_3\},
\end{equation}
where
\begin{equation}
    \mathscr{S}_1=\cup_{x\in \{1\}}N(x),
\end{equation}
\begin{equation}
    \mathscr{S}_2=\cup_{x\in \{7\}}N(x),
\end{equation}
\begin{equation}
    \mathscr{S}_3=\cup_{x\in\{13\}}N(x).
\end{equation}
 Since $V_N$ contains seven elements, we have that $m_{\mathscr{Y}}(V_N)=(7/19)^3$, and $\delta^*<m_{\mathscr{Y}}(V_N)$. Also, we have that
 \begin{equation}
    \hat{\delta}_*=\max_{\mathscr{S}\in\llbracket{Y}^*|{X}^*\rrbracket^*_{{\delta}^*/|\llbracket{X}^*\rrbracket|}}\frac{m_{\mathscr{Y}}(\mathscr{S})}{m_{\mathscr{Y}}(\llbracket{Y}^*\rrbracket)} = \bigg(\frac{7}{19}\bigg)^3.
 \end{equation}
 It follows that $\hat{\delta}_* |\llbracket X^*\rrbracket|<1$. 

 Since all the sufficient conditions in Theorem \ref{thm:singleLetterShannon} are satisfied, we have
 \begin{equation}
    {C}_{N}(\{\downarrow 0\})_*=\log_2|\llbracket\bar Y|\bar{X}\rrbracket^*_{2/18}|=\log_2(3). 
 \end{equation}
 \end{example}

\subsection{Discussion}
The results in our examples show that 
for the channel presented in Figure \ref{fig:Chaneldelta}, and the uncertainty function \eqref{eq:uncertainFunct}, there exists a vanishing sequence $\delta_1=2/9, \{\delta_n\}_2^{\infty} =\{(7/342)^n\}$ such that
\begin{equation}\label{eq:examp23}
 C_{N}^{0*}=C_{N}(\{\delta_n\})^*.
\end{equation}
For the same channel and  uncertainty function, there is another   vanishing sequence {$\delta_1=4/9,\{\delta_n\}_2^{\infty} = \{(14/342)^n\}$}, such that
\begin{equation}\label{eq:examp23.1}
     C_{N}^{0*}<C_{N}(\{\delta_n\})^*. 
\end{equation}

On the other hand, for the same channel using the uncertainty function  \eqref{eq:uncertainFunct1}, there exists a vanishing sequence $\delta_1=(2/6)^3, \{\delta_n\}_2^{\infty} = \{ (2/6)^3(3(7/19)^3)^{n-1}\}$ such that 
\begin{equation}\label{eq:examp45}
C_{N}(\{\delta_n\})_*={C}_{N}(\{\downarrow 0\})_*.
\end{equation}
For the same channel and   uncertainty function \eqref{eq:uncertainFunct1}, there exists another sequence $\delta_1=(2/19)^3,\{\delta_n\}_2^\infty=(2/19)^3(3(12/19)^3)^{n-1}$ such that
\begin{equation}\label{eq:examp45.1}
    C_{N}(\{\delta_n\})_*<{C}_{N}(\{\downarrow 0\})_*.
\end{equation}

\section{Applications}
We  now  discuss some applications of the developed non-stochastic theory. 
\subsection{Error Correction in Adversarial Channels}


Various adversarial channel models have been considered in the literature. A popular one considers a   binary alphabet, and a codeword of length $n$  that is sent from the transmitter to the receiver.  The channel can flip at most a fraction $0<\tau\leq 1$ of the $n$ symbols in an arbitrary fashion~\cite{alon2018list}. In this case, the input and output spaces are $\mathscr{X}^n=\mathscr{Y}^n=\{0,1\}^n$, and a codebook is $\mathcal{X}_n \subseteq\mathscr{X}^n$.
Due to the constraint on the total number of bit flips, the channel is non-stationary and with memory. 
For any  $x\in \mathscr{X}^n$,
we can let the norm be the Hamming distance from the $n$-dimensional all zero vector representing the origin of the space, namely 
\begin{equation}\label{eq:norm}
    \|x\|= H(x,\{0\}_n)\leq n.
\end{equation}
In this framework,
for any transmitted codeword $x\in\mathcal{X}_n$, the set of possible received codewords is
\begin{equation}\label{eq:advChann}
    S_{\tau n}(x)= \{y\in \mathscr{Y}^n: {H}(x,y)\leq \tau  n\},
\end{equation}
where   $\tau n$  is the analogous of a noise range $\epsilon_n=\epsilon n$ in the non-stochastic channel model described in Section \ref{sec:cap}. 


For all $x_1,x_2\in\mathcal{X}_n$, the equivocation region corresponds to $S_{\tau n}(x_1)\cap S_{\tau n}(x_2)$ and for any $\mathscr{S} \subseteq \mathscr{Y}^n$, we can define the uncertainty function 
\begin{equation}
    m_{\mathscr{Y}^n}(\mathscr{S})=
    \begin{cases}
D(\mathscr{S})+1, \mbox{ if }\mathscr{S}\neq\emptyset,\\
0,  \mbox{ otherwise },
\end{cases}
\end{equation}
where $D$ indicates diameter, namely
\begin{equation}
    D(\mathscr{S})=\max_{y_1,y_2\in\mathscr{S}}H(y_1,y_2).
\end{equation}
With this definition, we have  that
\begin{equation}
    m_{\mathscr{Y}^n}(\mathscr{Y}^n)=n+1.
\end{equation}
For all $x_1,x_2\in\mathcal{X}_n$, we let the error
\begin{equation}
    e_{\tau n}(x_1,x_2)=\frac{m_{\mathscr{Y}^n}(S_{\tau n}(x_1)\cap S_{\tau n}(x_2))}{n+1}.
\end{equation}
As usual, we say that a codebook $\mathcal{X}_n$ is $(\tau n,\delta_n)$-distinguishable if $e_{\tau n}(x_1,x_2)\leq {\delta_n}/{|\mathcal{X}_n|}$, and for all $n \in \mathbb{Z}_{>0}$ the $(\tau n, \delta_n)$ capacity is
\begin{equation}
    C^{\delta_n}_{\tau n}=\sup_{\mathcal{X}_n\in \mathscr{X}^{\delta_n}_{\tau n}}\log_2(|\mathcal{X}_n|),
\end{equation}
where $\mathscr{X}^{\delta_n}_{\tau n}=\{\mathcal{X}_n:\mathcal{X}_n\mbox{ is   }(\tau n,\delta_n)\mbox{-distinguishable}\}$. 




We now show that any $(\tau n, \delta_n)$-distinguishable codebook $\mathcal{X}_n$ can be used to correct a certain number of bit flips. This number depends on how far apart any two codewords are, and a lower bound on this distance can be expressed in terms of the diameter of the equivocation set and of the amount of perturbation introduced by the channel, see Figure~\ref{fig:equivocation2}. The following theorem provides a  lower bound on the Hamming distance $H(x_1,x_2)$ between any two codewords.   The number of bit flips that can be corrected can then be computed using the well-known formula $\lfloor(H(x_1,x_2)-1)/2 \rfloor$. Finally, we point out that non-stochastic, adversarial, error correcting codes  are of interest and have been  studied in the context of multi-label classification in machine learning \cite{ferng2011multi}, and to improve the robustness of neural networks to adversarial attacks \cite{verma2019error}.

\begin{figure}[t]
\begin{center}
\includegraphics[width=.35\textwidth ]{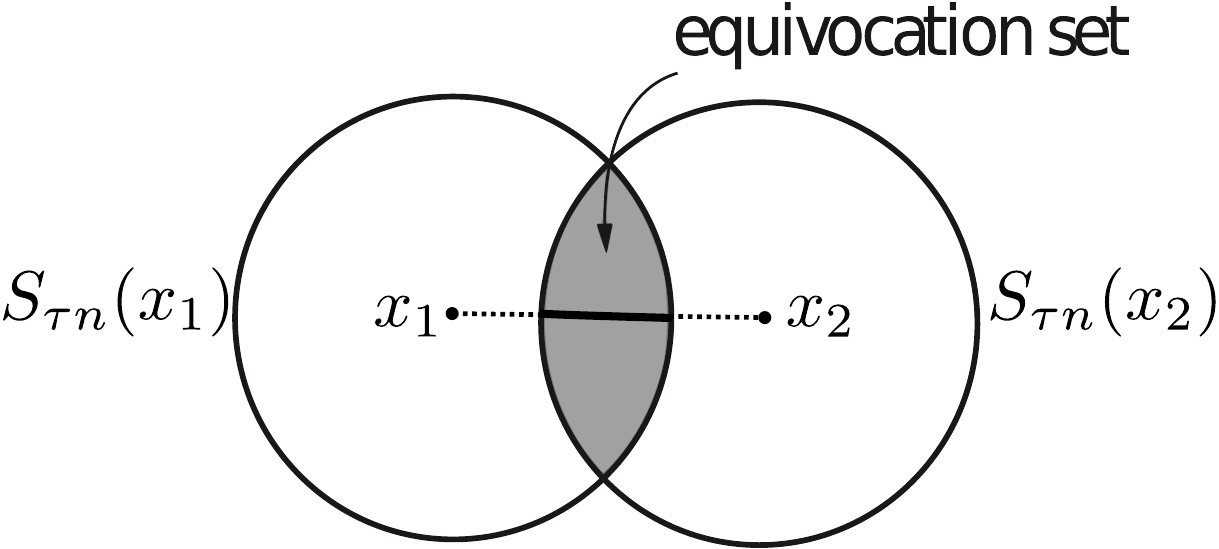}
\end{center}
\caption{The Hamming distance between any two overlapping codewords depends on the code parameters $\tau n$ and $\delta_n$.}
\label{fig:equivocation2}
\end{figure}



\begin{theorem}\label{lemma:application}
Given a channel satisfying \eqref{eq:advChann},  if a codebook $\mathcal{X}_n$ is $(\tau n,\delta_n)$-distinguishable, 
 then  for all $x_1,x_2\in \mathcal{X}_n$, we have
\begin{equation}\label{eq:resultErrorCorrecting}
    {H}(x_1,x_2)\geq \bigg(\frac{2\tau n}{n+1}-\frac{\delta_n}{|\mathcal{X}_n|}\bigg) (n+1)+1. 
\end{equation}
\end{theorem}
\begin{proof}

Let $\mathcal{X}_n\in\mathscr{X}^{\delta}_{\tau n}$. Then,  for all $x_1,x_2\in \mathcal{X}_n$, we have
\begin{equation}\label{eq:upperBe}
\begin{split}
    e_{\tau n}(x_1,x_2)&=\frac{m_{\mathscr{Y}^n}(S_{\tau n}(x_1)\cap S_{\tau n}(x_2))}{n+1},\\
    &\leq \frac{\delta_n}{|\mathcal{X}_n|}.
\end{split}
\end{equation}


First, we consider the case when 
\begin{equation}\label{eq:SecondSet}
    x_1,x_2\in S_{\tau n}(x_1)\cap S_{\tau n}(x_2). 
\end{equation}
Let  $B(x_1,x_2)$ be the boundary of the equivocation set, namely
\begin{equation}\label{}
\begin{split}
    B(x_1,x_2) &=\{x\in S_{\tau n}(x_1)\cap S_{\tau n}(x_2)):\\ &\qquad\exists x^\prime\notin S_{\tau n}(x_1)\cap S_{\tau n}(x_2)) \mbox{ such that }\\
    &\qquad H(x,x^\prime)=1\}. 
\end{split}
\end{equation}
By \eqref{eq:SecondSet}, there exist $A,B\in  B(x_1,x_2)$ such that
\begin{equation}
    H(A,B)= H(A,x_1)+H(x_1,x_2)+H(x_2,B).
\end{equation}
Then, we have
\begin{equation}
\begin{split}
     H(A,B)&= H(A,x_1)+H(x_1,x_2)+H(x_2,B),\\
     &\stackrel{(a)}{=}2\tau n-2H(x_1,x_2) +H(x_1,x_2),\\
     &\stackrel{}{=} 2\tau n -H(x_1,x_2),
\end{split}
\end{equation}
where $(a)$ follows from the fact that $A,B\in  B(x_1,x_2)$, which implies 
\begin{equation}
    H(A,x_1)=H(A,x_2)-H(x_1,x_2)=\tau n-H(x_1,x_2),
\end{equation}
and
\begin{equation}
    H(B,x_2)=H(B,x_1)-H(x_1,x_2)=\tau n-H(x_1,x_2).
\end{equation}
We now have  that
\begin{equation}\label{eq:secondDirection}
\begin{split}
    H(x_1,x_2)&= 2\tau n -H(A,B),\\
    &\stackrel{(a)}{\geq} 2\tau n - D(S_{\tau n}(x_1)\cap S_{\tau n}(x_2)),
\end{split}
\end{equation}
where $(a)$ follows from the fact that using $A,B\in S_{\tau n}(x_1)\cap S_{\tau n}(x_2))$, we have $H(A,B)\leq D(S_{\tau n}(x_1)\cap S_{\tau n}(x_2))$. 

Now, we consider the case when
\begin{equation}
    x_1,x_2\notin S_{\tau n}(x_1)\cap S_{\tau n}(x_2).
\end{equation}
In this case, we have 
\begin{equation}\label{eq:firstDirection}
\begin{split}
    &H(x_1,x_2)\\
    &\stackrel{}{\geq} 2\tau n-D(S_{\tau n}(x_1)\cap S_{\tau n}(x_2)).\\
    \end{split}
\end{equation}

The result now follows by combining  \eqref{eq:secondDirection}, \eqref{eq:firstDirection} and \eqref{eq:upperBe}.


\end{proof}
\subsection{Robustness of   Neural Networks to Adversarial Attacks}
\label{sec:robustness}
Motivated by security considerations, the robustness of neural networks to adversarial examples has recently received great attention  \cite{weng2018evaluating,weng2018towards,wang2018efficient}. While different algorithms have been proposed to improve robustness \cite{weng2018towards,wang2018efficient}, studies quantifying the limits of robustness have been limited \cite{weng2018evaluating,szegedy2013intriguing}.  We argue that the
 non-stochastic framework introduced in this paper can be a viable way to quantify   robustness, and   can be used as a baseline to evaluate the performance of different algorithmic solutions. 

We follow the framework of \cite{weng2018evaluating} and consider a neural network trained to classify the incoming data among $L$ possible labels in the set $\mathscr{L}=\{1,2,\ldots, L\}$. Let $x_0\in \mathbb{R}^d$ denote an input data point consisting of a feature vector of $d$ dimensions. A neural network can be modelled using a classification function $f:\mathbb{R}^d\to \mathbb{R}^L$ whose $\ell$th component $f_\ell(\cdot )$ indicates the belief that a given data point is of label $\ell$.
The network   classifies each input data point $x_0$ as being of  label
\begin{equation}
    c(x_0)=\mbox{argmax}_{\ell\in \{1,\ldots,L\}} f_{\ell}(x_0).
\end{equation}

We let $\mathcal{D}(\ell) \subseteq \mathcal{D}$ denote the set of points in a data set $\mathcal{D}$  that are classified as being of label $\ell$, namely
 \begin{equation}
     \mathcal{D}(\ell)=\{x_0\in\mathcal{D}: c(x_0)=\ell\}. 
 \end{equation}
When the points in this set are subject to an $\epsilon$-attack, they become part of  the perturbed input data set  
 \begin{equation}
 \begin{split}
      S_{\epsilon}(\ell)&=\{x\in\mathbb{R}^d:  \; \|x-x_0\|\leq \epsilon,   
      x_0\in\mathcal{D}(\ell)  \}.
 \end{split}
 \end{equation}

An $\epsilon$-attack on input point $x_0$ is successful if  there exists a noise vector $e\in\mathbb{R}^d$ such that 
\begin{equation}
    c(x_0)\neq c(x_0+e), \mbox{ and } \|e\| \leq\epsilon.
\end{equation}

 For any two labels $\ell_1,\ell_2\in \{1,2,\ldots,L\}$, we let
 \begin{equation}
 \begin{split}
     \mathcal{P}_{\epsilon}(\ell_1,\ell_2)=&\{x\in  S_{\epsilon}(\ell_1): 
     \mbox{ either } c(x)=\ell_1 \\ &\mbox{ or }c(x)=\ell_2\}.
 \end{split}
 \end{equation}
Then,  
$\mathcal{P}_{\epsilon}(\ell_1,\ell_2)\cap \mathcal{P}_{\epsilon}(\ell_2,\ell_1)$ represents the set of points in $\mathcal{S}_{\epsilon}(\ell_1)$ and $\mathcal{S}_{\epsilon}(\ell_2)$ that can lead to a successful  $\epsilon$-attack.

For all $\ell_1,\ell_2\in \{1,2,\ldots,L\}$, we let the error
\begin{equation}
    e_{\epsilon}(\ell_1,\ell_2)=\frac{m_{\mathscr{L}}(\mathcal{P}_{\epsilon}(\ell_1,\ell_2)\cap \mathcal{P}_{\epsilon}(\ell_2,\ell_1))}{m_{\mathscr{L}}(\mathscr{L})}.
\end{equation}
Finally, we say that a subset of labels (viz. a codebook) $\mathcal{L}\subseteq {\mathscr{L}}$ is $(\epsilon,\delta)$-robust if  for all $\ell_1,\ell_2\in\mathcal{L}$, we have $ e_{\epsilon}(\ell_1,\ell_2)\leq \delta/|\mathcal{L}|$. This implies that whenever a 
label in an $(\epsilon,\delta)$-robust codebook is assigned to any input point that is subject to an  attack, this  is the same as the label assigned to the same input in the absence of the attack, with confidence at least $1-\delta$.

We can then define the $(\epsilon,\delta)$- robust capacity of the neural network  as the logarithm of the maximum number of labels that can be robustly classified with confidence $1-\delta$ in the presence of an  $\epsilon$-attack, namely
\begin{equation}
    C^{\delta}_{\epsilon}(\mathcal{D})=\sup_{\mathcal{L}\in\mathscr{L}^{\delta}_{\epsilon}(\mathcal{D})} \log_2(|\mathcal{L}|),
\end{equation}
where $\mathscr{L}^{\delta}_{\epsilon}(\mathcal{D})=\{\mathcal{L}: \mathcal{L} \mbox{ is } (\epsilon,\delta)\mbox{-robust}\}$. 
This capacity represents the largest amount of information that the labeling task of the neural network can convey, at a given level of confidence, under a perturbation attack. This information is independent of whether the neural network classifies the input data correctly or not.

For $\epsilon=0$, for any two labels $\ell_1,\ell_2\in\mathscr{L}$, we have that $\mathcal{D}(\ell_1)=S_{\epsilon}(\ell_1)$ and $\mathcal{P}_{\epsilon}(\ell_1,\ell_2)=\mathcal{D}(\ell_1)$, which implies that 
\begin{equation}
    \mathcal{P}_{\epsilon}(\ell_1,\ell_2)\cap \mathcal{P}_{\epsilon}(\ell_2,\ell_1)=\emptyset,
\end{equation}
and the capacity $C^{\delta}_{\epsilon}(\mathcal{D})=\log_2(L)$, regardless of the value of $\delta$. This means that in the absence an attack, the amount of information conveyed by the network is simply the logarithm of the number of labels it classifies the data into.

The framework described above has been studied in the special case of $\delta=0$ and for a single input data point $x_0$ in \cite{weng2018evaluating}. 
Our $(\epsilon,\delta)$- robust capacity generalizes the notion of robustness
from a single point $x_0$ to the whole data set $\mathcal{D}$ and can quantify   the overall robustness of a neural network. 


\subsection{Performance of Classification Systems}
The non-stochastic $(N,\delta)$ capacity can also be used as a performance measure of classification systems operating on a given data set. Consider a system trained to classify the incoming data among $L$ possible labels in the set $\mathscr{L}=\{1,2,\ldots,L\}$. Let $x_0$ be an input data point, and $c(x_0)$ be the label assigned by the neural network to $x_0$. 
{ For a given data set $\mathcal{D}$, let $N(\ell)\subseteq \{1,2,\ldots,L\}$ 
denote the subset of labels such that
\begin{equation}
\begin{split}
      N(\ell)=&\{\ell^\prime\in\mathscr{L}: \mbox{ there exists a data point } x_0\in \mathcal{D} \\
      &  \mbox{ such that }\mbox{the correct label of $x_0$ is $\ell$ and }\\
      &c(x_0)=\ell^\prime\}.
\end{split}
\end{equation}
  If $\ell\in N(\ell)$, then  there exists a data point that is correctly classified as $\ell$. If $\ell' \in N(\ell)$ such that $\ell' \neq \ell$, then there exists a data point that is incorrectly classified  as $\ell^\prime$, and the correct label of this data point is $\ell$.}

For any two labels $\ell_1,\ell_2\in\mathscr{L}$, we have the equivocation region   $N(\ell_1)\cap N(\ell_2)$, and we let the error
\begin{equation}
    e_{N}(\ell_1,\ell_2)= \frac{|N(\ell_1)\cap N(\ell_2)|}{|\mathscr{L}|}.
\end{equation}
We say that a subset of labels $\mathcal{L}\subseteq\mathscr{L}$ is $(N,\delta)$-classifiable if for all $\ell_1,\ell_2\in \mathcal{L}$, we have $e_{N}(\ell_1,\ell_2)\leq \delta/|\mathcal{L}|$, and the $(N,\delta)$-capacity of the classification system is 
\begin{equation}
    C_{N}^{\delta}=\sup_{\mathcal{L}\in \mathscr{L}^{\delta}_{N}}\log_2(|\mathcal{L}|),
\end{equation}
where $\mathscr{L}^{\delta}_{N}=\{\mathcal{L}:\mathcal{L}\mbox{ is   }(N,\delta)\mbox{-classifiable}\}$. Given the set of labels, this capacity quantifies the amount of information that  the classifier is able to extract from a given   data set, in terms of the logarithm of the largest number of labels that can be identified with  confidence greater than $1-\delta$. In contrast to the robust capacity described in Section~\ref{sec:robustness}, here the capacity refers to the ability of the network to perform the classification task \emph{correctly} in the absence of an attack, rather than to its ability of performing classification \emph{consistently} (but not necessarily correctly) in the presence of an attack.

\section{Conclusion}
In this paper, we presented   a non-stochastic theory of information  that is based on a  notion of information with  worst-case confidence that is independent of   stochastic modeling assumptions. Using the non-stochastic variables  framework of Nair~\cite{nair2012nonstochastic}, we showed that the capacity of several channel models  equals the largest amount of information conveyed by the transmitter to the receiver, with a given level of confidence. These results are the natural generalization of Nair's ones, obtained in a zero-error framework, and provide an information-theoretic interpretation for the geometric problem of sphere packing with overlap, studied by Lim and Franceschetti~\cite{lim2017information}.  More generally, they show that the path laid by Shannon can be extended to a non-stochastic setting, which is an idea that dates back to Kolmogorov~\cite{kolmogorov1956certain}. 

Non-stochastic approaches to information, and their usage to quantify the performance of engineering systems   have recently received   attention  in the context of  estimation,   control, security, communication over non-linear optical channels, and learning systems~\cite{saberi2018estimation,saberi2019state,wiese2016uncertain,borujeny2020signal,weng2018evaluating,verma2019error,ferng2011multi}.  We  hope that the theory developed here can be useful in some of these contexts. To this end, we pointed out some possible applications in the context of classification systems and communication over adversarial channels. 


While refinements and extension of the theory are certainly of interest, further exploration of application domains is of paramount importance. There is evidence in the literature for the need of a non-stochastic approach to study the flow of information in complex systems, and there is   a certain tradition in computer science and especially in the field of  online learning to study various problems in both a stochastic and a non-stochastic   setting \cite{auer2002nonstochastic,agrawal1995sample, rangi2019unifying,rangi2019online,rangi2018multi}. Nevertheless, it seems that only a few isolated efforts have been made  towards the formal development of a non-stochastic information theory. A wider involvement of the community in developing alternative, even competing, theories is certainly advisable to eventually fulfill  the  need of several application areas.

 \bibliographystyle{IEEEtran}
\bibliography{nonStochatsicInfor}
\appendix
\subsection{Proof of Lemma \ref{lemma:Dissassociation}}\label{sec:Lemma1}
\begin{proof}
Let   $(X, Y) \stackrel{a}{ \leftrightarrow } (\delta_{1},\delta_{2})$. Then,
\begin{equation}\label{eq:Lemma1.1}
 \mathscr{A}(X;Y)\preceq\delta_{1},
\end{equation}
\begin{equation}\label{eq:Lemma1.2}
 \mathscr{A}(Y;X)\preceq \delta_{2}.
\end{equation}
Let 
\begin{equation}
    \mathscr{S}_{1}=\bigg\{(y_{1},y_{2}):\frac{m_{\mathscr{X}}(\llbracket X|y_{1}\rrbracket\cap \llbracket X|y_{2}\rrbracket|)}{m_{\mathscr{X}}(\llbracket X\rrbracket)}=0\bigg\}.
\end{equation}
 Then, for all $(y_{1},y_2)\in\mathscr{S}_{1}$, we have
\begin{equation}\label{eq:Lemma1.3}
\frac{m_{\mathscr{X}}(\llbracket X|y_{1}\rrbracket\cap \llbracket X|y_{2}\rrbracket|)}{m_{\mathscr{X}}(\llbracket X\rrbracket)}=0\leq \delta_{1}.
\end{equation}
Also, if $(y_1,y_2)\in \mathscr{S}_{1}$, then 
\begin{equation}
    \frac{m_{\mathscr{X}}(\llbracket X|y_{1}\rrbracket\cap \llbracket X|y_{2}\rrbracket|)}{m_{\mathscr{X}}(\llbracket X\rrbracket)}\notin \mathscr{A}(X;Y),
\end{equation}
and if $(y_1,y_2)\notin \mathscr{S}_{1}$, then using \eqref{ass:finiteness}, we have
\begin{equation}
    \frac{m_{\mathscr{X}}(\llbracket X|y_{1}\rrbracket\cap \llbracket X|y_{2}\rrbracket|)}{m_{\mathscr{X}}(\llbracket X\rrbracket)}\in \mathscr{A}(X;Y).
\end{equation}
 This  along with (\ref{eq:Lemma1.1}) and (\ref{eq:Lemma1.3}) implies that (\ref{eq:DissX1}) follows.

Likewise, let
\begin{equation}
\mathscr{S}_{2}=\bigg\{(x_{1},x_{2}): \frac{m_{\mathscr{Y}}(\llbracket Y|x_{1}\rrbracket\cap \llbracket Y|x_{2}\rrbracket|)}{m_{\mathscr{Y}}(\llbracket Y\rrbracket)}=0\bigg\}.    
\end{equation}
  Then, for all $(x_{1},x_2)\in\mathscr{S}_{2}$,
\begin{equation}\label{eq:Lemma1.4}
\frac{m_{\mathscr{Y}}(\llbracket Y|x_{1}\rrbracket\cap \llbracket Y|x_{2}\rrbracket|)}{m_{\mathscr{Y}}(\llbracket Y\rrbracket)}=0\leq \delta_{2}.
\end{equation}
Also, if $(x_1,x_2)\in \mathscr{S}_{2}$, then 
\begin{equation}
    \frac{m_{\mathscr{Y}}(\llbracket Y|x_{1}\rrbracket\cap \llbracket Y|x_{2}\rrbracket|)}{m_{\mathscr{Y}}(\llbracket Y\rrbracket)}\notin \mathscr{A}(Y;X),
\end{equation}
and if $(y_1,y_2)\notin \mathscr{S}_{2}$, then using \eqref{ass:finiteness}, we have
\begin{equation}
    \frac{m_{\mathscr{Y}}(\llbracket Y|x_{1}\rrbracket\cap \llbracket Y|x_{2}\rrbracket|)}{m_{\mathscr{Y}}(\llbracket Y\rrbracket)}\in \mathscr{A}(Y;X).
\end{equation}
This  along with (\ref{eq:Lemma1.2}) and (\ref{eq:Lemma1.4}) implies that (\ref{eq:DissX}) follows.

Now, we prove the opposite direction of the statement.   Given that for all $y_{1},y_{2}\in \llbracket Y\rrbracket$, we have
\begin{equation}\label{eq:Lemma1.5}
\frac{m_{\mathscr{X}}(\llbracket X|y_{1}\rrbracket\cap \llbracket X|y_{2}\rrbracket|)}{m_{\mathscr{X}}(\llbracket X\rrbracket)}\leq \delta_{1},
\end{equation}
and for all  $x_{1},x_{2}\in \llbracket X\rrbracket$, we have
\begin{equation}\label{eq:Lemma1.6}
\frac{m_{\mathscr{Y}}(\llbracket Y|x_{1}\rrbracket\cap \llbracket Y|x_{2}\rrbracket|)}{m_{\mathscr{Y}}(\llbracket Y\rrbracket)}\leq \delta_{2}.
\end{equation}
Then, using the definition of $\mathscr{A}(X;Y)$ and $\mathscr{A}(Y;X)$, we have
\begin{equation}\label{eq:Lemma1.7}
 \mathscr{A}(X;Y)\preceq\delta_{1},
\end{equation}
\begin{equation}\label{eq:Lemma1.8}
 \mathscr{A}(Y;X)\preceq \delta_{2}.
\end{equation}
The statement of the lemma follows.
\end{proof}
\subsection{Proof of Theorem \ref{thm:relBwRatenInfoRate}}\label{sec:Theorem9}
\begin{proof} We will show \eqref{eq:optimizationAtEachTimeN}.  Then, using Lemma \ref{lemma:cardinalityAssociation} in   Appendix \ref{sec:AuxResult}, 
\eqref{eq:RateNInfoRate}  follows using the same argument as in the proof of Theorem \ref{cor3}.

We proceed in three steps. 
First, we show that for all $n>0$, there exists an UV $X(1:n)$ and $\tilde{\delta}\leq \delta_n/m_{\mathscr{Y}}(\llbracket Y(1:n)\rrbracket)$ such that ${X}(1:n)\in\mathscr{F}_{\tilde{\delta}}(n)$, which implies $\mathscr{F}_{\tilde{\delta}}(n)$ is not empty, so that the supremum is well defined.
Second, for all $n>0$, and $X(1:n)$ and $\tilde{\delta}$ such that 
\begin{equation}
    X(1:n)\in \mathscr{F}_{\tilde \delta}(n),
\end{equation}
and
\begin{equation}
    \tilde{\delta}\leq \delta_n/m_{\mathscr{Y}}(\llbracket Y(1:n)\rrbracket),
\end{equation}
we show that 
\begin{equation*}
\frac{I_{\tilde{\delta}/|\llbracket X(1:n)\rrbracket|}(Y(1:n);X(1:n))}{n}   \leq R_{\delta_n}.
\end{equation*}
Finally,  for all $n>0$, we show the existence of  $X(1:n)\in  \mathscr{F}_{\tilde{\delta}}(n)$ and $\tilde\delta\leq\delta_n/m_{\mathscr{Y}}(\llbracket Y(1:n)\rrbracket)$ such that 
\begin{equation}
    \frac{I_{\tilde{\delta}/|\llbracket X(1:n)\rrbracket|}(Y(1:n);X(1:n))}{n} =R_{\delta_n}.
\end{equation}

Let us begin with the first step. Consider a point $x(1:n)\in\mathscr{X}^n$. Let $X(1:n)$ be a UV such that 
\begin{equation}\label{eq:Xval_1}
    \llbracket X(1:n)\rrbracket=\{x(1:n)\}.
\end{equation}
Then, we have that the marginal range of the UV $Y(1:n)$ corresponding to the received variable
is 
\begin{equation}\label{eq:Yval_1}
    \llbracket Y(1:n)\rrbracket=\llbracket Y(1:n)|x(1:n)\rrbracket,
\end{equation}
and therefore for all $y(1:n)\in\llbracket Y(1:n)\rrbracket$, we have 
\begin{equation}\label{eq:Yval_1.1}
    \llbracket X(1:n)|y(1:n)\rrbracket=\{x(1:n)\}.
\end{equation}
Using Definition \ref{def:AssSets} and \eqref{eq:Xval_1}, we have that
\begin{equation}\label{eq:Xsetdef_1}
    \mathscr{A}(Y(1:n);X(1:n))=\emptyset,
\end{equation}
because  $\llbracket X(1:n)\rrbracket$ consists of a single point, and therefore the set in \eqref{eqdef1} is empty. 

On the other hand, using Definition \ref{def:AssSets} and \eqref{eq:Yval_1.1}, we have
\begin{equation}\label{eq:Ysetdef_1}
    \mathscr{A}(X(1:n);Y(1:n))
    =\begin{cases}
    \{1\}  \mbox{ if } \exists y_1(1:n), y_2(1:n) \\
    \qquad \qquad \in \llbracket Y(1:n)\rrbracket,\\
    \emptyset \quad\mbox{ otherwise}.
    \end{cases}
\end{equation}
Using \eqref{eq:Xsetdef_1} and since $\mathscr{A}\preceq \delta$ holds for $\mathscr{A} = \emptyset$, we have 
\begin{equation}\label{eq:cond1_1_1}
    \mathscr{A}(Y(1:n);X(1:n))\preceq \delta_n/(|\llbracket X(1:n)\rrbracket|m_{\mathscr{Y}}(\llbracket Y(1:n)\rrbracket)).
\end{equation}
Similarly, using \eqref{eq:Ysetdef_1}   we have 
\begin{equation}\label{eq:cond1_2_1}
    \mathscr{A}(X(1:n);Y(1:n))\preceq 1.
\end{equation}
Now, combining \eqref{eq:cond1_1_1}  and \eqref{eq:cond1_2_1}, we have 
\begin{equation}
    (X(1:n),Y(1:n))\stackrel{a}{\leftrightarrow}(1,\delta_n/(|\llbracket X(1:n)\rrbracket|m_{\mathscr{Y}}(\llbracket Y(1:n)\rrbracket))).
\end{equation}
Letting $\tilde{\delta}=\delta_n/m_{\mathscr{Y}}(\llbracket Y(1:n)\rrbracket)$, this implies that  $X(1:n)\in \mathscr{F}_{\tilde\delta}(n)$ and the first step of the proof is complete. 

To prove the second step, we define 
\begin{equation}
\begin{split}
    \mathscr{G}(n)=&\bigg\{X(1:n):\llbracket X(1:n)\rrbracket\subseteq\mathscr{X}^n,\\
    &\exists \tilde{\delta}\leq \delta_n/m_{\mathscr{Y}}(\llbracket Y(1:n)\rrbracket) \mbox{ such that } \\
    &\forall \mathscr{S}_{1},\mathscr{S}_{2}\in \llbracket Y(1:n)|X(1:n)\rrbracket,\\
    &  \frac{m_{\mathscr{Y}}(\mathscr{S}_{1}\cap \mathscr{S}_{2})}{m_{\mathscr{Y}}(\llbracket Y(1:n)\rrbracket)}\leq \frac{\tilde{\delta}}{|\llbracket X(1:n)\rrbracket|}\bigg\},
\end{split}
\end{equation}
which is a larger set than the one containing all UVs $X(1:n)$ that are $(1,\tilde{\delta}/|\llbracket X(1:n)\rrbracket|)$ associated to $Y(1:n)$. Similar to \eqref{eq:TwoSetsEqual_1}, it can be shown that 
\begin{equation}\label{eq:TwoSetsEqual2}
  X(1:n)\in \mathscr{G}(n) \implies\mathcal{X}(1:n)\in \mathscr{X}_{N}^{\delta_n}(n)
\end{equation}
Consider now a pair   $X(1:n)$ and $\tilde{\delta}$ such that $\tilde{\delta}\leq\delta_n/m_{\mathscr{Y}}(\llbracket Y(1:n)\rrbracket)$, and
\begin{equation}\label{eq:setin1}
     X(1:n) \in \mathscr{F}_{\tilde{\delta}}(n).
\end{equation}
If $(X(1:n),Y(1:n))\stackrel{d}{\leftrightarrow}(0, \tilde{\delta}/|\llbracket X(1:n)\rrbracket|)$, then  using Lemma \ref{lemma:equivalenceAssDiss} in   Appendix \ref{sec:AuxResult}, there exist  UVs $\bar{X}(1:n)$ and $\bar{Y}(1:n)$ and $\bar{\delta}\leq \delta_n/m_{\mathscr{Y}}(\llbracket\bar{Y}(1:n)\rrbracket)$ such that
\begin{equation}\label{eq:chec12}
    (\bar X(1:n), \Bar{Y}(1:n))  
\stackrel{a} \leftrightarrow (1,\bar\delta/|\llbracket\bar{X}(1:n)\rrbracket|),
\end{equation}  and
\begin{equation}\label{eq:chec13}
\begin{split}
     &|\llbracket Y(1:n)|X(1:n)\rrbracket^*_{\tilde\delta/|\llbracket X(1:n)\rrbracket|}|\\
     &= |\llbracket \Bar Y(1:n)|\bar{X}(1:n) \rrbracket_{\bar\delta/|\llbracket\bar{X}(1:n)\rrbracket|}^*|.
\end{split}
\end{equation}
On the other hand, if $ ( X(1:n), {Y}(1:n))  
\stackrel{a} \leftrightarrow (1,\tilde\delta/|\llbracket{X}(1:n)\rrbracket|)$, then \eqref{eq:chec12} and \eqref{eq:chec13} also trivially hold.
It then follows that \eqref{eq:chec12} and \eqref{eq:chec13} hold for all $X(1:n)\in \mathscr{F}_{\tilde{\delta}}(n)$.
We now have
\begin{equation}
\begin{split}
&I_{\tilde{\delta}/|\llbracket X(1:n)\rrbracket|}(Y(1:n);X(1:n))\\
&=\log(|\llbracket Y(1:n)|X(1:n)\rrbracket_{\tilde{\delta}/|\llbracket X(1:n)\rrbracket|}^*|)\\
&\stackrel{(a)}{=}\log(|\llbracket \bar Y(1:n)|\bar X(1:n) \rrbracket_{\bar{\delta}/|\llbracket \bar X(1:n)\rrbracket|}^*|)\\
&\stackrel{(b)}{\leq} \log(|\llbracket \bar X(1:n)\rrbracket|),\\
&\stackrel{(c)}{=} \log(|\bar{\mathcal{X}}(1:n)|),\\
&\stackrel{(d)}{\leq}n R_{\delta_n},
\end{split}
\end{equation}
where $(a)$ follows from  \eqref{eq:chec12} and \eqref{eq:chec13}, $(b)$ follows from Lemma \ref{lemma:RelationToCardinality} in   Appendix \ref{sec:AuxResult} since $\bar{\delta}\leq\delta_n/m_{\mathscr{Y}}(\llbracket \bar Y(1:n)\rrbracket) < m_{\mathscr{Y}}(V^n_{N})/m_{\mathscr{Y}}(\llbracket \bar Y(1:n)\rrbracket)$, $(c)$ follows by defining the codebook $\bar{\mathcal{X}}(1:n)$ corresponding to the UV $\llbracket \bar X(1:n) \rrbracket$ , and $(d)$ follows from the fact that using \eqref{eq:chec12} and Lemma \ref{lemma:Dissassociation}, we have  $\bar{X}(1:n)\in \mathcal{G}(n)$, which implies by  \eqref{eq:TwoSetsEqual2} that $\bar{\mathcal{X}}(1:n)\in\mathscr{X}_{N}^{\delta_n}(n)$.

For any $n \in \mathbb{Z}_{>0}$,  let 
\begin{equation}
\mathcal{X}_n^* =\mbox{argsup}_{\mathcal{X}_n\in\mathscr{X}^{\delta_n}_{N}(n)}\frac{\log(|\mathcal{X}_n|)}{n},    
\end{equation}
which  achieves the rate $R_{\delta_n}$. 
Let ${{X}}^*$ be the UV whose marginal range corresponds to the codebook $\mathcal{{X}}_n^*$.
It follows that for all $\mathscr{S}_{1},\mathscr{S}_{1}\in \llbracket Y^*|{X}^*\rrbracket$, we have 
\begin{equation}
    \frac{m_{\mathscr{Y}}(\mathscr{S}_{1}\cap \mathscr{S}_{1})}{m_{\mathscr{Y}}(\mathscr{Y}^n)}\leq \frac{\delta_n}{|\llbracket{X}^*\rrbracket|},
\end{equation}
which implies using the fact that $m_{\mathscr{Y}}(\mathscr{Y}^n)=1$,
\begin{equation}
    \frac{m_{\mathscr{Y}}(\mathscr{S}_{1}\cap \mathscr{S}_{1})}{m_{\mathscr{Y}}(\llbracket Y^*\rrbracket)}\leq \frac{\delta_n}{(|\llbracket{X}^*\rrbracket| m_{\mathscr{Y}}(\llbracket Y^*\rrbracket))}.
\end{equation}
Letting $\delta^*=\delta_n/m_{\mathscr{Y}}(\llbracket Y^*\rrbracket)$, and
using Lemma \ref{lemma:Dissassociation},  we have that $(X^*,{Y}^*) \stackrel{a}\leftrightarrow (1,\delta^*/|\llbracket X^*\rrbracket|)$, which implies
\begin{equation}
{X}^*\in  \cup_{\tilde{\delta}\leq\delta_n/m_{\mathscr{Y}}(\llbracket Y^*\rrbracket)}\mathscr{F}_{\tilde{\delta}}(n),
\end{equation}
and \eqref{eq:optimizationAtEachTimeN} follows.
\end{proof}
\subsection{Proof of Lemma \ref{thm:tensorization}}\label{appendix:lemma2Proof}
 \begin{proof} 
Let us begin with part $1)$. 
We have
\begin{equation}\label{eq:reference2}
\begin{split}
      &\llbracket Y(1:n)\rrbracket\\
      &\stackrel{}{=} \cup_{x(1:n)\in \llbracket X(1:n)\rrbracket}\llbracket Y(1:n)|x(1:n)\rrbracket,\\
      &\stackrel{(a)}{=}\cup_{x(1:n)\in \llbracket X(1:n)\rrbracket} \llbracket Y(1)|x(1)\rrbracket\times\ldots \llbracket Y(n)|x(n)\rrbracket,\\
      &\stackrel{(b)}{=}\cup_{x(1)\in \llbracket X(1)\rrbracket}\llbracket Y(1)|x(1)\rrbracket\times\\
      &\quad\ldots \times\cup_{x(n)\in \llbracket X(n)\rrbracket}\llbracket Y(n)|x(n)\rrbracket,\\
      &=\llbracket Y(1)\rrbracket\times \llbracket Y(2)\rrbracket\times \ldots \times\llbracket Y(n)\rrbracket,\\
\end{split}
\end{equation}
where 
$(a)$ follows from \eqref{eq:reference3}, and $(b)$ follows from \eqref{eq:reference1}.
Now, we have 
\begin{equation}
\begin{split}
\bigcup_{\mathscr{S}\in \prod_{i=1}^{n}\llbracket Y(i)|X(i)\rrbracket^*_{\delta}} \mathscr{S}&
\stackrel{(a)}{=} \prod_{i=1}^{n} (\bigcup_{\mathscr{S}\in \llbracket Y(i)|X(i)\rrbracket^*_{\delta}} \mathscr{S} ),\\
&\stackrel{(b)}{=} \prod_{i=1}^{n} \llbracket Y(i)\rrbracket,\\
&\stackrel{(c)}{=}  \llbracket Y(1:n)\rrbracket,
\end{split}
\end{equation}
where $(a)$ follows from the fact that the cartesian product is distributive over union, namely
\begin{equation}
    \cup_{(i,j)\in I\times J} A_i\times B_j = (\cup_{i\in I} A_i)\times (\cup_{j\in J} B_j),
\end{equation}
$(b)$ follows from the fact that  for all $1\leq i\leq n$, $\llbracket Y(i)|X(i)\rrbracket^*_{\delta}$  is a covering of $\llbracket Y(i)\rrbracket$ by Definition \ref{defoverlap}, and $(c)$ follows from \eqref{eq:reference2}. Hence, part $1)$ follows.


Now, we prove part $2)$. Here, we will first show that for all $\mathscr{S}\in \prod_{i=1}^{n}\llbracket Y(i)|X(i)\rrbracket^*_{\delta}$, we have that $\mathscr{S}$ is $\delta^n$-connected. Second, we will show that $\mathscr{S}$ contains at least one singly $\delta^n$-connected set.

Let us begin with the first step of part $2)$. Consider a set $\mathscr{S}\in \prod_{i=1}^{n}\llbracket Y(i)|X(i)\rrbracket^*_{\delta}$. Then, there exists a sequence $\{\mathscr S_i\}_{i=1}^n$ such that 
\begin{equation}\label{eq:Subset2.0}
  \mathscr S=\mathscr S_1\times \mathscr S_2 \times \ldots \times\mathscr S_n, 
\end{equation}
and for all $1\leq i\leq n$,
\begin{equation}\label{eq:exist2}
    \mathscr S_i\in \llbracket Y(i)|X(i)\rrbracket^*_{\delta}.
\end{equation}
Now, consider two points $y_1(1:n),y_2(1:n)\in \mathscr S$. Then, using \eqref{eq:reference3}, \eqref{eq:reference2} and \eqref{eq:Subset2.0}, for all $1\leq i\leq n$, we have that 
\begin{equation}
    y_1(i),y_2(i)\in \mathscr S_i.
\end{equation}
Also, since $\mathscr S_i$ is $\delta$-connected using \eqref{eq:exist2} and Property 1 of Definition \ref{defoverlap}, we have
\begin{equation}
    y_1(i)\stackrel{\delta}{\leftrightsquigarrow}y_2(i),
\end{equation}
 namely there exists a sequence $\{\llbracket Y(i)|x_{k}(i)\rrbracket\}_{k=1}^{N(i)}$  such that  
 \begin{equation}\label{eq:insideOverlap}
    y_1(i)\in  \llbracket Y(i)|x_{1}(i)\rrbracket,    y_2(i)\in  \llbracket Y(i)|x_{N(i)}(i)\rrbracket,
 \end{equation}
 and 
 for all $1\leq k< N(i)$,
\begin{equation}\label{eq:consecutiveintersection}
    \frac{m_{\mathscr{Y}}(\llbracket Y(i)|x_{k}(i)\rrbracket\cap \llbracket Y(i)|x_{k+1}(i)\rrbracket)}{m_{\mathscr{Y}}(\llbracket Y(i)\rrbracket)}>\delta.
\end{equation}
Without loss of generality, let 
\begin{equation}
    N(1)\leq N(2)\leq \ldots\leq N(n).
\end{equation}
Now, consider the following sequence of conditional ranges
\begin{equation}\label{eq:constructionOfPath}
\begin{split}
    &\llbracket Y(1:n)|x_{1}(1),x_{1}(2)\ldots x_{1}(n)\rrbracket,\\
    &\llbracket Y(1:n)|x_{2}(1),x_{2}(2)\ldots x_{2}(n)\rrbracket,\\
    & \ldots \\
    &\llbracket Y(1:n)|x_{N(1)}(1),x_{N(1)}(2)\ldots x_{N(1)}(n)\rrbracket,\\
    &\llbracket Y(1:n)|x_{N(1)}(1),x_{N(1)+1}(2)\ldots x_{N(1)+1}(n)\rrbracket,\\
    & \ldots \\
    &\llbracket Y(1:n)|x_{N(1)}(1),x_{N(2)}(2)\ldots x_{N(n)}(n)\rrbracket.\\
\end{split}
\end{equation}
In this sequence, for all $1\leq k < N(n)$, if $x_k(i)=x_{k+1}(i)$, then we have
 \begin{equation}\label{eq:equalIf}
 \begin{split}
     &\frac{ m_{\mathscr{Y}}(\llbracket Y(i)|x_{k}(i)\rrbracket\cap \llbracket Y(i)|x_{k+1}(i)\rrbracket )}{m_{\mathscr{Y}}(\llbracket Y(i)\rrbracket)}\\
     &\stackrel{(a)}{=}\frac{ m_{\mathscr{Y}}(\llbracket Y(i)|x_{k}(i)\rrbracket )}{m_{\mathscr{Y}}(\llbracket Y(i)\rrbracket)},\\
     &\stackrel{(b)}{>} \frac{ \delta}{m_{\mathscr{Y}}(\llbracket Y(i)\rrbracket)},\\
     &\stackrel{(c)}{\geq}
     \frac{\delta}{m_{\mathscr{Y}}(\mathscr{Y})},\\
     &\stackrel{(d)}{>} \delta, 
 \end{split}
     \end{equation}
where $(a)$ follows from the fact that $x_k(i)=x_{k+1}(i)$, $(b)$ follows from \eqref{eq:rangeDeltaCheck}, $(c)$ follows from the fact  that $\llbracket Y(i)\rrbracket\subseteq\mathscr{Y}$ and \eqref{eq:strongtransitivity} holds, and $(d)$ follows from the fact that $m_{\mathscr{Y}}(\mathscr{Y})=1$. Additionally, in the sequence \eqref{eq:constructionOfPath}, for all $1\leq k < N(n)$, if $x_k(i)\neq x_{k+1}(i)$, then we have  that (\ref{eq:consecutiveintersection}) holds.  This along with  \eqref{eq:equalIf} implies that  for all $1\leq k < N(n)$, we have
\begin{equation}\label{eq:lowerBoundDElta}
    \frac{ m_{\mathscr{Y}}(\llbracket Y(i)|x_{k}(i)\rrbracket\cap \llbracket Y(i)|x_{k+1}(i)\rrbracket )}{m_{\mathscr{Y}}(\llbracket Y(i)\rrbracket)}>\delta.
\end{equation}

Now, using \eqref{eq:reference3} and \eqref{eq:insideOverlap}, we have 
\begin{equation}\label{eq:startingPoint}
    y_{1}(1:n)\in \llbracket Y(1:n)|x_{1}(1),\ldots x_{1}(n)\rrbracket,
\end{equation}
and 
\begin{equation}\label{eq:endingPoint}
    y_2(1:n)\in \llbracket Y(1:n)|x_{N(1)}(1),\ldots x_{N(n)}(n)\rrbracket.
\end{equation}
Also, for all $1\leq k< N(n)$, the uncertainty associated with the intersection of the two consecutive conditional ranges in the sequence \eqref{eq:constructionOfPath} is
{
\begin{equation}\label{eq:intersectionLowBound}
\begin{split}
      &\frac{m_{\mathscr{Y}}(\llbracket Y(1:n)|x_{k}(1:n)\rrbracket\cap \llbracket Y(1:n)|x_{k+1}(1:n)\rrbracket)}{m_{\mathscr{Y}}(\llbracket Y(1:n)\rrbracket)},\\
      &\stackrel{(a)}{=}\prod_{i=1}^{n}\frac{ m_{\mathscr{Y}}(\llbracket Y(i)|x_{k}(i)\rrbracket\cap \llbracket Y(i)|x_{k+1}(i)\rrbracket )}{m_{\mathscr{Y}}(\llbracket Y(i)\rrbracket)},\\
      &\stackrel{(b)}{>}\delta^n,
\end{split}
\end{equation} }
where $(a)$ follows from Assumption \ref{assumption:productrule}, \eqref{eq:reference3}, \eqref{eq:reference2}  and the fact that
\begin{equation}\label{eq:productrule}
    \prod_{i=1}^n \mathscr S_i\cap \prod_{i=1}^n \mathscr T_i=(\mathscr S_1\cap \mathscr T_1)\times \ldots\times (\mathscr S_n\cap \mathscr T_n),
\end{equation}
 $(b)$ follows from \eqref{eq:lowerBoundDElta}.
Hence, using \eqref{eq:startingPoint}, \eqref{eq:endingPoint} and \eqref{eq:intersectionLowBound}, we have
\begin{equation}
    y_{1}(1:n) \stackrel{\delta^n}{\leftrightsquigarrow} y_{2}(1:n).
\end{equation} Hence, $\mathscr S$ is $\delta^n$-connected. 

Now, let us prove the second step of part $2)$. For all $1\leq i\leq n$, since $\llbracket Y(i)|X(i)\rrbracket^*_{\delta}$ satisfies Property 1 of Definition \ref{defoverlap}, there exists an $x_i\in\llbracket X(i)\rrbracket$ such that 
\begin{equation}\label{eq:subset1}
\llbracket Y(i)|x_i\rrbracket\subseteq \mathscr S_i.
\end{equation}
Therefore, for $x(1:n)=[x_1,x_2,\ldots ,x_n]$, we have
\begin{equation}
\begin{split}
&\llbracket Y(1:n)|x(1:n)\rrbracket\\
&\stackrel{(a)}{=} \llbracket Y(1)|x(1)\rrbracket\times\ldots\times\llbracket Y(n)|x(n)\rrbracket,\\
&\stackrel{(b)}{=}\llbracket Y(1)|x_1\rrbracket\times\ldots\times\llbracket Y(n)|x_n\rrbracket,\\
&\stackrel{(c)}{\subseteq}\mathscr S_1\times \mathscr S_2 \times \ldots \mathscr S_n,\\
&\stackrel{(d)}{=} \mathscr S,
\end{split}
\end{equation}
where $(a)$ follows from  \eqref{eq:reference3}, $(b)$ follows from the fact that $x(1:n)=[x_1,x_2,\ldots ,x_n]$, $(c)$ follows from \eqref{eq:subset1}, and $(d)$ follows from \eqref{eq:Subset2.0}. Hence, $\mathscr{S}$ contains at least one singly $\delta^n$-connected set, which concludes the second step of part $2)$.   

Now, let us prove part $3)$. For all $1\leq i\leq n$,  
since $\llbracket Y(i)|X(i)\rrbracket^*_{\delta}$ satisfies Property 3 of Definition \ref{defoverlap}, for all $x(i)\in\llbracket X(i)\rrbracket$, 
there exist a set $\mathscr{S}(x(i))\in \llbracket Y(i)|X(i)\rrbracket^*_{\delta}$ such that 
\begin{equation}\label{eq:check13}
    \llbracket Y(i)|x(i)\rrbracket \subseteq \mathscr{S}(x(i)).
\end{equation}
Then    for all $x(1:n)\in\llbracket X(1:n)\rrbracket$, we have  
\begin{equation}
\begin{split}
      &\llbracket Y(1:n)|x(1:n)\rrbracket\\
      &\stackrel{(a)}{=} \llbracket Y(1)|x(1)\rrbracket\times\ldots\times\llbracket Y(n)|x(n)\rrbracket,\\
      &\stackrel{(b)}{\subseteq} \mathscr{S}(x(1))\times\ldots\times \mathscr{S}(x(n))\\
      &\in \prod_{i=1}^{n}\llbracket Y(i)|X(i)\rrbracket^*_{\delta},
\end{split}
\end{equation}
where $(a)$ follows from \eqref{eq:reference3}, and $(b)$ follows from \eqref{eq:check13}. 
Hence, part $3)$ follows.

Finally, let us prove part $4)$. Consider two distinct sets $\mathscr{S}_1,\mathscr{S}_2\in \prod_{i=1}^{n}\llbracket Y(i)|X(i)\rrbracket^*_{\delta}$. Then, we have
\begin{equation}\label{eq:def1.1}
    \mathscr{S}_1=\mathscr{S}_{11}\times \mathscr{S}_{12}\times\ldots \mathscr{S}_{1n},
\end{equation}
\begin{equation}\label{eq:def1.2}
    \mathscr S_2=\mathscr S_{21}\times \mathscr S_{22}\times\ldots \mathscr S_{2n},
\end{equation}
where for all $1\leq i\leq n$
\begin{equation}\label{eq:subs1}
    \mathscr S_{1i},\mathscr S_{2i}\in  \llbracket Y(i)|X(i)\rrbracket^*_{\delta}.
\end{equation}
 Since $\mathscr{S}_1\neq \mathscr{S}_2$, there exists   $1\leq i^*\leq n$ such that 
 \begin{equation}
     \mathscr{S}_{1i^*}\neq \mathscr{S}_{2i^*}.
 \end{equation}
  Then, by Property 2 of Definition \ref{defoverlap} and \eqref{eq:subs1}, we have   
\begin{equation}\label{eq:OneUpperBound}
    \frac{m_{\mathscr{Y}}(\mathscr S_{1i^*}\cap \mathscr S_{2i^*})}{m_{\mathscr{Y}}(\llbracket Y(i^*)\rrbracket)} \leq \delta.
\end{equation}
Also, using \eqref{eq:upperBoundDeltahat}, we have that for all $1\leq i\leq n$,
\begin{equation}\label{eq:TwoUpperBound}
    \frac{m_{\mathscr{Y}}(\mathscr{S}_{1i})}{m_{\mathscr{Y}}(\llbracket Y(i)\rrbracket)}\leq \hat{\delta}(n).
\end{equation}
Then, we have
\begin{equation}
\begin{split}
    &\frac{m_{\mathscr{Y}}(\mathscr{S}_1\cap \mathscr{S}_2)}{m_{\mathscr{Y}}(\llbracket Y(1:n)\rrbracket)}\\
    &\stackrel{(a)}{=}\frac{m_{\mathscr{Y}}((\mathscr S_{11}\times\ldots \mathscr S_{1n})\cap(\mathscr S_{21}\times\ldots \mathscr S_{2n}))}{m_{\mathscr{Y}}(\llbracket Y(1:n)\rrbracket)},\\
    &\stackrel{(b)}{=} \frac{m_{\mathscr{Y}}((\mathscr S_{11}\cap \mathscr S_{21})\times \ldots \times(\mathscr S_{1n}\cap \mathscr S_{2n}))}{m_{\mathscr{Y}}(\llbracket Y(1:n)\rrbracket)},\\
    &\stackrel{(c)}{=} \frac{m_{\mathscr{Y}}(\mathscr S_{11}\cap \mathscr S_{21})  \ldots  m_{\mathscr{Y}}(\mathscr S_{1n}\cap \mathscr S_{2n})}{\prod_{i=1}^n m_{\mathscr{Y}}(\llbracket Y(i)\rrbracket)},\\
    &\stackrel{(d)}{\leq}\frac{m_{\mathscr{Y}}(\mathscr S_{1i^*}\cap \mathscr S_{2i^*})}{m_{\mathscr{Y}}(\llbracket Y(i^*)\rrbracket)}\prod_{i\neq i^*}\frac{m_{\mathscr{Y}}(\mathscr{S}_{1i})}{m_{\mathscr{Y}}(\llbracket Y(i)\rrbracket)}\\
    &\stackrel{(e)}{\leq} \delta (\hat{\delta}(n))^{n-1},\\
\end{split}
\end{equation}
where $(a)$ follows from \eqref{eq:def1.1} and \eqref{eq:def1.2}, $(b)$ follows from the  fact that for all sequences of sets $\{\mathscr S_i\}_{i=1}^n$ and $\{\mathscr T_i\}_{i=1}^n$ , we have
\begin{equation}
    \prod_{i=1}^n \mathscr S_i\cap \prod_{i=1}^n \mathscr T_i=(\mathscr S_1\cap \mathscr T_1)\times \ldots\times(\mathscr S_n\cap \mathscr T_n),
\end{equation}
$(c)$ follows from  Assumption \ref{assumption:productrule} and \eqref{eq:reference2}, $(d)$ follows from \eqref{eq:strongtransitivity} and the fact that for all $1\leq i\leq n$, $\mathscr{S}_{1i}\cap \mathscr{S}_{2i}\subseteq \mathscr{S}_{1i}$, and $(e)$ follows from \eqref{eq:OneUpperBound} and \eqref{eq:TwoUpperBound}.
Hence, part $4)$ follows. 
\end{proof}
\subsection{Auxiliary Results}\label{sec:AuxResult}
\begin{lemma}\label{lemma:equivalenceAssDiss}
Given a $\delta<m_{\mathscr{Y}}(V_N)$, two UVs $X$ and $Y$ satisfying \eqref{uno} and \eqref{due}, and  
a  $\tilde{\delta}\leq\delta/m_{\mathscr{Y}}(\llbracket Y\rrbracket)$ such that 
\begin{equation}
    (X,Y)\stackrel{d}{\leftrightarrow}(0,\tilde\delta/|\llbracket X\rrbracket|).
\end{equation}
 Then, there exists  two UVs $\Bar{X}$  and $\Bar{Y}$ satisfying \eqref{uno} and \eqref{due}, and there exists a $\bar{\delta}\leq\delta/m_{\mathscr{Y}}(\llbracket\bar{Y}\rrbracket)$ such that 
\begin{equation}
    (\bar X, \Bar{Y})  
\stackrel{a} \leftrightarrow (1,\bar\delta/|\llbracket\bar{X}\rrbracket|),
\end{equation}  and
\begin{equation}
    |\llbracket Y|X\rrbracket^*_{\tilde\delta/|\llbracket X\rrbracket|}|= |\llbracket \Bar Y|\bar{X} \rrbracket_{\bar\delta/|\llbracket\bar{X}\rrbracket|}^*|.
\end{equation}
\end{lemma} 
\begin{proof}
Let the cardinality
\begin{equation}\label{eq:count}
    |\llbracket Y|X\rrbracket^*_{\tilde\delta/|\llbracket X\rrbracket|}|=K. 
\end{equation}
By Property 1 of Definition \ref{defoverlap}, we have that 
for all $\mathscr{S}_{i}\in \llbracket Y|X\rrbracket^*_{\tilde\delta/|\llbracket X\rrbracket|}$, there exists a $x_{i}\in \llbracket X\rrbracket$ such that $\llbracket Y|x_{i}\rrbracket\subseteq \mathscr{S}_{i}$. 
Now, consider a new UV $\bar{X}$ whose marginal range is composed of $K$ elements of $\llbracket X\rrbracket$, namely
\begin{equation}\label{eq:NCodebook2}
    \llbracket\bar{X}\rrbracket=\{x_1,x_2,\ldots,x_K\}. 
\end{equation}
Let $\bar{Y}$  be the UV  corresponding to the received variable. 
Using the fact that for all $x\in\mathscr{X}$, we have $\llbracket \Bar{Y}|x\rrbracket=\llbracket{Y}|x\rrbracket$ since \eqref{uno} holds,   and using Property 2 of Definition~\ref{defoverlap}, for all $x,x^\prime\in\llbracket\bar{X}\rrbracket$, we have  
\begin{equation}
\begin{split}
     \frac{ m_{\mathscr{Y}}(\llbracket\bar{Y}|x\rrbracket\cap \llbracket\bar{Y}|x^\prime\rrbracket)}{m_{\mathscr{Y}}(\llbracket Y\rrbracket)}&\leq \frac{\tilde{\delta}}{|\llbracket X\rrbracket|},\\
      &\stackrel{(a)}{\leq} \frac{\tilde{\delta}}{|\llbracket\bar{X}\rrbracket|},
\end{split}
\end{equation}
where $(a)$ follows from the fact that $\llbracket\bar{X}\rrbracket\subseteq\llbracket X\rrbracket$ using \eqref{eq:NCodebook2}. Then, for all $x,x^\prime\in\llbracket\bar{X}\rrbracket$,  we have that
\begin{equation}
\begin{split}
      \frac{m_{\mathscr{Y}}(\llbracket\bar{Y}|x\rrbracket\cap \llbracket\bar{Y}|x^\prime\rrbracket)}{m_{\mathscr{Y}}(\llbracket \bar{Y}\rrbracket)}&\leq \frac{\tilde{\delta} m_{\mathscr{Y}}(\llbracket{Y}\rrbracket)}{|\llbracket X\rrbracket| m_{\mathscr{Y}}(\llbracket\bar{Y}\rrbracket)},\\
      &\stackrel{(a)}{\leq} \frac{\bar{\delta}}{|\llbracket\bar{X}\rrbracket|},
\end{split}
\end{equation}
where $\bar{\delta}=\tilde{\delta} m_{\mathscr{Y}}(\llbracket{Y}\rrbracket)/m_{\mathscr{Y}}(\llbracket\bar{Y}\rrbracket)$. Then, by Lemma \ref{lemma:Dissassociation} it follows that
\begin{equation}\label{eq:dis11}
    (\bar{X},\bar{Y})\stackrel{a}{\leftrightarrow}(1,\bar{\delta}/|\llbracket\bar{X}\rrbracket|).
\end{equation}
Since $\tilde{\delta}\leq\delta/m_{\mathscr{Y}}(\llbracket Y\rrbracket)$, we have
\begin{equation}\label{eq:dis21}
    \bar{\delta}\leq \delta/ m_{\mathscr{Y}}(\llbracket\bar{Y}\rrbracket)<m_{\mathscr{Y}}(V_\epsilon)/ m_{\mathscr{Y}}(\llbracket\bar{Y}\rrbracket).
\end{equation}
Using  \eqref{eq:dis11} and \eqref{eq:dis21}, we now have that 
 \begin{equation}
 \begin{split}
     |\llbracket \bar Y|\bar{X}\rrbracket^{*}_{\bar{\delta}/|\llbracket \bar X\rrbracket|}| &\stackrel{(a)}{=}|\llbracket\bar{X}\rrbracket| \\
&     \stackrel{(b)}{=}  | \llbracket Y|X\rrbracket_{{\tilde{\delta}}/|\llbracket X\rrbracket|}^{*}|,
 \end{split}
 \end{equation}
 where $(a)$ follows from Lemma 
 \ref{lemma:cardinalityAssociation} in   Appendix \ref{sec:AuxResult}, and $(b)$ follows from \eqref{eq:count} and \eqref{eq:NCodebook2}.
 Hence, the statement of the lemma follows.
\end{proof}
\begin{lemma}\label{lemma:commutativeProperty}
Let  
 \begin{equation}
 ( X, Y)\stackrel{d}{\leftrightarrow}(\delta,\delta_{2}).    
 \end{equation}
If $x\stackrel{\delta}{\leftrightsquigarrow} x_{1}$ and  $x\stackrel{\delta}{\leftrightsquigarrow}x_{2}$, then we have that $x_{1}\stackrel{\delta}{\leftrightsquigarrow} x_{2}$.
\end{lemma}
\begin{proof}
 Let $\{\llbracket X|{y}_{i}\rrbracket\}_{i=1}^{{N}}$ be the sequence of conditional range connecting $x$ and $x_1$. Likewise,  let $\{\llbracket X|\tilde{y}_{i}\rrbracket\}_{i=1}^{\tilde{N}}$ be the sequence of conditional range connecting $x$ and $x_2$.

 Now, by Definition \ref{defn:overlapCon}, we have
 \begin{equation}
     x_1\in  \llbracket X|{y}_N\rrbracket,
 \end{equation}
  \begin{equation}
     x_2\in  \llbracket X|\tilde{y}_{\tilde N}\rrbracket,
 \end{equation}
\begin{equation}
    x\in \llbracket X|{y}_1\rrbracket,
\end{equation}
and
\begin{equation}
    x\in \llbracket X|\tilde{y}_1\rrbracket
\end{equation}
Then, using \eqref{ass:finiteness}, we have that
\begin{equation}
    \frac{m_{\mathscr{X}}(\llbracket X|{y}_1\rrbracket\cap \llbracket X|\tilde{y}_1\rrbracket|}{m_{\mathscr{X}}(\llbracket X\rrbracket)}>0,
\end{equation}
 which implies that
  \begin{equation}\label{eq:associationSetIncl}
     \frac{m_{\mathscr{X}}(\llbracket X|{y}_1\rrbracket\cap \llbracket X|\tilde{y}_1\rrbracket|}{m_{\mathscr{X}}(\llbracket X\rrbracket)}\in \mathscr{A}(X;Y).
 \end{equation}

Using the fact that
\begin{equation}\label{eq:dissCondi}
    ( X, Y)\stackrel{d}{\leftrightarrow}(\delta,\delta_{2}),
\end{equation}
 we will now show that 
\begin{equation}\label{eq:sequence23}
    \{\llbracket X|{y}_{{N}}\rrbracket,\llbracket X|{y}_{{N}-1}\rrbracket,\ldots \llbracket X|{y}_{1}\rrbracket,\llbracket X|\tilde{y}_1\rrbracket,\ldots \llbracket X|\tilde{y}_{\tilde{N}}\rrbracket \},
\end{equation} is a sequence of conditional ranges connecting $x_1$ and $x_2$.  Using \eqref{eq:associationSetIncl} and \eqref{eq:dissCondi}, we have that
\begin{equation}\label{eq:condt1}
     \frac{m_{\mathscr{X}}(\llbracket X|{y}_1\rrbracket\cap \llbracket X|\tilde{y}_1\rrbracket|}{m_{\mathscr{X}}(\llbracket X\rrbracket)}>\delta. 
\end{equation}
Also, for all $1<i\leq N$ and $1<j\leq \tilde{N}$, we have
\begin{equation}\label{eq:condt2}
    \frac{m_{\mathscr{X}}(\llbracket X|{y}_i\rrbracket\cap \llbracket X|{y}_{i-1}\rrbracket|}{m_{\mathscr{X}}(\llbracket X\rrbracket)}>\delta,
\end{equation}
and 
\begin{equation}\label{eq:condt3}
    \frac{m_{\mathscr{X}}(\llbracket X|\tilde{y}_j\rrbracket\cap \llbracket X|\tilde{y}_{j-1}\rrbracket|}{m_{\mathscr{X}}(\llbracket X\rrbracket)}>\delta. 
\end{equation}
Also, we have 
\begin{equation}\label{eq:condt4}
    x_1\in \llbracket X|{y}_{{N}}\rrbracket, \mbox{ and } x_2\in \llbracket X|\tilde{y}_{\tilde{N}}\rrbracket.
\end{equation}
Hence, combining \eqref{eq:condt1}, \eqref{eq:condt2}, \eqref{eq:condt3} and \eqref{eq:condt4}, we have that $x_1\stackrel{\delta}{\leftrightsquigarrow}x_2$ via \eqref{eq:sequence23}.  
\end{proof}
\begin{lemma}\label{lemma:RelationToCardinality}
Consider two UVs $X$ and $Y$. Let
\begin{equation}\label{eq:defDeltat}
    \delta^*=\frac{\min_{y\in \llbracket Y\rrbracket}m_{\mathscr{X}}(\llbracket X|y\rrbracket)}{m_{\mathscr{X}}(\llbracket X\rrbracket)}.
\end{equation}
 If $\delta_1<\delta^*$, then we have 
  \begin{equation}
     |\llbracket X|Y\rrbracket_{\delta_1}^*|\leq|\llbracket Y\rrbracket|.
 \end{equation}
\end{lemma}
\begin{proof}
We will prove this by contradiction. Let
 \begin{equation}\label{eq:contrdict122}
     |\llbracket X|Y\rrbracket_{\delta_1}^*|>|\llbracket Y\rrbracket|.
 \end{equation}
 Then, by Property 1 of Definition \ref{defoverlap}, there exists two sets $\mathscr{S}_1,\mathscr{S}_2\in \llbracket X|Y\rrbracket_{\delta_1}^*$ and one singly $\delta_1$-connected set $\llbracket X|y\rrbracket$  such that 
 \begin{equation}\label{eq:reeff1}
     \llbracket X|y\rrbracket\subseteq\mathscr{S}_1, \mbox{ and } \llbracket X|y\rrbracket\subseteq\mathscr{S}_2.
 \end{equation}
 Then, we have 
 \begin{equation}\label{eq:contradict1}
\begin{split}
    \frac{m_{\mathscr{X}}(\mathscr{S}_1\cap \mathscr{S}_2)}{ m_{\mathscr{X}}(\llbracket X\rrbracket)}
    &\stackrel{(a)}{\geq}\frac{m_{\mathscr{X}}(\llbracket X|y\rrbracket)}{ m_{\mathscr{X}}(\llbracket X\rrbracket)},\\
    &\stackrel{(b)}{\geq} \delta^*,\\
    &\stackrel{(c)}{>}\delta_1,
\end{split}
 \end{equation}
 where $(a)$ follows from \eqref{eq:reeff1} and \eqref{eq:strongtransitivity}, $(b)$ follows from \eqref{eq:defDeltat}, and $(c)$ follows from the fact that $\delta_1<\delta^*$. However, by Property 2 of Definition \ref{defoverlap}, we have
 \begin{equation}\label{eq:check34}
    \frac{ m_{\mathscr{X}}(\mathscr{S}_1\cap \mathscr{S}_2)}{ m_{\mathscr{X}}(\llbracket X\rrbracket)}\leq \delta_1.
 \end{equation}
 Hence, we have that \eqref{eq:contradict1} and \eqref{eq:check34} contradict each other, which implies \eqref{eq:contrdict122} does not hold. Hence, the statement of the theorem follows.
\end{proof}

\begin{lemma}\label{lemma:cardinalityAssociation}
Consider two UVs $X$ and $Y$. Let
\begin{equation}\label{eq:neqDelta}
    \delta^*=\frac{\min_{y\in \llbracket Y\rrbracket}m_{\mathscr{X}}(\llbracket X|y\rrbracket)}{m_{\mathscr{X}}(\llbracket X\rrbracket)}.
\end{equation}
 For   all $\delta_1<\delta^*$ and $\delta_2\leq 1$,  if $(X,Y)\stackrel{a}{\leftrightarrow}(\delta_1,\delta_2)$, then  we have
 \begin{equation}
     |\llbracket X|Y\rrbracket_{\delta_1}^*|=|\llbracket Y\rrbracket|.
 \end{equation}
 Additionally, $\llbracket X|Y\rrbracket$ is a $\delta_1$-overlap family.
\end{lemma}
\begin{proof} We  show that 
\begin{equation}
    \llbracket X|Y\rrbracket=\{\llbracket X|y\rrbracket:y\in \llbracket Y\rrbracket\}
\end{equation}
is a  $\delta_1$-overlap family.  First, note that $\llbracket X|Y\rrbracket$ is a cover of $\llbracket X\rrbracket$, since 
$\llbracket X\rrbracket=\cup_{y\in\llbracket Y\rrbracket}\llbracket X|y\rrbracket$. Second, each set in the family $\llbracket X|Y\rrbracket$  is singly $
\delta_1$-connected via $\llbracket X|Y\rrbracket$, since trivially any two points $x_1,x_2 \in \llbracket X|y \rrbracket$ are singly $\delta_1$-connected via the same set.
It follows that Property 1 of Definition~\ref{defoverlap} holds. 

Now, since  $(X, Y) \stackrel{a}{ \leftrightarrow } (\delta_1,\delta_{2})$,  then by Lemma \ref{lemma:Dissassociation} for all $y_1,y_2 \in \llbracket Y\rrbracket$  we have
\begin{equation}
    \frac{m_{\mathscr{X}}(\llbracket X|y_1\rrbracket\cap\llbracket X|y_2\rrbracket)}{m_{\mathscr{X}}(\llbracket X\rrbracket)}\leq \delta_1,
\end{equation}
which shows that Property 2 of Definition~\ref{defoverlap} holds.
Finally, it is also easy to see that Property 3 of Definition~\ref{defoverlap} holds, since $\llbracket X|Y\rrbracket$  contains all   sets $\llbracket X|y\rrbracket$. Hence, $\llbracket X|Y\rrbracket$ satisfies all the properties of $\delta_1$-overlap family, which implies
\begin{equation}\label{eq:LowBoun1}
   |\llbracket X|Y\rrbracket|\leq |\llbracket X|Y\rrbracket^*_{\delta_1}|.
\end{equation}
Since $|\llbracket X|Y\rrbracket|=|\llbracket Y\rrbracket|$, using Lemma \ref{lemma:RelationToCardinality}, we also have 
\begin{equation}\label{eq:UppBoun1}
    |\llbracket X|Y\rrbracket|\geq |\llbracket X|Y\rrbracket^*_{\delta_1}|.
\end{equation}
Combining \eqref{eq:LowBoun1}, \eqref{eq:UppBoun1} and the fact that $\llbracket X|Y\rrbracket$ satisfies all the properties of $\delta_1$-overlap family, the statement of the lemma follows. 
\end{proof}
\begin{lemma}\label{lemma:uniquenessProperty}
Consider two UVs $X$ and $Y$. Let
\begin{equation}\label{eq:neqDelta123}
    \delta^*=\frac{\min_{y\in \llbracket Y\rrbracket}m_{\mathscr{X}}(\llbracket X|y\rrbracket)}{m_{\mathscr{X}}(\llbracket X\rrbracket)}.
\end{equation}
 For   all $\delta_1<\delta^*$ and $\delta_2\leq 1$,  if $(X,Y)\stackrel{a}{\leftrightarrow}(\delta_1/|\llbracket Y\rrbracket|,\delta_2)$, then under Assumption \ref{assumption:triangleInequality}, for all $\llbracket X|y\rrbracket\in \llbracket X|Y\rrbracket$, there exists a point $x\in\llbracket X|y\rrbracket$ such that for all $\llbracket X|y^\prime\rrbracket\in \llbracket X|Y\rrbracket\setminus\{\llbracket X|y\rrbracket\}$,
 \begin{equation}\label{eq:contr0}
     x\notin \llbracket X|y^\prime\rrbracket.
 \end{equation}
\end{lemma}
\begin{proof}
We will prove this by contradiction. Consider a set $\llbracket X|y\rrbracket$. Let $x$ satisfying \eqref{eq:contr0} do not exist. Then, for all $x^\prime\in\llbracket X|y\rrbracket$, there exists a set $\llbracket X|y^\prime\rrbracket\in \llbracket X|Y\rrbracket\setminus\{\llbracket X|y\rrbracket\}$ such that
\begin{equation}\label{eq:contr2}
    x^\prime\in \llbracket X|y^\prime\rrbracket.
\end{equation}
Thus, we have
\begin{equation}\label{eq:contr3}
\begin{split}
&m_{\mathscr{X}}(\cup_{\llbracket X|y^\prime\rrbracket\in \llbracket X|Y\rrbracket\setminus\{\llbracket X|y\rrbracket\}} (\llbracket X|y\rrbracket\cap \llbracket X|y^\prime\rrbracket)),\\
&\stackrel{(a)}{\geq} m_{\mathscr{X}}(\llbracket X|y\rrbracket),\\
&\stackrel{(b)}{\geq} \delta^* m_{\mathscr{X}}(\llbracket X\rrbracket),\\
&\stackrel{(c)}{>}\delta_1 m_{\mathscr{X}}(\llbracket X\rrbracket),
\end{split}
\end{equation}
where $(a)$ follows from \eqref{eq:contr2}, $(b)$ follows from \eqref{eq:neqDelta123}, and $(c)$ follows from the fact that $\delta_1<\delta^*$. On the other hand, since $(X,Y)\stackrel{a}{\leftrightarrow}(\delta_1/|\llbracket Y\rrbracket|,\delta_2)$, we have
 \begin{equation}\label{eq:contr4}
    \begin{split}
      &m_{\mathscr{X}}(\cup_{\llbracket X|y^\prime\rrbracket\in \llbracket X|Y\rrbracket\setminus\{\llbracket X|y\rrbracket\}} (\llbracket X|y\rrbracket\cap \llbracket X|y^\prime\rrbracket)),\\
      &\stackrel{(a)}{\leq} \sum_{\llbracket X|y^\prime\rrbracket\in \llbracket X|Y\rrbracket\setminus\{\llbracket X|y\rrbracket\}} m_{\mathscr{X}}( \llbracket X|y\rrbracket\cap \llbracket X|y^\prime\rrbracket),\\
      &\stackrel{(b)}{\leq} |\llbracket Y\rrbracket|\delta_1 m_{\mathscr{X}}(\llbracket X\rrbracket)/|\llbracket Y\rrbracket|,\\
      &= \delta_1 m_{\mathscr{X}}(\llbracket X\rrbracket),
    \end{split}
 \end{equation}
 where $(a)$ follows from Assumption \ref{assumption:triangleInequality},  $(b)$ follows from Lemma \ref{lemma:Dissassociation}. It follows that \eqref{eq:contr3} and \eqref{eq:contr4} contradict each other, and therefore $x$ satisfying \eqref{eq:contr0} exists. The statement of the lemma follows. 
\end{proof}

\subsection{Proof of $4$ claims in Theorem \ref{thm:sufficientConditiondeltaerror}}\label{sec:5claims}
\noindent
\textit{Proof of Claim 1.} By Property 1 of Definition~\ref{defoverlap}, we have that for all $\mathscr{S}_{i}\in\llbracket {Y}(1:n)|{X}(1:n)\rrbracket_{{\delta}^\prime/|\llbracket {X}(1:n)\rrbracket|} ^{*}$, there exists a $\tilde{x}_{i}(1:n)\in \llbracket X(1:n)\rrbracket$ such that $\llbracket Y(1:n)|\tilde{x}_{i}(1:n)\rrbracket \subseteq \mathscr{S}_{i}$. Now, consider a new UV $\tilde{X}(1:n)$ whose marginal range is composed of  elements of $\llbracket {X}(1:n)\rrbracket$, namely 
\begin{equation}\label{eq:newCodebook1}
    \llbracket \tilde{X}(1:n)\rrbracket=\{\tilde x_1(1:n),\ldots \tilde x_K(1:n)\},
\end{equation}
where
\begin{equation}\label{eq:count23}
    K=|\llbracket {Y}(1:n)|{X}(1:n)\rrbracket_{{\delta}^\prime/|\llbracket {X}(1:n)\rrbracket|} ^{*}|.
\end{equation}
Let $\tilde{Y}(1:n)$ be the UV corresponding to the  received variable. Then, similar to \eqref{eq:existence1}, by Property 2 of Definition~\ref{defoverlap} and since $N$ is stationary memoryless channel,  for all $x(1:n),x^{\prime}(1:n)\in \llbracket\tilde{X}(1:n)\rrbracket$, we have   
\begin{equation}
\begin{split}
      &\frac{m_{\mathscr{Y}}(\llbracket \tilde{Y}(1:n)|x(1:n)\rrbracket\cap \llbracket \tilde{Y}(1:n)|x^{\prime}(1:n)\rrbracket)}{m_{\mathscr{Y}}(\llbracket {Y}(1:n)\rrbracket)}\\
      &\leq \frac{{\delta}^\prime}{|\llbracket X(1:n)\rrbracket|},\\
      &\stackrel{(a)}{\leq}  \frac{{\delta}^\prime}{|\llbracket \tilde{X}(1:n)\rrbracket|},
\end{split}
\end{equation}
where $(a)$ follows from the fact that $\llbracket\tilde{X}(1:n)\rrbracket\subseteq\llbracket X(1:n)\rrbracket$. Similar to \eqref{eq:existence2},  for all $x(1:n),x^{\prime}(1:n)\in \llbracket\tilde{X}(1:n)\rrbracket$, we have that
\begin{equation}
\begin{split}
      &\frac{m_{\mathscr{Y}}(\llbracket \tilde{Y}(1:n)|x(1:n)\rrbracket\cap \llbracket \tilde{Y}(1:n)|x^{\prime}(1:n)\rrbracket)}{m_{\mathscr{Y}}(\llbracket \tilde{Y}(1:n)\rrbracket)}\\
      &\stackrel{}{\leq} 
      \frac{{\delta}^\prime m_{\mathscr{Y}}(\llbracket Y(1:n)\rrbracket)}{|\llbracket \tilde{X}(1:n)\rrbracket| m_{\mathscr{Y}}(\llbracket \tilde{Y}(1:n)\rrbracket)},\\
      &=\frac{\tilde{\delta}}{|\llbracket \tilde{X}(1:n)\rrbracket|},
\end{split}
\end{equation}
where 
\begin{equation}
\begin{split}
     \tilde{\delta}=\frac{\delta^\prime m_{\mathscr{Y}}(\llbracket Y(1:n)\rrbracket)}{m_{\mathscr{Y}}(\llbracket \tilde{Y}(1:n)\rrbracket)}.
\end{split}
\end{equation}
Then, by Lemma~\ref{lemma:Dissassociation} it follows  that 
 \begin{equation}
     (\tilde X(1:n),\tilde{Y}(1:n))\stackrel{a}{\leftrightarrow}(1,\tilde{\delta} /|\llbracket \tilde X(1:n)\rrbracket|). \label{eq:asc1}
 \end{equation} 
Using \eqref{eq:rangeDelta}, we also have
\begin{equation}\label{eq:DelUPP}
    \tilde{\delta}\leq \frac{\delta_n}{m_{\mathscr{Y}}(\llbracket \tilde{Y}(1:n)\rrbracket)}. 
\end{equation}
Additionally, we have 
 \begin{equation}\label{eq:deltaConditionSatis}
 \begin{split}
      \tilde{\delta}&\leq \frac{\delta_n}{m_{\mathscr{Y}}(\llbracket \tilde{Y}(1:n)\rrbracket)}\\
      &\stackrel{(a)}{\leq} \bigg(\frac{\bar{\delta} m_{\mathscr{Y}}(V_N)}{|\llbracket\bar{X}\rrbracket|}\bigg)^n\frac{1}{m_{\mathscr{Y}}(\llbracket \tilde{Y}(1:n)\rrbracket)}\\
      &\stackrel{(b)}{\leq} \frac{(\bar{\delta}  m_{\mathscr{Y}}(V_N))^n}{m_{\mathscr{Y}}(\llbracket \tilde{Y}(1:n)\rrbracket)},\\
      &\stackrel{(c)}{<}\frac{(m_{\mathscr{Y}}(V_N))^n}{m_{\mathscr{Y}}(\llbracket \tilde{Y}(1:n)\rrbracket)},\\
      &\stackrel{(d)}{=}\frac{m_{\mathscr{Y}}(V^n_N)}{m_{\mathscr{Y}}(\llbracket \tilde{Y}(1:n)\rrbracket)},
 \end{split}
 \end{equation}
 where $(a)$ follows from the assumption in the theorem that
 \begin{equation}
     0\leq\delta_n\leq (\bar\delta m_{\mathscr{Y}}(V_N)/|\llbracket \bar X\rrbracket|)^n,
 \end{equation}
 $(b)$ follows from the fact that $|\llbracket\bar{X}\rrbracket|\geq 1$, $(c)$ follows from the fact that using $\delta_1<m_{\mathscr{Y}}(V_N)$, we have
 \begin{equation}
     \bar{\delta}\leq \frac{\delta_1}{m_{\mathscr{Y}}(\llbracket\bar{Y}\rrbracket)}< \frac{m_{\mathscr{Y}}(V_N)}{m_{\mathscr{Y}}(\llbracket\bar{Y}\rrbracket)}\leq 1,
 \end{equation}
 and $d)$ follows from Assumption \ref{assumption:productrule}.
Now, 
we  have  
 \begin{equation}\label{eq:equalityOfFamily}
 \begin{split}
     &|\llbracket \tilde Y(1:n)|\tilde{X}(1:n)\rrbracket^{*}_{\tilde{\delta}/|\llbracket \tilde X(1:n)\rrbracket|}|\\ &\stackrel{(a)}{=}|\llbracket\tilde{X}(1:n)\rrbracket| \\
&     \stackrel{(b)}{=}  | \llbracket Y(1:n)|X(1:n)\rrbracket_{{\tilde{\delta}}/|\llbracket X(1:n)\rrbracket|}^{*}|,
 \end{split}
 \end{equation}
 where $(a)$ follows by combining \eqref{eq:asc1}, \eqref{eq:deltaConditionSatis} and Lemma 
 \ref{lemma:cardinalityAssociation}, and $(b)$ follows from \eqref{eq:newCodebook1} and \eqref{eq:count23}. This along with \eqref{eq:ProveContradiction} implies that we have
 \begin{equation}\label{eq:ProveContradiction_1}
 \begin{split}
       &|\llbracket \tilde Y(1:n)|\tilde{X}(1:n)\rrbracket^{*}_{\tilde{\delta}/|\llbracket \tilde X(1:n)\rrbracket|}|\\
       &>|\prod_{i=1}^n\llbracket \bar Y|\bar X\rrbracket_{\bar\delta/|\llbracket\bar X\rrbracket|}^*|.
 \end{split}
 \end{equation}
This concludes the proof of Claim 1.
\hfill $\square$

\noindent
\textit{Proof of Claim 2}. Since \eqref{eq:asc1} and \eqref{eq:deltaConditionSatis} holds, using Lemma 
 \ref{lemma:cardinalityAssociation}, we have 
\begin{equation}\label{eq:overlapCardinality}
\begin{split}
    &\llbracket \tilde Y(1:n)|\tilde{X}(1:n)\rrbracket^{*}_{\tilde{\delta}/|\llbracket \tilde X(1:n)\rrbracket|} \\
    &= \llbracket \tilde{Y}(1:n)|\tilde{X}(1:n)\rrbracket.
\end{split}
 \end{equation}
Using \eqref{eq:overlapCardinality} and Property 1 of Definition \ref{defoverlap}, we have that for all $\mathscr{S}\in \llbracket \tilde Y(1:n)|\tilde{X}(1:n)\rrbracket^{*}_{\tilde{\delta}/|\llbracket \tilde X(1:n)\rrbracket|}$, there exists a $\tilde{x}(1:n)\in \llbracket\Tilde{X}(1:n)\rrbracket$ such that 
 \begin{equation}\label{eq:FormOfFamily}
     \mathscr{S}=\llbracket\tilde{Y}(1:n)|\tilde{x}(1:n)\rrbracket.
 \end{equation}

Now, for all $x\in\mathscr{X}$, let $\mathscr{S}(x)\in\llbracket\bar{Y}|\bar{X}\rrbracket^*_{\Bar{\delta}/|\llbracket\Bar{X}\rrbracket|}$ be such that
\begin{equation}\label{eq:condition123}
    \llbracket\bar{Y}|x\rrbracket\subseteq \mathscr{S}(x).
\end{equation}
For all $x\in \mathscr{X}\setminus\llbracket \bar{X} \rrbracket$, the set $\mathscr{S}(x)$ exists from the assumption in the theorem. 
Also, for all $x\in \llbracket \bar{X} \rrbracket$, the set $\mathscr{S}(x)$ exists using Property 3 in Definition \ref{defoverlap}. Hence, for all $x\in\mathscr{X}$, we have that $\mathscr{S}(x)$ satisfying \eqref{eq:condition123} exists. 

Hence, for all $\tilde{x}(1:n)\in\llbracket\tilde{X}(1:n)\rrbracket$, we have that 
\begin{equation}\label{eq:subseteq}
\begin{split}
    &\llbracket\tilde Y(1:n)|\tilde x(1:n)\rrbracket\\
    &\stackrel{(a)}{=}\llbracket \tilde Y(1)|
    \tilde x(1)\rrbracket\times \ldots\times\llbracket \tilde Y(n)|\tilde x(n)\rrbracket,\\
    &\stackrel{(b)}{\subseteq} \mathscr{S}(x(1))\times \ldots\times \mathscr{S}(x(n)),\\
    &\stackrel{(c)}{\in} \prod_{i=1}^n\llbracket \bar{Y}|\bar{X}\rrbracket_{\bar\delta/|\llbracket \bar X\rrbracket|}^*,
\end{split}
\end{equation}
where $(a)$ follows from the fact that $N$ is a stationary memoryless uncertain channel, $(b)$ follows from the fact that for all $x\in\mathscr{X}$,  $\mathscr{S}(x)$ exists, and $(c)$ follows from the fact that for all $x\in\mathscr{X}$,
$\mathscr{S}(x)\in\llbracket\bar{Y}|\bar{X}\rrbracket^*_{\Bar{\delta}/|\llbracket\Bar{X}\rrbracket|}$.  Hence, {Claim 2} is proved.
\hfill $\square$

\noindent
\textit{Proof of Claim 3}.
Combining \eqref{eq:ProveContradiction_1} and \eqref{eq:overlapCardinality}, we have that
\begin{equation}
    |\llbracket\tilde{Y}(1:n)|\tilde{X}(1:n)\rrbracket|>\prod_{i=1}^n\llbracket \bar{Y}|\bar{X}\rrbracket_{\bar\delta/|\llbracket \bar{X} \rrbracket|}^*.
\end{equation}
This  along with  \eqref{eq:subseteq} implies that  there exists a set $\mathscr S\in \prod_{i=1}^n\llbracket \bar{Y}|\bar{X}\rrbracket_{\bar\delta/|\llbracket \bar{X} \rrbracket|}^*$ which contains at least two sets $\mathscr D_1, \mathscr D_2 \in \llbracket \tilde{Y}(1:n)|\tilde{X}(1:n)\rrbracket^{*} _{\tilde{\delta}/|\llbracket \tilde X(1:n)\rrbracket|}$, namely
\begin{equation}\label{eq:subset2}
    \mathscr D_1\subset \mathscr{S},
\end{equation}
\begin{equation}\label{eq:subset3}
    \mathscr D_2\subset \mathscr{S}. 
\end{equation}
Using \eqref{eq:FormOfFamily}, without loss of generality, let 
\begin{equation}\label{eq:singlySet1}
    \mathscr D_1=\llbracket \Tilde{Y}(1:n)|\tilde{x}_1(1:n)\rrbracket,
\end{equation}
\begin{equation}\label{eq:singlySet2}
    \mathscr D_2=\llbracket \Tilde{Y}(1:n)|\tilde{x}_2(1:n)\rrbracket. 
\end{equation}
Also, let 
\begin{equation}\label{eq:cartesianProdSets}
    \mathscr S=\mathscr S_1\times\ldots \times \mathscr S_n,
\end{equation}
where $\mathscr S_1,\ldots, \mathscr S_n \in \llbracket \bar{Y}|\bar{X}\rrbracket_{\bar\delta/|\llbracket \bar{X} \rrbracket|}^*$.
Also, we have
\begin{equation} \label{eq:RangeDel}
    \frac{\bar \delta}{|\llbracket\bar{X}\rrbracket|}\leq \bar{\delta}\leq \frac{\delta_1}{m_{\mathscr{Y}}(\llbracket\bar Y\rrbracket)}<\frac{m_{\mathscr{Y}}(V_N )}{m_{\mathscr{Y}}(\llbracket\bar Y\rrbracket)}.
\end{equation}
{
Now, we have
\begin{equation}\label{eq:lowerBoundConnectivity}
\begin{split}
     \frac{\tilde\delta}{|\llbracket \tilde X(1:n)\rrbracket|}&\leq \frac{{\delta}_n}{m_{\mathscr{Y}}(\llbracket\tilde{Y}(1:n)\rrbracket)},\\
     &\stackrel{(a)}{\leq} \bigg(\frac{\bar\delta m_{\mathscr{Y}}(V_N)}{|\llbracket \bar X\rrbracket|}\bigg)^n\frac{1}{m_{\mathscr{Y}}(\llbracket\tilde{Y}(1:n)\rrbracket)},\\
     &\stackrel{(b)}{\leq} \bigg(\frac{\bar\delta}{|\llbracket \bar X\rrbracket|}\bigg)^n,
\end{split}
\end{equation} 
where $(a)$ follows from the assumption in the theorem that 
\begin{equation}
    \delta_n\leq \bigg(\frac{\bar\delta m_{\mathscr{Y}}(V_N)}{|\llbracket \bar X\rrbracket|}\bigg)^n, 
\end{equation}
and $(b)$ follows from the fact that using Assumption \ref{assumption:productrule}, we have
\begin{equation}
    m_{\mathscr{Y}}(V_{N}^n)= (m_{\mathscr{Y}}(V_{N}))^n\leq m_{\mathscr{Y}}(\llbracket\tilde{Y}(1:n)\rrbracket).
\end{equation}
Combining  Lemma \ref{lemma:Dissassociation} and \eqref{eq:asc1}, we have
\begin{equation}\label{eq:maxOverlap1}
\begin{split}
     &\frac{m_{\mathscr{Y}}(\llbracket \tilde{Y}(1:n)|\tilde{x}_1(1:n)\rrbracket \cap \llbracket \tilde{Y}(1:n)|\tilde{x}_2(1:n)\rrbracket)}{m_{\mathscr{Y}}(\llbracket\tilde{Y}(1:n)\rrbracket)}\\ 
     &\leq \frac{\tilde\delta}{|\llbracket \tilde X(1:n)\rrbracket|}\\
     &\stackrel{(a)}{\leq} \bigg(\frac{\bar{\delta}}{|\llbracket\bar{X}\rrbracket|}\bigg)^n,
\end{split}
\end{equation}
where $(a)$ follows from   \eqref{eq:lowerBoundConnectivity}.
This implies that there exists a $1\leq i^*\leq n$ such that 
\begin{equation}\label{eq:boundOverlap}
    \frac{m_{\mathscr{Y}}(\llbracket \tilde{Y}(i^*)|\tilde{x}_1(i^*)\rrbracket \cap \llbracket \tilde{Y}(i^*)|\tilde{x}_2(i^*)\rrbracket)}{(m_{\mathscr{Y}}(\llbracket\tilde{Y}(1:n)\rrbracket))^{1/n}}\leq \frac{\bar{\delta}}{|\llbracket\bar{X}\rrbracket|},
\end{equation}
otherwise \eqref{eq:maxOverlap1} does not hold, namely
\begin{equation}
\begin{split}
&\frac{m_{\mathscr{Y}}(\llbracket \tilde{Y}(1:n)|\tilde{x}_1(1:n)\rrbracket \cap \llbracket \tilde{Y}(1:n)|\tilde{x}_2(1:n)\rrbracket)}{m_{\mathscr{Y}}(\llbracket\tilde{Y}(1:n)\rrbracket)}\\
&\stackrel{(a)}{=}  \prod_{i=1}^n\bigg(\frac{m_{\mathscr{Y}}(\llbracket \tilde{Y}(i)|\tilde{x}_1(i)\rrbracket \cap \llbracket \tilde{Y}(i)|\tilde{x}_2(i)\rrbracket)}{(m_{\mathscr{Y}}(\llbracket\tilde{Y}(1:n)\rrbracket))^{1/n}} \bigg),\\
&\stackrel{(b)}{>} \bigg(\frac{\bar{\delta}}{|\llbracket\bar{X}\rrbracket|}\bigg)^n,
\end{split}
\end{equation}
where $(a)$ follows from Assumption \ref{assumption:productrule} and the fact that $N$ is stationary memoryless, $(b)$ follows from the hypothesis that $i^*$ satisfying \eqref{eq:boundOverlap} does not exist.
\hfill $\square$
}

\noindent
\textit{Proof of Claim 4.}  
Now, consider a UV   $X^\prime$ such that
\begin{equation}\label{eq:XprimeDef}
\begin{split}
    \llbracket X^\prime\rrbracket&=(\llbracket\bar{X}\rrbracket\setminus \{x\in \mathscr{X}:\llbracket Y|x\rrbracket\subseteq \mathscr S_{i^*}\} )\\
     &\qquad\cup \{\tilde{x}_1(i^*)\} \cup\{\tilde{x}_2(i^*)\}. 
\end{split}
\end{equation}
For $\llbracket X^{\prime}_1\rrbracket= (\llbracket \bar{X}\rrbracket\setminus \{x \in \mathscr{X}:\llbracket Y|x\rrbracket\subseteq \mathscr S_{i^*}\} )$, the 
$\delta^\prime_1$- overlap family of $\llbracket Y^\prime_1|X^\prime_1\rrbracket$ satisfies 
\begin{equation}
   |\llbracket \bar Y|\bar X\rrbracket^*_{\bar\delta/|\llbracket \bar X\rrbracket
|}|-1\stackrel{
(a)}{\leq} |\llbracket Y^{\prime}_1|X^{\prime}_1\rrbracket^*_{\delta^\prime_1}|,
\end{equation}
where 
\begin{equation}
    \delta^\prime_1=(\bar\delta m_{\mathscr{Y}}(\llbracket\bar{Y}\rrbracket))/(|\llbracket \bar X\rrbracket
| m_{\mathscr{Y}}(\llbracket{Y}^\prime_1\rrbracket),
\end{equation}
and 
$(a)$ follows from the fact that
\begin{equation}\label{eq:SetOneOne}
    \mathcal{S}_1=\{\mathscr{S}^\prime\in \llbracket \bar Y|\bar X\rrbracket^*_{\bar\delta/|\llbracket \bar X\rrbracket|}:\mathscr{S}^\prime\neq \mathscr{S}_{i^*}\}
\end{equation} 
satisfies all the properties of $\llbracket Y_1^\prime|X_1^\prime\rrbracket^*_{\delta^\prime_1}$ in Definition \ref{defoverlap}.

Now,  consider the UV $X^\prime$ such that $\llbracket X^\prime\rrbracket=\llbracket X^\prime_1\rrbracket\cup  \{\tilde{x}_1(i^*)\} \cup\{\tilde{x}_2(i^*)\}$. We will show that 
\begin{equation}\label{eq:S_3Def}
\begin{split}
    \mathcal{S}_3
    &=\mathcal{S}_1\cup\{\llbracket \tilde{Y}(i^*)|\tilde{x}_1(i^*)\rrbracket\}\cup\{\llbracket \tilde{Y}(i^*)|\tilde{x}_2(i^*)\rrbracket\}.
\end{split}
\end{equation}
satisfies the property of $\llbracket Y^\prime |X^\prime\rrbracket^*_{\delta^*/|\llbracket\Bar{X}^\prime\rrbracket|}$, where
\begin{equation}
    \delta^*= \frac{\Bar{\delta} |\llbracket {X}^\prime\rrbracket| m_{\mathscr{Y}}(\llbracket\bar{Y}\rrbracket)}{|\llbracket\Bar{X}\rrbracket| m_{\mathscr{Y}}(\llbracket{Y}^\prime\rrbracket)}.
\end{equation}
Using \eqref{eq:subset2}, \eqref{eq:subset3} and Claim 2, we have 
\begin{equation}\label{eq:SetContain1}
    \llbracket \tilde{Y}(i^*)|\tilde{x}_1(i^*)\rrbracket, \llbracket \tilde{Y}(i^*)|\tilde{x}_2(i^*)\rrbracket\subseteq\mathscr{S}_{i^*}.
\end{equation}
This along with the fact that $\llbracket \bar Y|\Bar{X}\rrbracket_{\Bar{\delta}/|\llbracket \Bar{X}\rrbracket|}$ is an overlap family implies that for all $\mathscr{S}^\prime\in \mathcal{S}_1$, 
\begin{equation}\label{eq:S1_2}
    m_{\mathscr{Y}}(\llbracket \tilde{Y}(i^*)|\tilde{x}_1(i^*)\rrbracket\cap \mathscr{S}^\prime)\leq \frac{\Bar{\delta} m_{\mathscr{Y}}(\llbracket\Bar{Y}\rrbracket)}{|\llbracket\bar X\rrbracket|},
\end{equation}
and 
\begin{equation}\label{eq:S1_3}
    m_{\mathscr{Y}}(\llbracket \tilde{Y}(i^*)|\tilde{x}_2(i^*)\rrbracket\cap \mathscr{S}^\prime)\leq \frac{\Bar{\delta} m_{\mathscr{Y}}(\llbracket\Bar{Y}\rrbracket)}{|\llbracket\bar X\rrbracket|}.
\end{equation}
Also,  we have that
\begin{equation}\label{eq:S1_4}
\begin{split}
 & m_{\mathscr{Y}}(\llbracket \tilde{Y}(i^*)|\tilde{x}_1(i^*)\rrbracket\cap\llbracket \tilde{Y}(i^*)|\tilde{x}_2(i^*)\rrbracket)\\
 &\stackrel{(a)}{\leq} \frac{\bar{\delta}(m_{\mathscr{Y}}(\llbracket \tilde{Y}(1:n)\rrbracket))^{1/n}}{|\llbracket\bar{X}\rrbracket|}\\
 &\stackrel{(b)}{\leq} \frac{\Bar{\delta} (m_{\mathscr{Y}}(\llbracket\Bar{Y}(1:n)\rrbracket))^{1/n}}{|\llbracket\bar X\rrbracket|}\\
 &\stackrel{(c)}{=} \frac{\Bar{\delta}   m_{\mathscr{Y}}(\llbracket\Bar{Y}\rrbracket)}{|\llbracket\bar X\rrbracket|}, 
\end{split}
\end{equation}
where $(a)$ follows from \eqref{eq:boundOverlap},  
$(b)$ follows from \eqref{eq:strongtransitivity} and  $\llbracket\tilde{Y}(1:n)\rrbracket\subseteq\llbracket Y(1:n)\rrbracket$ by Claim 2, and $(c)$ follows from Assumption \ref{assumption:productrule} and \eqref{eq:Def2}.  Additionally, $\llbracket \tilde{Y}(i^*)|\tilde{x}_1(i)\rrbracket$ and $\llbracket \tilde{Y}(i^*)|\tilde{x}_2(i)\rrbracket$ are singly ${\delta}^*/|\llbracket
\Bar{X}^\prime\rrbracket|$ connected sets. This along with \eqref{eq:S1_2}, \eqref{eq:S1_3} and \eqref{eq:S1_4} implies that $\mathcal{S}_3$ satisfies all the properties of $\llbracket Y^\prime|X^\prime\rrbracket^*_{\delta^*/|\llbracket X^\prime\rrbracket|}$. 
It follows that
\begin{equation} \label{eq:lowerBoundCheck1}
\begin{split}
     |\llbracket Y^\prime|X^\prime\rrbracket^*_{\delta^*/|\llbracket X^\prime\rrbracket|}|&\geq |\mathcal{S}_3|,\\
     &\stackrel{(a)}{=}|\mathcal{S}_1|+2,\\
     &\stackrel{(b)}{=} |\llbracket\bar{Y}|\bar{X}\rrbracket^*_{\bar{\delta}/|\llbracket\bar{X}\rrbracket|}|+1,
\end{split}
\end{equation}
where $(a)$ follows from  \eqref{eq:S_3Def}, and $(b)$ follows from
\eqref{eq:SetOneOne}. 

Now, we will show that
\begin{equation}\label{eq:Bound3}
    |\llbracket X^\prime\rrbracket|\leq |\llbracket \bar{X}\rrbracket|+1.
\end{equation}
We split the analysis into two mutually exclusive cases: $\tilde{x}_1(i^*)\in \llbracket\bar{X}\rrbracket$ or 
$\tilde{x}_2(i^*)\in \llbracket\bar{X}\rrbracket$; and $\tilde{x}_1(i^*), \tilde{x}_2(i^*) \notin \llbracket\bar{X}\rrbracket$. 
In the first case, if $\tilde{x}_1(i^*)\in \llbracket\bar{X}\rrbracket$ or 
$\tilde{x}_2(i^*)\in \llbracket\bar{X}\rrbracket$, then using \eqref{eq:XprimeDef}, we have
\begin{equation}\label{eq:Bound1}
    |\llbracket X^\prime\rrbracket|\leq |\llbracket \bar{X}\rrbracket|+1.
\end{equation}
In the second case, if $\tilde{x}_1(i^*), \tilde{x}_2(i^*) \notin \llbracket\bar{X}\rrbracket$, then using \eqref{eq:SetContain1}, there exists a non-empty set $\mathscr{P}\subseteq \llbracket \bar{X}\rrbracket$ such that 
\begin{equation}\label{eq:claim1}
    \llbracket \tilde{Y}(i^*)|\tilde{x}_1(i^*)\rrbracket\cup \llbracket \tilde{Y}(i^*)|\tilde{x}_2(i^*)\rrbracket\subseteq \cup_{x\in\mathscr{P}} \llbracket \bar{Y}|x\rrbracket.
\end{equation}
Also, there exists a $x^\prime\in\mathcal{P}$ such that
\begin{equation}\label{eq:xprimeExist}
     \llbracket \bar{Y}|x^\prime\rrbracket\subseteq \mathscr{S}_{i^*}.
\end{equation}
This can be proved by contradiction. Let $x^\prime\in\mathcal{P}$ satisfying \eqref{eq:xprimeExist} does not exist. We have \begin{equation}\label{eq:contradiction1}
\begin{split}
     &m_{\mathscr{Y}}(\cup_{\mathscr{{S}}^\prime\in \llbracket \bar{Y}|\bar{X}\rrbracket_{\bar\delta/|\llbracket \bar{X} \rrbracket|}^*:\mathscr{{S}}^\prime\neq \mathscr{S}_{i^*}}(\mathscr{S}_{i^*}\cap \mathscr{{S}}^\prime))\\
     &\stackrel{(a)}{\geq} m_{\mathscr{Y}}(\cup_{x: x\in\mathcal{P}}(\mathscr{S}_{i^*} \cap \llbracket \bar{Y}|{x}\rrbracket)),\\
     &\stackrel{(b)}{\geq} m_{\mathscr{Y}}(\llbracket \tilde{Y}(i^*)|\tilde{x}_1(i^*)\rrbracket\cup \llbracket \tilde{Y}(i^*)|\tilde{x}_2(i^*)\rrbracket),\\
     &\stackrel{(c)}{\geq} m_{\mathscr{Y}}(V_N)\\
     &\stackrel{(d)}{>}\delta_1\\
     &\stackrel{(e)}{\geq} \bar\delta  m_{\mathscr{Y}}(\llbracket \bar Y\rrbracket),\\
\end{split}
\end{equation}
where $(a)$ follows from the fact that combining $\mathcal{P}\subseteq\llbracket X\rrbracket$, Property 3 of Definition \ref{defoverlap}, and the hypothesis that $x^\prime$ does not exist, we have
\begin{equation}
     \cup_{x\in\mathscr{P}} \llbracket \bar{Y}|x\rrbracket\subseteq \mathscr{{S}}^\prime\in \cup_{\llbracket \bar{Y}|\bar{X}\rrbracket_{\bar\delta/|\llbracket \bar{X} \rrbracket|}^*:\mathscr{{S}}^\prime\neq \mathscr{S}_{i^*}} \mathscr{{S}}^\prime,
\end{equation}
$(b)$ follows from \eqref{eq:SetContain1} and \eqref{eq:claim1}, 
$(c)$ follows from \eqref{eq:strongtransitivity} and the fact that for all $x\in\mathscr{X}$,
\begin{equation}
    m_{\mathscr{Y}}(V_N)\leq m_{\mathscr{Y}}(\llbracket Y|x\rrbracket),
\end{equation}
 $(d)$ follows from the fact that $\delta_1<m_{\mathscr{Y}}(V_N)$, and $(e)$ follows from the fact that $\Bar{\delta}\leq \delta_1/m_{\mathscr{Y}}(\llbracket\bar{Y}\rrbracket)$. On the other hand, since $\llbracket\Bar{Y}|\bar{X}\rrbracket^*_{\bar{\delta}/|\llbracket\Bar{X}\rrbracket|}$ is an overlap family,  we have
\begin{equation}\label{eq:contradiction2}
\begin{split}
&m_{\mathscr{Y}}(\cup_{\mathscr{{S}}^\prime\in \llbracket \bar{Y}|\bar{X}\rrbracket_{\bar\delta/|\llbracket \bar{X} \rrbracket|}^*:\mathscr{{S}}^\prime\neq \mathscr{S}_{i^*}}(\mathscr{S}_{i^*}\cap \mathscr{{S}}^\prime)),\\
&\stackrel{(a)}{\leq}\sum_{\mathscr{{S}}^\prime\in \llbracket \bar{Y}|\bar{X}\rrbracket_{\bar\delta/|\llbracket \bar{X} \rrbracket|}^*:\mathscr{{S}}^\prime\neq \mathscr{S}_{i^*}} m_{\mathscr{Y}}(\mathscr{S}_{i^*}\cap \mathscr{{S}}^\prime),\\
&\stackrel{(b)}{\leq}\frac{\Bar{\delta}  |\llbracket \Bar{Y}|\bar{X}\rrbracket^*_{\Bar{\delta}/|\llbracket\bar{X}\rrbracket|}| m_{\mathscr{Y}}(\llbracket\bar Y\rrbracket)}{|\llbracket\Bar{X}\rrbracket|},\\
&\stackrel{(c)}{\leq} \Bar{\delta} m_{\mathscr{Y}}(\llbracket\bar Y\rrbracket)
\end{split}
\end{equation}
where $(a)$ follows from Assumption \ref{assumption:triangleInequality}, $(b)$ follows from Property 2 of Definition \ref{defoverlap}, and $(c)$ follows from the fact that using \eqref{eq:RangeDel}, Lemma \ref{lemma:RelationToCardinality} holds. Hence, \eqref{eq:contradiction1} and \eqref{eq:contradiction2} contradict each other, which implies $x^\prime$ satisfying \eqref{eq:xprimeExist} exists.
Now, using \eqref{eq:xprimeExist} and \eqref{eq:XprimeDef} ,  we have that 
\begin{equation}\label{eq:Bound2}
    |\llbracket X^\prime\rrbracket|\leq |\llbracket \bar{X}\rrbracket|+1.
\end{equation}
Hence, \eqref{eq:Bound3} holds. 

Finally, we have
\begin{equation}\label{eq:deltaBoundUpper}
\begin{split}
\delta^*&=\frac{\Bar{\delta}|\llbracket{X}^\prime\rrbracket| m_{\mathscr{Y}}(\llbracket\bar{Y}\rrbracket)}{|\llbracket\Bar{X}\rrbracket| m_{\mathscr{Y}}(\llbracket{Y}^\prime\rrbracket)},\\
&\stackrel{(a)}{\leq }\frac{\Bar{\delta}  m_{\mathscr{Y}}(\llbracket\bar{Y}\rrbracket)}{ m_{\mathscr{Y}}(\llbracket{Y}^\prime\rrbracket)}\bigg(1+\frac{1}{|\llbracket\bar{X}\rrbracket|}\bigg),\\
&\stackrel{(b)}{\leq} \frac{\delta_1}{ m_{\mathscr{Y}}(\llbracket{Y}^\prime\rrbracket)},
\end{split}
\end{equation}
where $(a)$ follows from \eqref{eq:Bound3}, and $(b)$ follows from the assumption in the theorem that $\Bar{\delta}(1+1/|\llbracket\bar{X}\rrbracket|)\leq \delta_1/m_{\mathscr{Y}}(\llbracket\bar{Y}\rrbracket)$. 
Now, using \eqref{eq:lowerBoundCheck1} and \eqref{eq:deltaBoundUpper}, we have that there exists a $\delta^*\leq \delta_1/m_{\mathscr{Y}}(\llbracket{Y}^\prime\rrbracket)$ such that
\begin{equation}
    |\llbracket Y^\prime|X^\prime\rrbracket^*_{\delta^*/|\llbracket X^\prime\rrbracket|}|>|\llbracket\bar{Y}|\bar{X}\rrbracket^*_{\bar{\delta}/|\llbracket\bar{X}\rrbracket|}|.
\end{equation}
This concludes the proof of Claim 4. 
\hfill $\square$

\subsection{Taxicab symmetry of the mutual information}\label{sec:taxicabSection}
\begin{definition}\label{defn:taxicanConn}{$(\delta_1,\delta_2)$-taxicab connectedness and $(\delta_1,\delta_2)$-taxicab isolation.}
\begin{itemize}
\item Points $(x,y),(x^\prime,y^\prime)\in \llbracket X, Y\rrbracket$ are $(\delta_{1},\delta_{2})$-taxicab connected via $\llbracket X,Y\rrbracket$, and are denoted by $(x,y)\stackrel{\delta_{1},\delta_{2}}{\leftrightsquigarrow} (x^\prime,y^\prime)$, if there exists a finite  sequence $\{(x_{i},y_{i})\}_{i=1}^{N}$ of points in $\llbracket X,Y \rrbracket$  such that $(x,y)=(x_1,y_1)$, $(x^\prime,y^\prime)=(x_N,y_N)$ and for all $2<i\leq N$, we have either
\begin{equation*}\label{eq:condition1}
\begin{split}
       A_{1}&=\{x_{i}=x_{i-1} \, \mbox{\emph{and} } \, \frac{m_{\mathscr{X}}(\llbracket X|y_{i}\rrbracket\cap \llbracket X|y_{i-1}\rrbracket)}{m_{\mathscr{X}}(\llbracket X\rrbracket)}>\delta_{1}\},
\end{split}
\end{equation*}
or
\begin{equation*}\label{eq:condition2}
\begin{split}
       A_{2}&=\{y_{i}=y_{i-1}  \, \mbox{\emph{and} } \,\frac{m_{\mathscr{Y}}(\llbracket Y|x_{i}\rrbracket\cap \llbracket Y|x_{i-1}\rrbracket)}{m_{\mathscr{Y}}(\llbracket Y\rrbracket)}>\delta_{2}\}.
\end{split}
\end{equation*}
If $(x,y)\stackrel{\delta_{1},\delta_{2}}{\leftrightsquigarrow} (x^\prime,y^\prime)$ and  $N=2$, then we say that $(x,y)$ and $(x^\prime,y^\prime)$ are  singly $(\delta_{1},\delta_{2})$-taxicab connected, i.e. either $y=y^{\prime}$ and $x,x^\prime\in \llbracket X|y\rrbracket$ or $x=x^{\prime}$ and $y,y^\prime\in \llbracket Y|x\rrbracket$.
 \item A set $\mathscr{S}\subseteq\llbracket X,Y\rrbracket$ is (singly) $(\delta_{1},\delta_{2})$-taxicab connected via $\llbracket X,Y\rrbracket$ if every pair of points in the set is (singly) $(\delta_{1},\delta_{2})$-taxicab connected in $\llbracket X, Y\rrbracket$. 
\item Two sets $\mathscr{S}_1,\mathscr{S}_2\subseteq\llbracket X,Y\rrbracket$ are $(\delta_{1},\delta_{2})$-taxicab isolated via $\llbracket X,Y\rrbracket$ if no point in $\mathscr{S}_1$ is  $(\delta_{1},\delta_{2})$-taxicab connected to any point in $\mathscr{S}_2$.
\end{itemize}
\end{definition}
\begin{definition}
 Projection of a set\\
 \begin{itemize}
     \item  The projection $\mathscr{S}^+_x$ of a set $\mathscr{S}\subseteq\llbracket X,Y\rrbracket$ on the $x$-axis is defined as
     \begin{equation}
    \mathscr{S}^+_x=\{x:(x,y)\in\mathscr{S}\}. 
\end{equation}
\item The projection $\mathscr{S}^+_y$ of a set $\mathscr{S}\subseteq\llbracket X,Y\rrbracket$ on the $y$-axis is defined as
     \begin{equation}
    \mathscr{S}^+_y=\{y:(x,y)\in\mathscr{S}\}. 
\end{equation}
 \end{itemize}
\end{definition}

 \begin{definition}{$(\delta_1,\delta_2)$-taxicab family}\label{def:taxicabFamily}\\
 A   $(\delta_{1},\delta_{2})$-taxicab family of $\llbracket X,Y\rrbracket$, denoted by $\llbracket X,Y\rrbracket^*_{(\delta_1,\delta_2)}$, is a largest family of distinct sets covering $\llbracket X,Y\rrbracket$ such that:
 \begin{enumerate}
     \item Each set in the family is  $(\delta_{1},\delta_{2})$-taxicab connected and contains at least one  singly $\delta_1$-connected set of form $\llbracket X|y\rrbracket\times \{y\}$, and at least one singly $\delta_2$-connected set of the form $\llbracket Y|x\rrbracket\times \{x\}$.
     \item The measure of overlap between the projections on the $x$-axis and $y$-axis of any two distinct sets in the family are at most $\delta_{1} m_{\mathscr{X}}(\llbracket X\rrbracket)$ and $\delta_{2} m_{\mathscr{Y}}(\llbracket Y\rrbracket)$ respectively.
     \item For every singly $(\delta_1,\delta_2)$-connected set, there exists a set in the family containing it. 
 \end{enumerate}
 \end{definition}
We now show that when   $( X, Y)\stackrel{d}{\leftrightarrow}(\delta_{1},\delta_{2})$ hold, the cardinality of $(\delta_1,\delta_2)$-taxicab family is same as the cardinality of the $\llbracket X|Y\rrbracket$ $\delta_1$-overlap family and  $\llbracket Y|X\rrbracket$ $\delta_2$-overlap family. 
 
\noindent
\\
{\textit{Proof of Theorem \ref{corr:symmetricPropOfInfo}}}

\noindent
We will show that $|\llbracket X,Y\rrbracket_{(\delta_1,\delta_2)}^*|=|\llbracket X|Y\rrbracket^*_{\delta_1}|$. Then, $|\llbracket X,Y\rrbracket_{(\delta_1,\delta_2)}^*|=|\llbracket Y|X\rrbracket^*_{\delta_2}|$ can be derived along the same lines. Hence, the statement of the theorem follows. 

First, we will show that 
\begin{equation}
    \mathcal{D}=\{\mathscr{S}^+_x: \mathscr{S}\in\llbracket X,Y\rrbracket_{(\delta_1,\delta_2)}^*\},
\end{equation}
satisfies all the properties of  $\llbracket X|Y\rrbracket^*_{\delta_1}$.

Since $\llbracket X,Y\rrbracket_{(\delta_1,\delta_2)}^*$ is a covering of $\llbracket X,Y\rrbracket$, we have
\begin{equation}\label{eq:CoveringofX}
    \cup_{\mathscr{S}^+_x\in\mathcal{D}}\mathscr{S}^+_x=\llbracket X\rrbracket,
\end{equation}
which implies $\mathcal{D}$ is a covering of $\llbracket X\rrbracket$.

Consider a set $\mathscr{S}\in \llbracket X,Y\rrbracket_{(\delta_1,\delta_2)}^*$. For all $(x,y),(x^\prime,y^\prime)\in\mathscr{S}$, $(x,y)$ and $(x^\prime,y^\prime)$ are $(\delta_{1},\delta_{2})$-taxicab connected. Then, there exists a taxicab sequence of the form
\[(x,y),(x_{1},y),(x_{1},y_{1}),\ldots (x_{n-1},y^{\prime}),(x^\prime,y^\prime).\]
such that either $A_{1}$ or $A_{2}$ in Definition \ref{defn:taxicanConn} is true. 
Then, the sequence $\{y,y_1,\ldots,y_{n-1},y^\prime\}$ yields a sequence of conditional range $\{\llbracket X|\tilde{y}_{j}\rrbracket\}_{j=1}^{n+1}$ such that for all $1<j\leq n+1$,
\begin{equation}
    \frac{m_{\mathscr{X}}(\llbracket X|\tilde y_{j}\rrbracket\cap \llbracket X|\tilde y_{j-1}\rrbracket)}{m_{\mathscr{X}}(\llbracket X\rrbracket)}>\delta_{1},
\end{equation}
\begin{equation}
    x\in \llbracket X|\tilde{y}_1\rrbracket,\mbox{ and } x^\prime\in \llbracket X|\tilde{y}_{n+1}\rrbracket.
\end{equation}
 Hence, $x\stackrel{\delta_1}{\leftrightsquigarrow}x^\prime$  via $\llbracket X|Y\rrbracket$.
 Hence,  $\mathscr{S}^+_x$
 is $\delta_{1}$-connected via $\llbracket X|Y\rrbracket$. Also, $\mathscr{S}$ contains at least one singly $\delta_1$-connected set of the form $\llbracket X|y\rrbracket\times\{y\}$, which implies $\llbracket X|y\rrbracket\subseteq\mathscr{S}^+_x$. Hence, $\mathscr{S}^+_x$ contains at least one singly $\delta_1$-connected set of the form $\llbracket X|y\rrbracket$. Hence, $\mathcal{D}$ satisfies Property 1 in Definition \ref{defoverlap}.
 
 For all $\mathscr{S}_1,\mathscr{S}_2\in\llbracket X,Y\rrbracket_{(\delta_1,\delta_2)}^*$, we have 
 \begin{equation}
     m_{\mathscr{X}}(\mathscr{S}_{1,x}^+\cap \mathscr{S}_{2,x}^+)\leq \delta_1 m_{\mathscr{X}}(\llbracket X\rrbracket),
 \end{equation}
 using Property 2 in Definition \ref{def:taxicabFamily}. Hence, $\mathcal{D}$ satisfies Property 2 in Definition \ref{defoverlap}. 
 
Using Property 3 in Definition \ref{def:taxicabFamily}, we have that for all $\llbracket X|y\rrbracket\times\{y\}$, there exists a set $\mathscr{S}(y)\in \llbracket X,Y\rrbracket_{(\delta_1,\delta_2)}^*$ containing it. This implies that for all $\llbracket X|y\rrbracket\in \llbracket X|Y\rrbracket$, we have 
\begin{equation}
    \llbracket X|y\rrbracket\subseteq \mathscr{S}(y)^+_x. 
\end{equation}
Hence, $\mathcal{D}$ satisfies Property 3 in Definition \ref{defoverlap}.

Thus, $\mathcal{D}$ satisfies all the three properties of  $\llbracket X|Y\rrbracket_{\delta_1}^{*}$. This implies along with Theorem \ref{lemma:overlap} that 
\begin{equation}\label{eq:uppBound1}
    |\mathcal{D}|= |\llbracket X|Y\rrbracket_{\delta_1}^*|,
\end{equation}
which implies
\begin{equation}
    |\llbracket X,Y\rrbracket^*_{(\delta_1,\delta_2)}|=|\llbracket X|Y\rrbracket_{\delta_1}^*|. 
\end{equation}
Hence, the statement of the theorem follows.
\hfill $\square$
\end{document}